\documentclass[acmsmall,10pt]{acmart}

\usepackage{booktabs}   
\usepackage{subcaption} 
\usepackage{stmaryrd}
\usepackage{bussproofs}
\usepackage{multirow}

\setcopyright{none}

\usepackage{listings}
\usepackage{amsmath}
\usepackage{amsthm,amssymb,amsfonts,amsxtra}
\usepackage{mathrsfs}
\usepackage{verbatim}
\usepackage{algorithmic}
\usepackage{graphicx}
\usepackage{xcolor}
\usepackage{enumitem}
\usepackage{listings}
\usepackage{url}
\usepackage{semantic}
\usepackage{array}

\newcommand*{\xzw}[1]{}

\newcommand*{\BN}{BN}

\newcommand*{\termEmph}[1]{\emph{#1}}
\newcommand*{\Tcal}{\mathcal{T}}
\newcommand*{\PS}{\ensuremath{\mathbf{P}}}
\newcommand*{\PPS}{\ensuremath{\!\mathscr{P}\!}}
\newcommand*{\SL}{\mathscr{L}}

\makeatletter
\def\maketag@@@#1{\hbox{\m@th\normalfont\scriptsize#1}}
\makeatother

\newcommand*{\ruleTagText}[1]{{\text{#1}}}

\newcommand*{\leVAL}[1]{\ensuremath{#1}}
\newcommand*{\leVAR}[1]{\ensuremath{#1}}
\newcommand*{\leOP}[2]{\ensuremath{{#1}~\textbf{op}~{#2}}}

\newcommand*{\lcASS}[2]{\ensuremath{{#1}:={#2}}}
\newcommand*{\lcIF}[3]{\ensuremath{{\textbf{if}~{#1}~\textbf{then}~{#2}~\textbf{else}~{#3}}}}
\newcommand*{\lcSEQ}[2]{\ensuremath{{#1};{#2}}}
\newcommand*{\lcWHILE}[2]{\ensuremath{\textbf{while}~{#1}~\textbf{do}~{#2}}}
\newcommand*{\lcCALL}[3]{\ensuremath{{#1}:=\textbf{call}~{#2}({#3})}}
\newcommand*{\lcCP}[3]{\ensuremath{{\textbf{test}({#1})~{#2} ~\textbf{else}~ {#3}}}}
\newcommand*{\lcLETVAR}[3]{\ensuremath{{\textbf{letvar}~ {#1} = {#2} ~\textbf{in}~ {#3}}}}

\newcommand*{\FD}{\ensuremath{\mathit{FD}}}
\newcommand*{\FT}{\ensuremath{\mathit{FT}}}

\newcommand*{\lpRES}{\ensuremath{r}}
\newcommand*{\INITVAL}{0}
\newcommand*{\lpINIT}[2]{\ensuremath{{\textbf{init}~ {#1}=\INITVAL ~\textbf{in}~\{{#2};\textbf{return}~{#1}\}}}}
\newcommand*{\lpPROG}[4]{\ensuremath{{#1}(#2)\big\{{\lpINIT{#4}{#3}}\big\}}}
\newcommand*{\lpFUN}[3]{\ensuremath{\lpPROG{#1}{#2}{#3}{\lpRES}}}

\newcommand*{\app}[2]{\ensuremath{\pi_{{#2}}({#1})}}
\newcommand*{\omerge}[3]{\ensuremath{{#2} \triangleright_{#1} {#3}}}

\newcommand*{\cPair}[2]{\ensuremath{(#1,#2)}}
\newcommand*{\cRawCond}[3]{\ensuremath{\big(#1,#2\leq#3\big)}}
\newcommand*{\cCond}[4]{\ensuremath{\big(\cPair{#1}{#2}\leq\cPair{#3}{#4}\big)}}
\newcommand*{\rank}[1]{{\textbf{r}}(#1)}
\newcommand*{\tlength}[1]{\ensuremath{len(#1)}}

\SetLabelAlign{myAlign}{\textup{[}#1\textup{]}:}
\newlist{ProofEnumDesc}{description}{2}
\setlist[ProofEnumDesc]{style=sameline,nosep}

\newcommand*{\trace}{\ensuremath{\Lambda}}
\newcommand*{\tplus}[1]{\ensuremath{\oplus #1}}
\newcommand*{\tminus}[1]{\ensuremath{\ominus #1}}
\newcommand*{\tplusminus}[1]{\ensuremath{\circledcirc #1}}
\newcommand*{\tplusminuss}[1]{\ensuremath{\circledast #1}}
\newcommand*{\issatisfied}[1]{\ensuremath{\Delta(#1)}}
\newcommand*{\dnf}[1]{\ensuremath{dnf(#1)}}

\begin{document}

\title{A Permission-Dependent Type System for Secure Information Flow Analysis} 

\author{Hongxu Chen}
\affiliation{
  \institution{Nanyang Technological University}           
}

\author{Alwen Tiu}
\affiliation{
  \position{Position1}
  \institution{Nanyang Technological University}           
  \and
  \position{Position2}
  \institution{The Australian National University}
}
\author{Zhiwu Xu}
\affiliation{
  \institution{Shenzhen University}            
}

\author{Yang Liu}
\affiliation{
  \institution{Nanyang Technological University}            
}

\begin{abstract}
We introduce a novel type system for enforcing secure information flow
in an imperative language. Our work is motivated by the problem of
statically checking potential information leakage in Android
applications. To this end, we design a lightweight type 
system featuring Android permission model, where the permissions
are statically assigned to applications and are used to enforce
access control in the applications. We take inspiration
from a type system by Banerjee and Naumann (BN) to allow security
types to be dependent on the permissions of the applications. A novel
feature of our type system is a typing rule for conditional branching
induced by permission testing, which introduces a merging operator on
security types, allowing more precise security policies to be enforced.
The soundness of our type system is proved with respect to a notion
of noninterference. In addition, a type inference algorithm is presented
for the underlying security type system, by reducing the inference problem to a
constraint solving problem in the lattice of security types.
\end{abstract}

\maketitle

\section{Introduction}\label{sec:intro}

This work is motivated by the problem of securing information flow in
Android applications (\emph{apps} for short). We follow
the language-based security approach whereby
information flow is enforced through type systems~\cite{Denning:1976cl,Denning:1977hwa,Volpano:1996,Sabelfeld:2003}.
In particular, we propose a design of a type system that 
guarantees noninterference property~\cite{Volpano:1996}, i.e., typable
programs are noninterferent. As shown in~\cite{Goguen:1982ta}, noninterference provides
a general and natural way to model information flow security. 
noninterference
The type-based approach to noninterference typically requires assigning 
security labels to program variables and security policies to functions or procedures. Such policies
are typically encoded as types, and typeability of the program implies that the runtime behaviour of
the program complies with the stated policies (types). 
Security labels form a lattice structure with an underlying partial order $\leq$,
e.g., a lattice with two elements ``high'' ($H$) and ``low'' ($L$) with $L \leq H$.
Typing rules can then be designed to prevent both explicit flow from $H$ to $L$, 
and implicit flow through conditionals (e.g., if-then-else statement). 
To prevent an explicit flow, the typing rule for an assignment statement
such as \texttt{x := e} would require that $l(e) \leq l(x)$ where $l(.)$
denotes the security level of an expression. To prevent an implicit flow through conditionals, e.g., 
\texttt{if (y = 0) then x:= 0 else x := 1},
most type systems for noninterference require that
the condition $(y=0)$ and the assignments in both branches are given the same
security level. For example, if $y$ is of type $H$ and $x$ is of type $L$, the statement would not be typable. 

\subsection{Motivating Example}
In designing an information flow type system for Android, we encounter 
a common pattern of conditionals that would not be typable using the conventional type system.
Consider the following pseudo-code:
\begin{lstlisting}
String getContactNo(String name) {
    String number; 
    if(checkPermission(READ_CONTACT))
       number = ... ; 
    else number = ""; 
    return number;
}
\end{lstlisting}
Such a code fragment could be part of a phone dialer app or a contact provider, where the function \texttt{getContactNo} provides
a public interface to query the phone number associated with a name. The (implicit) security policy in this
context is that contact information (the phone number) can only be released if the calling app possesses
the required permission (READ\_CONTACT). 
The latter is enforced using the \texttt{checkPermission} API in Android. 
Suppose phone numbers are labelled with $H$, and the empty string is labelled with $L$.
If the interface is invoked by an app that has the required permission, 
the phone number ($H$) is returned; otherwise the empty string ($L$) 
is returned. In either case, there is no leakage of information: in the former case, the calling app
can access the phone number directly anyway; and in the latter case, no phone numbers are returned. 
So by this informal reasoning, the function complies with the implicit security policy and it 
should be safe to be called in any context, whether or not the calling app has the required permission. 
However, in the traditional (non-value dependent) typing rule for the if-then-else construct, one would need to 
assign \emph{the same} security level to both branches, which would result in the return value of the function to be
assigned the security level $H$. As a consequence, if this function is called from an app with no permission,
and if the return value of the function is assigned to a variable with security level $L$, it would be considered
a potential direct flow, even though there is actually no information being leaked. 
To cater for such a scenario, we need to make the security type of \texttt{getContactNo}
 depend on the \emph{permissions} possessed by the caller. 

Banerjee and Naumann~\cite{Banerjee:2005ht} proposed a type system (which we shall refer to as the {\BN} type system) 
that incorporates permissions into  function types. 
Their type system was designed for an access control mechanism different from Android's permission based mechanism, 
but the basic principles are still applicable. 
In their system, permissions are associated with a Java class and need to be \emph{explicitly enabled} for them to have any effect.
Depending on permissions of
the calling app, a function such as \texttt{getContactNo} can have a collection of types. In {\BN} type system, the types of a function take
the form $(l_1,\dots, l_n)\xrightarrow{~P~}l$
where $l_1,\dots,l_n$ denote the security levels of the input
to the function, $l$ denotes the security level of the output
and $P$ denotes a set of permissions \emph{that are not enabled} by the caller.
The underlying idea in their type system is that permissions are 
treated as \emph{guards} to sensitive values, thus conservatively,
one would type the return value of \texttt{getContactNo} as $L$ only if
one knows that the permission READ\_CONTACT is \emph{not enabled}.
The following represents a collection of types for \texttt{getContactNo} 
in {\BN} system
\begin{equation*}
getContactNo : L~\xrightarrow{~P~}~L
\qquad
getContactNo : L~\xrightarrow{~\emptyset~}~H\\
\end{equation*}
where $P = \{ \mathrm{READ\_CONTACT} \}.$
When typing a function call to \texttt{getContactNo} by 
an app without permissions,
the first type of \texttt{getContactNo} will be used; otherwise
the second type will be used. 

In {\BN} type system, the typing judgment is parameterized
by a permission set $P$ containing the permissions that are
\emph{not enabled}. Their language has a command
called \emph{test} which checks the presence of 
a permission set. That is, ``\lcCP{P}{c_1}{c_2}''
means that if permissions in set $P$ are \emph{all enabled}, then 
the command behaves like $c_1$; otherwise it behaves like $c_2$. 
Roughly, the typing rules for the \texttt{test} command (in a much
simplified form) are:
\[
\inference[R1]{
Q \cap P = \emptyset & Q |- c_1 : \tau & Q \vdash {c_2} : \tau
}
{
Q \vdash \lcCP{P}{c_1}{c_2} : \tau
}
\qquad
\inference[R2]{
Q \cap P \not = \emptyset & Q \vdash c_2 : \tau
}
{
Q \vdash \lcCP{P}{c_1}{c_2} : \tau
}
\]
where $Q$ refers to the permissions that are \emph{not enabled}. 
When $Q \cap P \not = \emptyset$, that means 
at least one of the permissions in $P$ is \emph{not enabled}, thus
one can determine statically that only the \emph{else} branch is relevant. This case is
reflected in the typing rule R2.
When $Q \cap P = \emptyset$, it could be that all permissions in $P$
are enabled, \emph{or} it could mean that some permissions in $P$ are not granted to the class.
So in this case, one cannot determine statically which branch of the \texttt{test} will be
taken at runtime. The typing rule R1 therefore conservatively considers typing both branches.

When adapting {\BN} type system to Android, 
we found that in some scenarios R1
is too strong, especially when it is desired that the \emph{absence} of some permissions 
leads to the release of sensitive values. Consider for example an application
that provides location tracking information related to a certain
{\em advertising ID}. The latter provides a unique ID for the purpose of anonymizing mobile users
to be used for advertising (instead of relying on hardware device IDs such
as IMEI numbers). If one can correlate an advertising
ID with a unique hardware ID, it will defeat the purpose of
the anonymizing service provided by the advertising ID. To prevent
that, the following function~\texttt{getInfo} returns the location information for an advertising ID
only if the caller \emph{does not} have access to device ID.
To simplify discussion, let us assume that the permissions to access 
IMEI and location information are denoted by $p$ and $q$, respectively.
\begin{lstlisting}
String getInfo() {
   String r = "";
   test(p) {
     test(q) r = loc; else r = "";
   } else  {
     test(q) r = id++loc; else r = "";
   }
   return r;             
}
\end{lstlisting}
Here \texttt{id} denotes a unique \emph{advertising ID} (rather than IMEI) generated and stored by the app for
the purpose of anonymizing user tracking; \texttt{loc} denotes location
information.
The function first tests whether the caller has access to IMEI number.
If it does, and if it has access to location, then only the location information is returned.
If the caller has no access to IMEI number, but can access location information,
then the combination of advertising id and location \texttt{id++loc} is returned.
In all the other cases, a empty string is returned.
Let us consider a lattice with four elements ordered as:
$L \leq l_1, l_2 \leq H$, where $l_1$ and $l_2$ are incomparable.
Suppose we specify that empty string (``'') is of type $L$, \texttt{loc} is of type $l_1$,
\texttt{id} is of type $l_2$, and the aggregate \texttt{id++loc} is of type $H.$
Consider the case where the caller has both permissions $p$ and $q$ and both are enabled initially. 
When applying {\BN} type system, the desired type of \texttt{getInfo} in this case is  $()~{\xrightarrow{~\emptyset~}}~l_1$.
This means that the type of \texttt{r} has to be at most $l_1$. 
Since no permissions are \emph{not explicitly enabled}, only \ruleTagText{R1} is applicable
to type this program. This, however, will force both branches of \textbf{test}(p)
to have the same type (i.e., security level). As a result, \texttt{r} 
has to be typed as $H$ so that all four assignments in the program
can be typed. 

The issue with the example above is that the stated security policy is non-monotonic in the permission set of 
the calling app. That is, an app with more permissions does not necessarily have access to information with higher
level of security compared to an app with fewer permissions.  The fact that the {\BN} type system
cannot precisely capture the desired policy stated above appears to be a design decision
on their part; we quote from \cite{Banerjee:2005ht}:
\begin{quote}
"On the other hand, one can envision a kind of dual to the design pattern we consider: 
A method may release H information just if a certain permission is absent. 
Our analysis does not handle this pattern and we are unaware of motivating examples; permission checks usually guard sensitive actions such as the release of H information. A tempting idea is to downgrade or declassify information in the presence of certain permissions. But declassification violates noninterference and it is an open problem what is a good information flow property in the presence of declassification"
\end{quote}
As we have seen in the above example, non-monotonic policies can arise quite naturally in the setting of mobile
applications. In general, non-monotonic policies may be required to solve the {\em aggregation problem}
studied in the information flow theory~\cite{Landauer93}, where several pieces of low security level information may
be pooled together to learn information at a higher security level. 

We adopt the basic idea from {\BN} type system
for relating permissions and security types, to design a more precise type system for information flow
under an access control model inspired by Android framework. 
Our type system solves the problem of typing non-monotonic policies such as the one described above, without resorting
to downgrading or declassifying information as suggested in \cite{Banerjee:2005ht}. 
This is done technically via a new operator on security types,
which we call \emph{merging}, to keep information related to both
branches of a test. This in turn requires a significant revision
of the soundness proof for our type system. Additionally, there is a significant
difference in the permission model used in BN type system, where permissions
may be propagated across method invocations, and the permission model we use
here, which is based on Android permission mechanism. Essentially, permissions are relevant
in Android only during inter-apps or inter-components calls, which we model here as
remote procedure calls. More importantly, permissions are not tracked along the call
chains. As we shall see in Section~\ref{sec:type_system}, this may create a potential attack via ``parameter laundering'' and
the design of the typing rule for remote procedure call needs special attention.

\subsection{Contributions}

The contributions of our work are three-fold.
\begin{enumerate}
	\item We develop a lightweight type system in which security types are functions from permissions to a lattice of security levels, and prove its soundness with respect to a notion of noninterference (Sec.~\ref{sec:type_system}). It allows us to precisely model information flow that depends on security policies, including those that are non-monotonic in terms of permission set containment.
	\item We present a novel approach to model the permission checking mechanism in Android (Sec.~\ref{sec:type_system} \& Sec.~\ref{sec:constraint_gen}) which requires a different approach in typing (remote) function calls. 
	\item We propose a decidable inference algorithm by reducing type inference to a constraint solving problem and applying an efficient algorithm to solve it (Sec.~\ref{sec:constraint_gen}), which is useful in recovering precise security types (such as the \texttt{getInfo} case above).
\end{enumerate}

Due to space constraints, detailed proofs of some lemmas and theorems are omitted, but they are available in the
supplementary material provided along with this submission.

 \section{A Secure Information Flow Type System}\label{sec:type_system}

In this section, we present a simple imperative language featuring
function calls and permission checks. We first discuss informally a permission-based access control model we consider, which is an abstraction of the permission mechanism used in Android. 
We then give the operational semantics of a simple imperative language that includes permission checking 
constructs based on the abstract permission model described in Section~\ref{sec:permission}. 
In Section~\ref{sec:types} we provide the type system for our
imperative language and prove its soundness with respect to a notion of noninterference. 
The detailed definitions and proofs of this section can be found at Appx.~\ref{sec:app-soundness}.

\subsection{A model of permission-based access control}
\label{sec:permission}

The Android framework has a large and complicated code base, which we will not be able to model completely.
Instead we focus on the permission model used in inter process or inter component communication in Android.
Such permissions are used to regulate access to protected resources, such as device id, location information, contact information, etc.
As mentioned in Sec.~\ref{sec:intro}, the motivation
 is to design an information flow type system incorporating
some access control mechanisms in Android, where access to 
data and system services may be associated with permissions. 

An app specifies the permissions it needs at installation time, via a
meta-file called the manifest file.  
In recent versions of Android, since version 6.0, some of these permissions need to be confirmed at runtime. But
at no point a permission request is allowed if it is not already specified in the manifest.
For now, we assume a permission enforcement mechanism that applies to Android prior to version 6.0, so it does not
account for permission granting at runtime. This will have a consequence in terms of noninterference for
non-monotonic policies (i.e, those policies that link the absence of certain permissions to  ``high'' security levels, such
as the example we presented in the introduction).
We shall come back to this point later in Section~\ref{sec:conclusion}. 
 
An Android app may provide services to other apps,
or other components within itself. Such a service provider
may impose permissions on other apps who want to 
access its services.  Communication between apps are implemented
through the Binder IPC (interprocess communications) mechanism~\cite{Android-Binder-IPC}. 
In our model, a program can be seen as a very much simplified version of a service application in Android, 
and the main intention here is to show how one can reason about information flow in
such a service provider when access control is imposed on the calling application. 
In the following we shall not model explicitly the IPC mechanism of Android,
but will instead abstract it into a (remote) function call. 
Note that this abstraction is practical since it can be achieved by conventional data and control flow analysis, together with the modeling of Android IPC specific APIs. And this has been implemented by existing static analysis frameworks like FlowDroid~\cite{Arzt:2014:FPC:2666356.2594299}, Amandroid~\cite{Wei:2014:APG:2660267.2660357}, IccTA\cite{Li:2015:IDI:2818754.2818791}, etc~\footnote{In fact, we are also implementing a permission-dependent information flow analysis tool on top of Amandroid. The fundamental idea is similar to the one mentioned here, however the focus is improving the precision of information leakage \emph{detection} rather than noninterference \emph{certification}.}.

Android framework does not track the IPC call chains between apps and permissions of an app are not
propagated when a remote procedure is called.  An app A calling another app B does not grant B the permissions assigned to A.
This is different from the BN type system where permissions can potentially propagate along the call stacks.
Note however that B can potentially have more permissions than A, leading to a potential privilege escalation,
which is a known weakness in Android permission system~\cite{Chin:2011wa}.  Another consequence of this lack of transitivity
is that in designing the type system, one must be careful to avoid what we call a "parameter laundering" attack (see Section~\ref{sec:types}).

Access to resources in Android are guarded by a set of permissions. In most of the cases, at most one permission is needed
for each access (for system resources), but individual apps may implement custom permissions and may require the presence or absence of one or more permissions. 
A (system) resource can be accessed through different API calls. We assume that these API calls are implemented
correctly and the permissions required are enforced consistently (note that this is not necessarily the case with
Android, as the study in \cite{ShaoNDSS16} has shown). Some API calls may require less permissions than others though they
allow access to the same resource. In principle, our typing system can be applied to type check the entire
framework, but this is not practical. In practice, one needs to make simplifying assumptions about the
enforcement of the permissions in the framework. 

\subsection{An Imperative Language with Permission Checks}\label{sec:language}

We do not model directly the Android (mostly Java) source code but a much simplified language, which is a variant 
of the imperative language considered in \cite{Volpano:1996}, extended with functions and an operator for permission checks. 

We model an \emph{app}  as a collection of \emph{functions} (\emph{services}), together with a statically assigned permission set;  a \emph{system}, denoted as $\mathcal{S}$, consists of a set of apps. We use capital letters $A,B,\ldots$ to denote apps. A function $f$ defined in an app $A$ is abbreviated to $A.f$, and may be called independently of other functions in the same app. 
The intention is that a function models an application component (i.e., \emph{Activity}, \emph{Service}, \emph{BroadCastReceiver}, or \emph{ContentProvider}) in Android, which may be called from within the same app or from other apps. We assume that at any given time only one function is executed, so we do not model concurrent executions of apps. Additionally, functions in a system are assumed to be not (mutually) recursive, so in a given system, there is a finite chain of function calls from any given function. Each app is assigned a static set of permissions, drawn from a finite set of permissions $\PS$. The powerset of $\PS$ is written as $\PPS$.

For simplicity, we consider only programs manipulating \emph{integers}, 
so the expressions in our language 
all have the integer type. Boolean values are encoded as zero (false) and nonzero (true).
The grammar for expressions is thus as follows:
\begin{equation*}
e ::= n \mid \leVAR{x} \mid \leOP{e}{e}
\end{equation*}
where $n$ denotes an integer literal, $\leVAR{x}$ denotes a variable, and $\leOP{e}{e}$ denotes a binary operation.
The commands of the language are given in the following grammar:
\begin{equation*}
\begin{array}{ll}
  c\;::=\; &\lcASS{x}{e} \mid \lcIF{e}{c}{c} \mid  \lcWHILE{e}{c} \mid \lcSEQ{c}{c}  \\
  & \mid \lcLETVAR{x}{e}{c} \mid \lcCALL{x}{A.f}{\overline{e}} \mid \lcCP{p}{c}{c} 
\end{array}
\end{equation*}
The first four constructs are respectively assignment, conditional, while-loop and sequential composition.
The statement ``\lcLETVAR{x}{e}{c}'' is a local variable declaration statement. Here $x$ is
declared and initialized to $e$, and its scope is the command $c$. We require
that $x$ does not occur in $e.$
The statement ``\lcCALL{x}{A.f}{\overline{e}}''  denotes an assignment whose right hand side is
a function call to $A.f$.
``\lcCP{p}{c_1}{c_2}'' 
checks whether the calling app has permission $p$: if it does then $c_1$ is executed,
otherwise $c_2$ is executed. 
This is similar to the \emph{test} construct in {\BN}'s paper. 
But unlike the \emph{test} construct in their language, where \emph{test} can be done on a set of permissions,
we allow testing \emph{only one permission} at a time.
This is the typical practice in Android~\cite{Android-CheckPerm}. And it is not a restriction in theory, since we can use nested permission checks to simulate {\BN}'s \emph{test} behavior.

A function declaration has the following syntax:
\begin{equation*}
F ::= \lpPROG{A.f}{\overline{x}}{c}{{\lpRES}}
\end{equation*}
where $A.f$ is the name of the function, $\bar x$ are function parameters,
``$\textbf{init}~ r={\INITVAL}$ in \{c;~\textbf{return}~\lpRES\}'' is the
function body, where $c$ is a command and {\lpRES} is a local variable 
that holds the return value of the function.  
The variables $\overline{x}$ and $r$ are bound variables with the command $c; \textbf{return}~r$ in their scopes.
We consider only {\em closed functions}, i.e., the variables occurring in $c$
are either introduced by \textbf{letvar} or from the set $\{\overline{x}, r\}$.
To simplify proofs of various properties, we assume that
bound variables are named differently
so there are no naming clashes between them in a system.

\begin{figure}[ht]
$$
\inference[E-VAL]{}{\eta \vdash \leVAL{v}  \leadsto v} \qquad
\inference[E-VAR]{}{\eta \vdash \leVAR{x}  \leadsto \eta(x)}  
$$

$$
\inference[E-OP]{\eta \vdash e_1 \leadsto v_1\quad \eta \vdash e_2 \leadsto v_2}
{\eta \vdash \leOP{e_1}{e_2} \leadsto v_1~op~v_2}
\quad 
\inference[E-ASS]{\eta \vdash e \leadsto v}
{\eta; A; P \vdash \lcASS{x}{e} \leadsto \eta [ x \mapsto v ]}
$$

$$
\inference[E-IF-T]{\eta \vdash e \leadsto v\quad v\neq 0\quad \eta; A; P \vdash c_1 \leadsto \eta'}
{\eta; A; P \vdash \lcIF{e}{c_1}{c_2} \leadsto \eta'}
\quad
\inference[E-IF-F]{\eta \vdash e \leadsto v\quad v = 0\quad \eta;A; P \vdash c_2 \leadsto \eta'}
{\eta;A; P \vdash \lcIF{e}{c_1}{c_2} \leadsto \eta'}
$$

$$
\inference[E-WHILE-T]{
\eta \vdash e \leadsto v\quad v\neq 0\quad 
\eta;A; P \vdash c \leadsto \eta'\\
\quad \eta' ; A ; P \vdash \lcWHILE{e}{c} \leadsto \eta''
}
{\eta;A; P \vdash \lcWHILE{e}{c} \leadsto \eta''}
\quad
\inference[E-WHILE-F]{\eta \vdash e \leadsto v\quad v = 0}
{\eta;A; P \vdash \lcWHILE{e}{c} \leadsto \eta}
$$

\[
\inference[E-SEQ]{\eta;A; P \vdash c_1 \leadsto \eta'
\quad \eta';A; P \vdash c_2 \leadsto \eta''}
{\eta;A;P \vdash \lcSEQ{c_1}{c_2} \leadsto \eta''}
\quad
\inference[E-LETVAR]{\eta \vdash e \leadsto v \quad \eta [ x \mapsto v ]; A ; P \vdash c \leadsto \eta'}
{\eta; A; P \vdash \lcLETVAR{x}{e}{c} \leadsto \eta' - x}
\]

\[
\inference[E-CP-T]{p\in P \quad \eta; A; P \vdash c_1 \leadsto \eta'}
{\eta; A ; P \vdash \lcCP{p}{c_1}{c_2} \leadsto \eta'}
\quad
\inference[E-CP-F]{p\notin P \quad \eta;A;P \vdash c_2 \leadsto \eta'}
{\eta;A;P \vdash \lcCP{p}{c_1}{c_2} \leadsto \eta'}
\]

\[
\inference[E-CALL]{
\FD(B.f) = \lpPROG{B.f}{\overline{y}}{c}{\lpRES} \\
\eta\vdash\overline{e} \leadsto \overline{v} \quad
[\overline{y} \mapsto \overline{v}, r \mapsto 0] ; B; \Theta(A) \vdash c \leadsto \eta'}
{\eta;A;P \vdash \lcCALL{x}{B.f}{\overline{e}} \leadsto \eta[x \mapsto \eta'(r) ]}
\]

\caption{Evaluation rules for expressions and commands, 
given a function definition table $\FD$ and a permission assignment $\Theta.$
}
\label{fig:semantics}
\end{figure}
 
\subsection{Operational Semantics}\label{sec:semantics}
We assume that function definitions in a system are stored in a
table $\FD$ indexed by function names, and the permission set assigned to a certain app is given by a table $\Theta$ indexed by app names.

An {\termEmph{evaluation environment} 
is a finite mapping from variables to values (i.e., integers).
We denote with $EEnv$ the set of evaluation environments.  
Elements of $EEnv$ are ranged over by $\eta$. 
We also use the notation 
$[ x_1 \mapsto v_1, \cdots,  x_n \mapsto v_n]$
to denote an evaluation environment mapping variable $x_i$ to value $v_i$ for $i \in \{1,\dots, n\}$; this will sometimes be abbreviated as $[ \overline x \mapsto \overline v ].$
The domain of $\eta = [ x_1 \mapsto v_1, \cdots, x_n \mapsto v_n]$ (i.e., $\{x_1,\dots,x_n\}$) is written as $dom(\eta)$.
Given two environments $\eta_1$ and $\eta_2$, we define 
$\eta_1\eta_2$ as an environment $\eta$ such that $\eta(x) = \eta_2(x)$ if $x \in dom(\eta_2)$, 
otherwise $\eta(x) =\eta_1(x)$.
For example, $\eta[x \mapsto v]$ maps $x$ to $v$, and $y$ to $\eta(y)$ 
for any $y \in dom(\eta)$ such that $y \not = x.$
Given a mapping $\eta$ and a variable $x$, we write $\eta\!-\!x$ to denote the
mapping by removing $x$ from $dom(\eta)$.

The semantics of expressions and commands are defined via
the evaluation rules in Fig.~\ref{fig:semantics}. 
The semantics of a function will be implicitly defined in function calls.
The evaluation judgment
for expressions has the form $\eta\vdash e\leadsto v$,
which states that expression $e$ evaluates to value $v$ when variables
in $e$ are interpreted in the evaluation environment $\eta.$
We write $\eta \vdash \overline{e} \leadsto \overline{v}$,
where $\overline{e}=e_1,\dots,e_n$ and $\overline{v} = v_1,\dots,v_n$ for some $n$,
to denote a sequence of judgments
$\eta \vdash e_1 \leadsto v_1, \ldots, \eta \vdash e_n \leadsto v_n.$

The evaluation judgment for commands takes the form
$\eta;A;P\vdash c \leadsto \eta'$
where $\eta$ is an evaluation environment before the execution
of the command $c$, and $\eta'$ is the evaluation environment
after the execution of $c$. The permission set $P$ denotes
the permissions of the \emph{caller} of command, and $A$ refers to the app
to which the command $c$ belongs. 

The operational semantics of most commands are straightforward.
We explain the semantics of the \emph{test} primitive and the function call.
Rules (\ruleTagText{E-CP-T}) and (\ruleTagText{E-CP-F}) capture the semantics of the \emph{test} primitive.
These are where the permission $P$ in the evaluation judgement is used.
The semantics of function calls is given by (\ruleTagText{E-CALL}). Notice
that $c$ inside the body of \emph{callee} is executed under
the permission of $A$, i.e., $\Theta(A)$. The permission $P$
in the conclusion of that rule is not used in the premise. That is,
permissions of the caller of app $A$ are not transferred over to
the callee function $B.f$. This reflects the way permissions in Android
are passed on during IPCs~\cite{Android-CheckPerm,Android-Binder-IPC}, and is also a major difference
between our handling of permissions and {\BN}'s, 
where permissions are inherited by successive function calls. 

 \subsection{Security Types}\label{sec:types}

In information flow type systems such as \cite{Volpano:1996}, it is common to adopt a lattice
structure to encode security levels. Security types in this setting are just security
levels. In our case, we generalize the security types to account for
the \emph{dependency of security levels on permissions}. So we shall distinguish
security levels, given by a lattice structure which encodes sensitivity levels
of information, and security types, which are mappings from permissions to security levels. 
We assume the security levels are given by a lattice $\SL$, with a partial order $\leq_{\SL}$.
Security types are defined in the following.

\begin{definition}\label{def:base_type}
A \termEmph{base security type} (or \termEmph{base type}) $t$ is a mapping from $\PPS$ to $\SL$.
We denote with $\Tcal$ the set of base types. 
Given two base types $s,t$, we say $s=t$ \text{iff} $s(P)=t(P)$ $\forall$ $P\in\PPS$ and  
$s\leq_\Tcal t$ \text{iff} $\forall$ $P\in \PPS$, $s(P) \leq_{\SL} t(P)$.
\end{definition}

As we shall see, if a variable is typed by a base type, the sensitivity of
its content may depend on the permissions of the app which writes to the variable.
In contrast, in traditional information flow type systems,
a variable annotated with a security level has a fixed sensitivity level
regardless of the permissions of the app who writes to the variable. 

A security level $l$ can also be treated as a special base type $t_l$ that is a constant function, mapping all permission sets to level $l$ itself. Therefore, we define an \emph{embedding function} $\hat{\cdot}$
from security levels to base types, such that $\hat{l}=t_l$. This means that, given a security level $l$, we have $\hat{l}(P)=l, \forall P\in\PPS$.

Next, we show that the set of base types with the order $\leq_{\Tcal}$ forms a lattice.

\begin{lemma}\label{lem:tleq_po}
	$\leq_\Tcal$ is a partial order relation on $\Tcal$.
\end{lemma}

\begin{definition}\label{def:type-cup-cap}
For $s, t\in\Tcal$, $s\sqcup t$ and $s\sqcap t$ are defined as
\begin{align*}
(s\sqcup t)(P) =  s(P)\sqcup t(P), \forall P \in \PPS\\
(s\sqcap t)(P) =  s(P)\sqcap t(P), \forall P \in \PPS
\end{align*}
\end{definition}

\begin{lemma}\label{lem:lattice}
$(\Tcal, \leq_\Tcal)$ forms a lattice.
\end{lemma}
 
From now on, we shall drop the subscripts in $\leq_{\SL}$ and $\leq_\Tcal$ when no ambiguity arises.

 \begin{definition}\label{def:fun-type}
   A \termEmph{function type} has the form $\overline{t} \rightarrow t$, where $\overline{t}=(t_1,\ldots,t_m)$, $m\geq 0$ and $t$ are base types.
   The types $\overline t$ are the types for the arguments of the function and $t$ is the return type of the function.
 \end{definition}

In the type system that we shall define later, security types of variables and expressions may be 
altered depending on the execution context. That is, when a variable is used in
a context where a permission check has been performed (either successfully or unsuccessfully),
the type of the variable may be adjusted to take into account the presence or absence
of the checked permission. 
Such an adjustment is called a {\termEmph{promotion} or a {\termEmph{demotion}.

\begin{definition}\label{def:updown}
Given a permission $p$, the \termEmph{promotion} and \termEmph{demotion} of a base type $t$
with respect to $p$ are: 
\begin{align*}
(t\uparrow_{p}) (P) = t(P\cup \{p\}), \forall P\in \PPS
\tag{promotion}
\\
(t\downarrow_{p}) (P) = t(P\setminus\{p\}), \forall P\in \PPS
\tag{demotion}
\end{align*}
The \termEmph{promotion} and \termEmph{demotion} of a function type~$\overline{t}\rightarrow t$,
where $\overline{t} = (t_1,\dots, t_m)$, are respectively: 
\begin{align*}
(\overline{t}\rightarrow t)\uparrow_{p} = \overline{t}\uparrow_{p} \rightarrow t\uparrow_{p},~
\text{where}~\overline{t}\uparrow_{p} = (t_1\uparrow_{p},\ldots, t_m\uparrow_{p}),
\\
(\overline{t}\rightarrow t)\downarrow_{p} = \overline{t}\downarrow_{p}\rightarrow t\downarrow_{p},~
\text{where}~\overline{t}\downarrow_{p} = (t_1\downarrow_{p},\ldots, t_m\downarrow_{p}).
\end{align*}
\end{definition}

\subsection{Security Type System}\label{sec:typing_rules}

We first define a couple of operations on security types and permissions that will
be used later.

\begin{definition}\label{def:projection}
Given $t\in\Tcal$ and $P\in\PPS$, the \termEmph{projection} of $t$ on $P$ is a security type  $\app{t}{P}$ defined as follows:
\begin{equation*}
\app{t}{P}(Q) = t(P),\; \forall Q\in \PPS.
\end{equation*}
Type projection of a list of types
on $P$ is then written as
\[
\app{(t_1,\dots,t_n)}{P} = (\app{t_1}{P}, \dots, \app{t_n}{P}).
\]
\end{definition}
\begin{definition}\label{def:merge}
Given a permission $p$ and types $t_1$ and $t_2$, the {\termEmph{merging}} of 
$t_1$ and $t_2$ along $p$, denoted as $\omerge {p} {t_1} {t_2}$, is: 
\begin{equation*}
(\omerge p {t_1} {t_2}) (P) = 
\begin{cases}
t_1(P) & p\in P \\
t_2(P) & p\not\in P\\
\end{cases}
\quad\forall P\in{\PPS}
\end{equation*}
\end{definition}

A {\em typing environment} is a finite mapping from variables to base types. 
We use the notation  $[x_1 : t_1, \dots, x_n : t_n]$
to enumerate a typing environment with domain $\{x_1,\dots,x_n\}.$
Typing environments are ranged over by $\Gamma.$
Given $\Gamma_1$ and $\Gamma_2$ such that $dom(\Gamma_1) \cap dom(\Gamma_2) = \emptyset$,
we write $\Gamma_1\Gamma_2$ to denote a typing environment that is the (disjoint) union of the mappings
in $\Gamma_1$ and $\Gamma_2$.

\begin{definition}\label{def:tenv-pd}
  Given a typing environment $\Gamma$, its {\em promotion} and {\em demotion} along $p$
  are typing environments $\Gamma\!\uparrow_p$ and $\Gamma\!\downarrow_p$, such that
$(\Gamma\!\uparrow_{p})(x) = \Gamma(x)\!\uparrow_{p}$ and 
$(\Gamma\!\downarrow_{p})(x) = \Gamma(x)\!\downarrow_{p}$ for every $x \in dom(\Gamma).$	
  The projection of $\Gamma$ on $P \in \PPS$ is a typing environment $\app{\Gamma}{P}$ such that  
$(\app{\Gamma}{P})(x) = \app{\Gamma(x)}{P}$ for each $x \in dom(\Gamma).$
\end{definition}

There are three typing judgments in our type system as explained below. All these judgments are implicitly parameterized by
a function type table, $\FT$, which maps all function names to function types, and a mapping $\Theta$ assigning permission sets to apps.

\begin{itemize}
\item Expression typing: $\Gamma \vdash e : t.$ 
This says that under $\Gamma$,
the expression $e$ has a base type at most $t$.

\item Command typing: $\Gamma; A \vdash c : t$.
This means that the command $c$ writes to variables with type at least $t$, when 
executed by app $A$, under the typing environment $\Gamma.$
 
\item Function typing:
The typing judgment takes the form: 
$$\vdash \lpFUN{B.f}{\overline{x}}{c} :  \overline{t} \xrightarrow{} t'$$
where $\overline{x} = (x_1,\dots,x_n)$ and
$\overline{t} = (t_1,\dots,t_n)$ for some $n \geq 0.$
Functions in our setting are 
polymorphic in the permissions of the caller.
Intuitively, this means that each caller of the function above with permission set $P$
``sees'' the function as having type
$\app{\overline{t}}{P} \rightarrow \app{t'}{P}.$
That is, if the function is called from another app with permission $P$,
then it expects input of type at least $\app{\overline{t}}{P}$ and
a return value of type up to $\app{t'}{P}$.

\end{itemize}
The typing rules are given in Fig.~\ref{fig:typing-rules}.
Most of these rules are common to 
information flow type systems~\cite{Volpano:1996,Banerjee:2005ht,Sabelfeld:2003},
except for (\ruleTagText{T-CP}) and (\ruleTagText{T-CALL}), which we explain below.

In Rule (\ruleTagText{T-CP}), to type statement \lcCP{p}{c_1}{c_2},
we need to type $c_1$ in a promoted typing environment, to account for the case
where the permission check for $p$ succeeds, and $c_2$ in a demoted typing environment,
to account for the case where the permission check fails.
The challenge here is then how one combines the types obtained in the two premises
of the rule to obtain the type for the conclusion. One possibility is to force
the type of the two premises and the conclusion to be identical to those in the conditional construct typing, i.e.,
Rule (\ruleTagText{T-IF}). This, as we have seen in Sec.~\ref{sec:intro},
leads to a loss in precision of the type for the test statement. Instead, we consider
a more refined \emph{merged} type $\omerge{p}{t_1}{t_2}$ for the conclusion,
where $t_1$ ($t_2$ resp.) is the type of the left (right resp.) premise.
To understand the intuition behind the merged type, 
consider a scenario where the statement
is executed in a context where permission $p$ is \emph{present}.
Then the permission check succeeds and
the statement $\lcCP{p}{c_1}{c_2}$ is equivalent to $c_1$.
In this case, one would expect that the behaviour of $\lcCP{p}{c_1}{c_2}$ would be
equivalent to that of $c_1$. This is in fact captured by the equation
$
(\omerge p {t_1} {t_2})(P) = t_1(P)
$
for all $P$ such that $p \in P$, which holds by definition. A dual scenario arises when
$p$ is not in the permissions of the execution context.

In Rule (\ruleTagText{T-CALL}), the callee function $B.f$ is assumed to be typed checked already
and its type is given in the $FT$ table. Here the function $B.f$ is called
by $A$ so the type of $B.f$ as seen by $A$ should be
a projection of the type given in $FT(B.f)$ on the permissions of $A$ (given by $\Theta(A)$): 
$\app{\overline{t}}{\Theta(A)} \rightarrow \app{t'}{\Theta(A)}.$
Therefore the arguments for the function call should be typed
as 
$\Gamma \vdash \overline{e} : \app{\overline{t}}{\Theta(A)}$
and the return type (as viewed by $A$) should be dominated
by the type of $x$, i.e., $\app{t'}{\Theta(A)} \leq \Gamma(x)$.

It is essential that in Rule (\ruleTagText{T-CALL}), the arguments $\overline{e}$ and the return value of
the function call are typed according to the \emph{projection of $\overline{t}$} and $t'$ on $\Theta(A)$. If they are instead typed with $\overline{t}$, then there is a potential implicit flow via what we call a ``parameter laundering'' attack. To see why, consider the following alternative formulation of (\ruleTagText{T-CALL}): 
\begin{equation*}
\inference[T-CALL']
{
 \FT(B.f) = \overline{t} \rightarrow t' &
 \Gamma \vdash \overline{e} :\overline{t} &
 t'  \leq \Gamma(x)}
{\Gamma; A \vdash \lcCALL{x}{B.f}{\overline{e}} :
  \Gamma(x)}
\end{equation*}
Now consider the following functions: 
\begin{lstlisting}
A.f(x) {  
	init r in {
	   r := call B.g (x); 
	   return r
	}  
}

B.g(x) {
	init r in {
		test(p) r := 0 else r := x;
		return r
	}
}

C.getsecret() {
	init r in {
		test(p) r := SECRET else r := 0;
		return r
	}
}

M.main() { 
	init r in
	letvar 
	( 
	   r := call A.f(
	);
	return r
}
\end{lstlisting}

Let $\PS=\{p\}$ and let $t$ be the base type $t =  \{\emptyset \mapsto L, \{p\} \mapsto H\}$, 
where $L$ and $H$ are bottom and top levels respectively. Here we assume the value SECRET is a ``high'' value that needs to be protected so
we require that the function C.getsecret to have type $() \rightarrow t$. That is, only apps that have the required permission $p$ may obtain
the secret value. 
Suppose the permissions assigned to the apps are given by:  $\Theta(A) =  \Theta(B) = \emptyset, \Theta(C) = \Theta(M) = \{p\}.$ 
If we were to adopt the modified \ruleTagText{T-CALL'} instead of \ruleTagText{T-CALL}, then we can assign the following 
types to the above functions:
\[
FT := \left\{
\begin{array}{lcl}
A.f & \mapsto & t \rightarrow \hat{L} \\
B.g & \mapsto & t \rightarrow \hat{L}\\
C.getsecret & \mapsto & () \rightarrow t\\
M.main & \mapsto  & () \rightarrow \hat{L}
\end{array}
\right.
\]
Notice that the return type of $M.main$ is $\hat{L}$ despite having a return value that contains SECRET. 
If we were to use \ruleTagText{T-CALL'} in place of \ruleTagText{T-CALL}, the above functions can
be typed as follows: 

\begin{prooftree}
\AxiomC{$FT(B.g) = t \rightarrow \hat{L}$}
\AxiomC{$x : t, r : \hat{L} \vdash x : t$}
\AxiomC{$\hat{L} \leq t$}
\LeftLabel{T-CALL'}
\TrinaryInfC{$x : t, r : \hat{L} ; A \vdash \lcCALL{r}{B.g}{x} : \hat{L}$}
\LeftLabel{T-FUN}
\UnaryInfC{$\vdash \lpPROG{A.f}{x} {\lcCALL{r}{B.g}{x} }{\lpRES} : t \rightarrow \hat{L}$}
\end{prooftree}

\vskip1ex

\begin{prooftree}
\AxiomC{$x : t\uparrow_p,  r : \hat{L}\uparrow_p \vdash \lcASS{r}{0} : \hat{L}$ }
\AxiomC{$x : t\downarrow_p,  r : \hat{L}\downarrow_p \vdash \lcASS{r}{x} : \hat{L} $ }
\LeftLabel{T-CP}
\BinaryInfC{$x : t, r : \hat{L} \vdash \lcCP {p} { \lcASS{r}{0} } {\lcASS{r}{x}} : \omerge{p}{\hat{L}}{\hat{L}} $}
\LeftLabel{T-FUN}
\UnaryInfC{$\vdash \lpPROG{B.g}{x}{ \lcCP {p} { \lcASS{r}{0} } {\lcASS{r}{x}} }{\lpRES} : t \rightarrow \hat{L}$}
\end{prooftree}
Note that $t\uparrow_p = \hat{H}$ and $t\downarrow_p = \hat{L} = \hat{L}\downarrow = \hat{L}\uparrow.$

\vskip1ex

\begin{prooftree}
\AxiomC{$r : t\uparrow_p \vdash \lcASS{r}{\mathrm{SECRET}} : \hat{H}$}
\AxiomC{$r : t\downarrow_p \vdash \lcASS{r}{0} : \hat{L}$}
\LeftLabel{T-CP}
\BinaryInfC{$r : t \vdash \lcCP{p}{\lcASS{r}{\mathrm{SECRET}}}{\lcASS{r}{0}} :  \omerge{p}{\hat{H}}{\hat{L}} $}
\LeftLabel{T-FUN}
\UnaryInfC{$\vdash \lpPROG{C.getsecret}{~}
		{ \lcCP{p}{\lcASS{r}{\mathrm{SECRET}}}{\lcASS{r}{0}}}
		{\lpRES} : () \rightarrow t$}
\end{prooftree}
Note that $\omerge{p}{\hat{H}}{\hat{L}} = t.$

Finally, still assuming \ruleTagText{T-CALL'}, a partial typing derivation for $M.main$ is as follows: 
\begin{prooftree}
\AxiomC{$r : \hat{L} \vdash 0 : t$}
	\AxiomC{$\Gamma ; M \vdash \lcCALL{x_H}{C.getsecret}{~} : \hat{L}$}
	\AxiomC{$\Gamma ; M \vdash \lcCALL{r}{A.f}{x_H} : \hat{L}$ }
\LeftLabel{T-SEQ}
\BinaryInfC{$r : \hat{L}, x_H : t ; M \vdash \lcSEQ {\lcCALL{x_H}{C.getsecret}{~} } {\lcCALL{r}{A.f}{x_H} }  :  \hat{L}$}
\LeftLabel{T-LETVAR}
\BinaryInfC{$r : \hat{L} ~; M \vdash \lcLETVAR{x_H}{0}
		{\lcSEQ {\lcCALL{x_H}{C.getsecret}{~} } {\lcCALL{r}{A.f}{x_H} } } : \hat{L} $}
\LeftLabel{T-FUN}
\UnaryInfC{$
\begin{array}{l}
\vdash M.main(~) \{ \textbf{init} ~ r = 0 ~ \textbf{in} \\
\quad \textbf{letvar} ~ x_H = 0 ~ \textbf{in} \{ \\
\qquad \lcCALL{x_H}{C.getsecret}{~}; \\ 
\qquad \lcCALL{r}{A.f}{x_H} \\
\quad \} \\
\quad \textbf{return}~r 
\quad \}
\end{array}
: () \rightarrow \hat{L}
$}
\end{prooftree}
where $\Gamma = \{ r : \hat{L}, x_H : t\}.$ 
To finish this typing derivation, we need to complete the derivations for the leaves:
\begin{prooftree}
\AxiomC{$FT(C.getsecret) = () \rightarrow t$}
\AxiomC{$\Gamma \vdash () : () $}
\AxiomC{$t \leq \Gamma(x_H) = t $}
\LeftLabel{T-CALL'}
\TrinaryInfC{$\Gamma; M \vdash \lcCALL{x_H}{C.getsecret}{~} : t$}
\LeftLabel{T-SUB${}_c$}
\UnaryInfC{$\Gamma; M \vdash \lcCALL{x_H}{C.getsecret}{~} : \hat{L}$}
\end{prooftree}

\vskip1ex

\begin{prooftree}
\AxiomC{$FT(A.f) = t \rightarrow \hat{L}$}
\AxiomC{$\Gamma \vdash x_H : t$}
\AxiomC{$\hat{L} \leq \Gamma(x_H) = t$}
\LeftLabel{T-CALL'}
\TrinaryInfC{$\Gamma ; M \vdash \lcCALL{r}{A.f}{x_H} : \hat{L}$ }
\end{prooftree}

Since $M.main$ returns a sensitive value, the typing rule \ruleTagText{T-CALL'} is obviously unsound.
With the correct typing rule for function calls, the function $A.f$ cannot be assigned type
$t \rightarrow \hat{L}$, since that would require the instance of T-CALL (i.e., when making the
call to $B.g$) in this case to satisfy
the constraint:
\[
x : t, r : \hat{L} \vdash x : \app{t}{\Theta(A)}
\]
where $\app{t}{\Theta(A)} = \hat{L}$, which is impossible since $t \not \leq \hat{L}.$ What this
means is essentially that in our type system, information received by an app $A$ from the parameters 
cannot be propagated by $A$ to another app $B$, unless $A$ is already authorized to access the
information contained in the parameter. Note that this only restricts the propagation of such parameters
to other apps; the app $A$ can process the information internally without necessarily violating
the typing constraints.

Finally, the reader may check that if we fix the type of $B.g$ to $t \rightarrow \hat{L}$ then
$A.f$ can only be assigned type $\hat{L} \rightarrow \hat{L}.$ In no circumstances can
$M.main$ be typed, since the statement $x_H := C.getsecret()$ forces $x_H$ to have
type $\hat{H}$, and thus cannot be passed to $A.f$ as an argument. 

\begin{figure}

\[
\inference[T-VAR]{}
{\Gamma \vdash \leVAR{x} : \Gamma(x)}
\qquad
\inference[T-OP]
{\Gamma\vdash e_1: t & \Gamma ; A\vdash e_2: t}
{\Gamma\vdash \leOP{e_1}{e_2} : t}
\]

\[
\inference[T-SUB$_{e}$]
{\Gamma\vdash e : s & s \leq t}
{\Gamma\vdash e : t}
\qquad 
\inference[T-SUB$_{c}$]
{\Gamma; A \vdash c : s & t \leq s}
{\Gamma; A\vdash c : t}
\]

\[
\inference[T-ASS]
{\Gamma \vdash e : \Gamma(x)}
{\Gamma ;A \vdash \lcASS{x}{e} : \Gamma(x)}
\qquad
\inference[T-LETVAR]
{\Gamma \vdash e : s &
\Gamma[x:s]; A \vdash c : t }
{\Gamma; A\vdash \lcLETVAR{x}{e}{c} :  t}
\]

\[
\inference[T-IF]
{\Gamma \vdash e : t & \Gamma ; A\vdash c_1 : t & \Gamma; A \vdash c_2 : t }
{\Gamma; A \vdash \lcIF{e}{c_1}{c_2} : t}
\quad
\inference[T-CP]
{\Gamma\uparrow_{p} ; A \vdash c_1 : t_1 &
\Gamma\downarrow_{p} ; A\vdash c_2 : t_2}
{\Gamma; A\vdash \lcCP{p}{c_1}{c_2} : \omerge p {t_1} {t_2}}
\]

\[
\inference[T-WHILE]
{\Gamma \vdash e : t & \Gamma ; A\vdash c : t  }
{\Gamma; A \vdash \lcWHILE{e}{c} : t }
\quad
\inference[T-SEQ]
{\Gamma; A\vdash c_1 : t & \Gamma ; A\vdash c_2 : t }
{\Gamma ; A\vdash \lcSEQ{c_1}{c_2} : t}
\]

\[
\inference[T-CALL]
{\FT(B.f) = \overline{t} \rightarrow t' &
\Gamma \vdash \overline{e} : \app{\overline{t}}{\Theta(A)} & 
\app{t'}{\Theta(A)} \leq \Gamma(x)}
{\Gamma; A \vdash \lcCALL{x}{B.f}{\overline{e}} : 
 \Gamma(x)  }
\]

\[
\inference[T-FUN]
{[\overline{x}:\overline{t}, r : t'] ; B\vdash c : s}
{ \vdash \lpFUN{B.f}{\overline{x}}{c} :  
  \overline{t} \rightarrow t'}
\]
\caption{Typing rules for expressions, commands and functions, given a function-type table $\FT$ and a permission assignment $\Theta$.
}
\label{fig:typing-rules}
\end{figure}

\subsection{Noninterference and Soundness Proof}\label{sec:noninterference}

We now discuss a notion of \termEmph{noninterference} and prove the soundness of our type system
with respect to it.
Firstly we define an \termEmph{indistinguishability} relation between evaluation environments.
Such a definition typically assumes an observer who may observe values of variables at a certain security level.
In the non-dependent setting, the security level of the observer is fixed, say at $l_O$, and valuations
of variables at level $l_O$ or below are required to be identical. In our setting, the security level of
a variable can vary depending on the permissions of the caller app of a particular service, and 
the observer may itself be an app within the system. We do not assume a priori that the policies assigned
to functions and variables to agree with one another, e.g., it may be the case that we have two variables $x : t$ and $y : t'$  
such that $t(P) \not = t'(P)$ where $P$ is the permission set associated with the observer. So to be sound, 
our notion of indistinguishability needs to take into account all the security levels assigned to permission set $P$ by all variables.
In effect, this is equivalent to simply fixing the security level of the observer and ignoring the permission set $P$ of the observer.
This leads to the following definition. 

\begin{definition}\label{def:envequ}
Given two evaluation environments $\eta, \eta'$, 
a typing environment $\Gamma$, a security level $l_O \in \SL$ of the observer,  
{\termEmph{indistinguishability relation}} is defined as follows:
\begin{equation*}
\eta =_{\Gamma}^{l_{O}} \eta' \text{ iff. } 
\forall x\in dom(\Gamma) .\; \big(\Gamma(x) \leq \hat{l}_{O} \Rightarrow \eta(x) = \eta'(x) \big)
\end{equation*}
where 
$\eta(x) = \eta'(x)$ holds iff both sides of the equation are defined and equal,
or both sides are undefined.
\end{definition}

Notice that in Definition~\ref{def:envequ}, since base types are functions from permissions to security level, the security
level $l_O$ needs to be lifted to a base type in the comparison $\Gamma(x) \leq \hat{l}_O.$ 
This ordering on base types implies the ordering on security levels
$\Gamma(x)(P) \leq l_O$ (in the latice $\SL$) for every permission set $P.$
This is equivalent to saying that $l_O$ is an upper bound of
the security levels of $x$ under all possible permission sets: 
$
\sqcup \{\Gamma(x)(P) \mid P \in \PPS \} \leq l_O. 
$
If the base type of each variable assigns the same security level to every permission set (i.e., the security level
is independent of the permissions), then our notion of indistinguishability coincides with the standard definition
for the non-dependent setting.

For the remaining of this paper, we assume that the security level of the observer is fixed to some
value $l_O \in \SL.$

\begin{lemma}\label{lem:ni-eq}
$=_{\Gamma}^{l_{O}}$ is an equivalence relation on $EEnv$.
\end{lemma}

\begin{lemma}
\label{lem:proj}
If $\eta =_{\Gamma}^{l_{O}} \eta'$
then for each $P\in\PPS$,
$\eta =_{\app{\Gamma}{P}}^{l_O} \eta'.$
\end{lemma}

Before we proceed with proving the soundness of our type system, we
need to make sure that the system of apps are \termEmph{well-typed} in
the following sense:
\begin{definition}\label{def:cmd-welltype}
Let $\mathcal{S}$ be a system, and let $\FD$, $\FT$ and $\Theta$ be
its function declaration table, function type table, and permission
assignments. We say $\mathcal{S}$ is \termEmph{well-typed} iff for
every function $A.f$, $\vdash \FD(A.f) : \FT(A.f)$ is derivable.
\end{definition}
Recall that we assume no (mutual) recursions, so every function call
chain in a well-typed system is finite; this is formalized via the rank function below. We will use this
as a measure in our soundness proof (Lemma~\ref{lem:comni}).
\begin{equation*}
\arraycolsep=1.0pt\def\arraystretch{1}
\begin{array}{lcl}
\rank{\lcASS{x}{e}} & = &0\\
\rank{\lcIF{e}{c_1}{c_2}} & = &max(\rank{c_1}, \rank{c_2)}\\ 
\rank{\lcSEQ{c_1}{c_2}} & = &max(\rank{c_1}, \rank{c_2})\\
\rank{\lcWHILE{e}{c}} & = &\rank{c} \\
\rank{\lcLETVAR{x}{e}{c}} & = &\rank{c} \\
\rank{\lcCALL{x}{A.f}{\overline{e}}} & = &
   \rank{FD(A.f)} + 1\\
\rank{\lcCP{p}{c_1}{c_2}} & = & max(\rank{c_1}, \rank{c_2})\\
\rank{\lpFUN{A.f}{\overline x}{c}} &=& \rank{c}
\end{array}
\end{equation*}

\noindent The next two lemmas relate projection, promotion/demotion
and the indistinguishability relation.
\begin{lemma}\label{lem:up}
If $p \in P$, then $\eta =_{\app{\Gamma}{P}}^{l_{O}} \eta'$ iff $\eta =_{\app{\Gamma\uparrow_{p}}{P}}^{l_{O}} \eta'$.
\end{lemma}

\begin{lemma}\label{lem:down}
If $p \notin P$, then $\eta =_{\app{\Gamma}{P}}^{l_{O}} \eta'$ iff $\eta =_{\app{\Gamma\downarrow_{p}}{P}}^{l_{O}} \eta'$.
\end{lemma}

The following two lemmas are the analogs to the simple security property 
and the confinement property in~\cite{Volpano:1996}.
\begin{lemma}\label{lem:expsafe}
Suppose $\Gamma\vdash e : t$. For $P\in\PPS$, if~$t(P)\leq l_{O}$ and 
$\eta =_{\app{\Gamma}{P}}^{l_{O}} \eta'$,  $\eta\vdash e\leadsto v$ and $\eta'\vdash e\leadsto v'$, then $v = v'$.
\end{lemma}

\begin{lemma}\label{lem:comsafe}
Suppose $\Gamma; A \vdash c : t$. Then for any $P\in\PPS$, 
if $t(P)\nleq l_{O}$ and $\eta;A; P \vdash c \leadsto \eta' $, then 
$\eta =_{\app{\Gamma}{P}}^{l_O} \eta'$.
\end{lemma}
\begin{definition}\label{def:ni}
A command $c$ executed in app $A$ is said to be {\termEmph{noninterferent}} \text{iff.}
for all $\eta_1, \eta'_1,\Gamma, P, l_{O}$, 
if  $\eta_1 =_{\app{\Gamma}{P}}^{l_{O}} \eta'_1$, \; $\eta_1;A ; P\vdash c \leadsto \eta_2 $ and 
 $\eta'_1; A; P\vdash c\leadsto \eta'_2  $ 
then $\eta_2 =_{\app{\Gamma}{P}}^{l_{O}} \eta'_2$.
\end{definition}
The main technical lemma is that well-typed commands are noninterferent.
\begin{lemma}\label{lem:comni}
Suppose $\Gamma; A\vdash c : t$, for any $P\in\PPS$, if  $\eta_1 =_{\app{\Gamma}{P}}^{l_{O}} \eta'_1$, $\eta_1; A; P \vdash c \leadsto \eta_2$,
and  $\eta'_1; A; P \vdash c \leadsto \eta'_2$, 
then  $\eta_2 =_{\app{\Gamma}{P}}^{l_{O}} \eta'_2$. 
\end{lemma}
\begin{proof}
The proof proceeds by induction on \rank{c}, with subinduction
on the derivations of $\Gamma;\Theta;A \vdash c : t$ and $\eta_1;\Theta;P;A\vdash c \leadsto \eta_2$.
In the following, we omit the superscript $l_O$ from
$=^{l_O}_{\app{\Gamma}{P}}$ to simplify presentation.
We show two interesting cases here; the complete proof can be found in the appendices.

\textbf{T-CALL:}  
In this case, $c$ has the form 
$\lcCALL{x}{B.f}{\overline{e}}$. Suppose the typing
derivation is the following (where we label the premises
for ease of reference later):
$$
\inference
{
\FT(B.f) = \overline{s} \rightarrow s'  &
(\mathbf{T_1}) ~ \Gamma \vdash \overline{e} : \app{\overline{s}}{\Theta(A)} &
(\mathbf{T_2}) ~ \app{s'}{\Theta(A)} \leq \Gamma(x)}
{\Gamma; A \vdash \lcCALL{x}{B.f}{\overline{e}} : 
 \Gamma(x)  }
$$
where $t = \Gamma(x)$,
and the executions under $\eta_1$ and $\eta_1'$ are 
derived, respectively, as follows:
$$
\inference{
(\mathbf{E_1}) ~ \eta_1\vdash\overline{e} \leadsto \overline{v_1} &
(\mathbf{E_2}) ~ [\overline{y} \mapsto \overline{v_1}, r \mapsto 0] ; B; 
  \Theta(A) \vdash c_1 \leadsto \eta_3
}
{\eta_1; A; P \vdash \lcCALL{x}{B.f}{\overline{e}} \leadsto 
              \eta_1 [x \mapsto \eta_3(r) ]}
$$
$$
\inference{
(\mathbf{E_1'}) ~ \eta_1'\vdash\overline{e} \leadsto \overline{v_2}
&
(\mathbf{E_2'}) ~ [\overline{y} \mapsto \overline{v_2}, r \mapsto 0] ; B; 
\Theta(A) \vdash c_1 \leadsto \eta_3'}
{\eta_1';A;P \vdash \lcCALL{x}{B.f}{\overline{e}} \leadsto 
              \eta_1' [x \mapsto \eta_3'(r) ]}
$$
where 
$\FD(B.f) = \lpFUN{B.f}{\overline{x}}{c_1}$, 
$\eta_2 = \eta_1[x \mapsto \eta_3(r)]$
and $\eta_2' = \eta_1'[x \mapsto \eta_3'(r)].$

Moreover, since we consider only well-typed systems,
the function $\FD(B.f)$ is also typable:
$$
\inference
{
(\mathbf{T_3}) ~ [\overline{y}:\overline{s}, r : s']; B\vdash c_1 : s
}
{
\Theta \vdash \lpFUN{B.f}{\overline{y}}{c_1} :  
\overline{s}\rightarrow s'
}
$$
First we note that if $t(P) \nleq l_O$ then the result
follows from Lemma~\ref{lem:comsafe}.
So in the following, we assume $t(P) \leq l_O$.
Since $t = \Gamma(x)$, it follows that
$\Gamma(x)(P) \leq l_O.$

Let $\Gamma' = \app{[\overline{y} : \overline{t}, r : s]}{\Theta(A)}.$
We first prove several claims:
\begin{itemize}
\item Claim 1: 
$[\overline{y} \mapsto \overline{v_1}, r \mapsto 0] 
=_{\Gamma'} [\overline{y} \mapsto \overline{v_2}, r\mapsto 0].$

Proof: Let $\rho = [\overline{y} \mapsto \overline{v_1}, r \mapsto 0]$
and $\rho' = [\overline{y} \mapsto \overline{v_2}, r \mapsto 0]$.
We only need to check that the two mappings
agree on mappings of $\overline{y}$ that are of type $\leq \hat{l}_O.$
Suppose $y_u$ is such a variable, i.e., $\Gamma'(y_u) = u \leq \hat{l}_O$,
and suppose $\rho(y_u) = v_u$ and
$\rho'(y_u) = v_u'$ for some $y_u \in \overline{y}.$
From $(\mathbf{E_1})$ we have 
$\eta_1 \vdash e_u \leadsto v_u$
and from $(\mathbf{E_2})$ 
we have $\eta_1' \vdash e_u \leadsto v_u'$,
and from $(\mathbf{T_1})$ we have $\Gamma \vdash e_u : u.$ 
Since $u \leq l_O$, applying Lemma~\ref{lem:expsafe},
we get $v_u = v_u'$.

\item Claim 2: $\eta_3 =_{\Gamma'} \eta_3'.$ 

Proof: From Claim 1, we know that 
$$
[\overline{y} \mapsto \overline{v_1}, r \mapsto 0] 
=_{\Gamma'} [\overline{y} \mapsto \overline{v_2}, r \mapsto 0].
$$
Since $\rank{c_1}<\rank{c}$, we can apply the outer induction
hypothesis to $(\mathbf{E_2})$, $(\mathbf{E_2'})$
and $(\mathbf{T_3})$ to obtain
$\eta_3 =_{\Gamma'} \eta_3'.$

\item Claim 3: $\eta_3(r) = \eta_3'(r).$

Proof: We first note that from $(\mathbf{T_2})$
and the assumption that $\Gamma(x)(P) \leq l_O$, we get
$(\app{s'}{\Theta(A)})(P) \leq l_O$.
The latter, by Definition~\ref{def:projection}, implies  that $s'(\Theta(A)) \leq l_O.$
Since $r \in dom(\Gamma')$, it is obvious that
$\Gamma' \vdash r : s'$,
$\eta_3 \vdash r \leadsto \eta_3(r)$
and $\eta_3' \vdash r \leadsto \eta_3'(r).$
From Claim 2, we have $\eta_3 =_{\Gamma'} \eta_3'$.
Therefore by Lemma~\ref{lem:expsafe}, we have
$\eta_3(r) = \eta_3'(r).$
\end{itemize}
The statement we are trying to prove, i.e., $\eta_2 =_{\app{\Gamma}{P}} \eta_2'$,
follows immediately from Claim 3.

{\bf T-CP:} Suppose $c$ is $\lcCP{p}{c_1}{c_2}$ and we have
$$
\inference
{\Gamma\uparrow_{p} ; A \vdash c_1 : t_1 &
\Gamma\downarrow_{p} ; A\vdash c_2 : t_2 &
t=\omerge {p} {t_1} {t_2}}
{\Gamma; A\vdash \lcCP{p}{c_1}{c_2} : t}
$$
Suppose that $p \in P.$
Then the evaluation of $c$ under $\eta_1$ and $\eta_1'$
are respectively:
\[
\inference{p\in P \quad \eta_1; A; P \vdash c_1 \leadsto \eta_2}
{\eta_1; A ; P \vdash \lcCP{p}{c_1}{c_2} \leadsto \eta_2}
\qquad
\inference
    {p\in P \quad \eta_1'; A; P \vdash c_1 \leadsto \eta_2'}
    {\eta_1'; A ; P \vdash \lcCP{p}{c_1}{c_2} \leadsto \eta_2'}
\]
Since $\eta_1 =_{\app{\Gamma}{P}} \eta_1'$ and since $p \in P$,
by Lemma~\ref{lem:up}, we have $\eta_1 =_{\app{\Gamma\uparrow_p}{P}} \eta_1'$.
Therefore by the induction hypothesis applied to $\Gamma\uparrow_p ; A \vdash c_1 : t_1$,
$\eta_1; A; P \vdash c_1 \leadsto \eta_2$ and
$\eta_1'; A; P \vdash c_1 \leadsto \eta_2'$, we obtain
$\eta_2 =_{\app{\Gamma\uparrow_p}{P}} \eta_2'$, and by Lemma~\ref{lem:up},
we get $\eta_2 =_{\app{\Gamma}{P}} \eta_2'$.

For the case where $p \not \in P$, we apply a similar reasoning as above,
but using Lemma~\ref{lem:down} in place of Lemma~\ref{lem:up}.
\end{proof}

\begin{definition}\label{def:sys-ni}
Let $\mathcal{S}$ be a system.
A function
$$\lpFUN{A.f}{\overline{x}}{c}$$
in $\mathcal{S}$
with $FT(A.f) = \overline{t} \rightarrow t'$
is {\termEmph{noninterferent}}
if for all $\eta_1, \eta'_1, P, v, l_{O}$,
if the following hold:   
\begin{itemize}
\item $t'(P) \leq l_{O}$,
\item $\eta_1 =_{\app{\Gamma}{P}}^{l_{O}} \eta'_1$, where
$\Gamma = [\overline{x} : \overline{t}, r : t']$,  
\item $\eta_1;A ; P\vdash c \leadsto \eta_2 $, and $\eta'_1; A; P\vdash c\leadsto \eta'_2 $, 
\end{itemize}
then $\eta_2(r) = \eta_2'(r).$
The system $\mathcal{S}$ is \termEmph{noninterferent} iff all functions
in $\mathcal{S}$ are noninterferent.
\end{definition}

The noninterferent property of well-typed systems follows
from Lemma~\ref{lem:comni}. 
\begin{theorem}\label{thm:ni}
Well-typed systems are noninterferent.
\end{theorem}
 \section{Type Inference}\label{sec:constraint_gen}
This section provides a decidable inference algorithm for the language  in Sec.~\ref{sec:language}.
Sec.~\ref{sec:trace_rules} presents the permission trace rules (Fig.~\ref{fig:infer-rules}) and discusses its equivalence to the typing rules (Fig.~\ref{fig:typing-rules}) in Sec.~\ref{sec:types}; based on this, Sec.~\ref{sec:constraint-gen-rules} reduces the type inference into a constraint solving problem; Sec.~\ref{sec:constraint_solve} provides the detailed procedures to solve the generated constraints. Detailed definitions and proofs can be found at Appx.~\ref{sec:app-infer}.

\subsection{Permission Tracing}\label{sec:trace_rules}
To infer a security type for an expression, a command or a function, we need to track the adjustments of variables depending on permission checks, i.e., the applications of promotions $\Gamma\!\uparrow_{p}$ and demotions $\Gamma\!\downarrow_{q}$ in their typing derivations.
To this end,  we keep the applications symbolic and collect the promotions and demotions into a sequence.  
In other words, we treat them as a sequence of promotions $\uparrow_{p}$ and demotions $\downarrow_{p}$ applied on a typing environment $\Gamma$. 
For example, $(\Gamma\!\uparrow_{p})\downarrow_{q}$ can be viewed as 
an application of the sequence $\uparrow_{p}\downarrow_{q}$ on $\Gamma$.
The sequence of promotions and demotions is called \emph{permission trace} and denoted as \trace. The grammar of {\trace} is:
\begin{equation*}
\trace~::=~\tplus{p} ::\trace~|~ \tminus{p} :: \trace~|~ \epsilon
\quad p\in \PS
\end{equation*}

Let $occur(p,\trace)$ be the number of occurrences of $p$ in $\trace$.
We say $\trace$ is consistent 
iff. $occur(p,\trace) \in \{0,1\}, \forall p\in\PS$. 
The length of $\trace$, denoted as $\tlength{\trace}$, is defined as: 
\begin{align*}
	\tlength{\trace}=
	\begin{cases}
	    0 & \text{if}~ \trace=\epsilon\\
		1+\tlength{\trace'} & \text{if}~ \trace=\tplusminus{p} :: \trace', \tplusminus{} \in\{\tplus{}, \tminus{}\}
	\end{cases}
\end{align*}

\begin{definition}\label{def:app-pt-t}
Given a base type $t$ and a permission trace $\trace$, the \termEmph{application} of $\trace$ on $t$, denoted as $t\cdot\trace$,  is defined as:
\begin{equation*}
{t\cdot \trace =}
\begin{cases}
t & {\normalfont\text{if}}~\trace = \epsilon\\
(t\uparrow_{p})\cdot \trace' & {\normalfont\text{if}}~\exists p, \trace'. (\trace = \tplus{p} :: \trace')\\
(t\downarrow_{p})\cdot \trace' & {\normalfont\text{if}}~\exists p, \trace'. (\trace = \tminus{p} :: \trace')
\end{cases}
\end{equation*}
\end{definition}

We also extend the application of a permission trace $\trace$ to a typing environment $\Gamma$ (denoted as $\Gamma\cdot\trace$), such that $\forall x.\; (\Gamma\cdot{\trace})(x) = \Gamma(x)\cdot \trace$.
Based on permission traces, we give the definition of \termEmph{partial subtyping relation}.
\begin{definition}\label{def:partial-subtype}
The \termEmph{partial subtyping relation} $\leq_{\trace} $, which is the subtyping relation applied on the permission trace, is defined as $s\leq_{\trace}  t \text{ iff. } s\cdot \trace \leq t\cdot \trace$.
\end{definition}

\iffalse
\begin{lemma}\label{lem:monotrace}
$\forall s, t \in \Tcal$, $s\leq t \implies s\leq_{\trace} t$.
\end{lemma}
\fi

\begin{lemma}\label{lem:traceorder}
$\forall t\in\Tcal$, $p,q\in\PS$ s.t. $ p\neq q$ , $t \cdot (\tplusminus{p} \tplusminuss{q}) =  t \cdot (\tplusminuss{q}  \tplusminus{p}) $, where $\tplusminus{},\tplusminuss{}, \in\{\tplus{},\tminus{}\}$. 
\end{lemma}

\begin{lemma}\label{lem:traceorder-whole}
$\forall t\in\Tcal$,$(t\cdot \tplusminus{p} )\cdot \trace = (t\cdot\trace)\cdot \tplusminus{p}$, where $\tplusminus{}\in\{\tplus{}, \tminus{}\}$ and $p \notin \trace$.
\end{lemma}

\begin{lemma}\label{lem:tracepsame}
$\forall t\in\Tcal$, $p\in\PS$,$(t \cdot \tplusminus{ p})\cdot\tplusminuss{ p} =  t \cdot (\tplusminus{ p}) $, where $\tplusminus{}, \tplusminuss{} \in \{\tplus{},\tminus{}\}$.
\end{lemma}

\begin{lemma}\label{lem:tracesame}
$\forall t\in\Tcal$, $(t \cdot \trace)\cdot \trace =  t \cdot \trace $. 
\end{lemma}

Lemmas \ref{lem:traceorder} and \ref{lem:traceorder-whole} state that the order of applications of promotions and demotions on \emph{different} permissions does not affect the result. Lemmas \ref{lem:tracepsame} and \ref{lem:tracesame} indicate that only the first application takes effect if there exist several (consecutive) applications of promotions and demotions on the \emph{same} permission $p$. Therefore, we can safely keep the first application, by removing the other applications on the same permission.
In the remaining, we assume that all permission traces are consistent.
Moreover, to ensure that the traces collected from the derivations of commands are 
consistent, we assume that in nested permission checks, each permission is checked 
at most once. 

\iffalse
\begin{lemma}\label{lem:tracemerge}
$\forall s,t\in\Tcal.\; \forall p\in\PS.  (\omerge{p}{s}{t})\cdot \trace = \omerge{p}{(s\cdot\trace)}{(t\cdot\trace)} $, where $p\notin \trace$.
\end{lemma}

\begin{lemma}\label{lem:tracesub}
$\forall s,t\in\Tcal.\; \forall p\in\PS.$ $s\leq t\Leftrightarrow  s\cdot\tplus{ p}\leq t\cdot\tplus{ p}$ and $s\cdot\tminus{ p}\leq t\cdot\tminus{ p}$. 
\end{lemma}
\fi

Besides, we split the applications of the promotions and demotions into two parts  (i.e., typing environments and permission traces), and move the subsumption rules for expressions and commands to where they are needed. 
This yields the syntax-directed  typing rules, which are called the permission trace rules and given in Fig.~\ref{fig:infer-rules}.
The judgments of the trace rules are similar to those of typing rules, except that each trace rule is guarded by the permission trace $\trace$ collected from the context, which keeps track of the adjustments of variables depending on the permission checks,
and that the subtyping relation in the trace rules is the partial subtyping one $\leq_{\trace}$.

\begin{figure}[t]

\[
\inference[TT-VAR]{}
{\Gamma;\trace\vdash_{t} \leVAR{x} : \Gamma(x)}
\qquad
\inference[TT-OP]
{\Gamma;\trace\vdash_{t} e_1: t_1 & \Gamma;\trace \vdash_{t} e_2: t_2}
{\Gamma;\trace \vdash_{t} \leOP{e_1}{e_2} : t_1\sqcup t_2}
\]

\[
\inference[TT-ASS]
{\Gamma;\trace \vdash_{t} e : t & t \leq_{\trace} \Gamma(x)}
{\Gamma;\trace \vdash_{t} \lcASS{x}{e} : \Gamma(x)}
\quad
\inference[TT-LETVAR]
{\Gamma;\trace \vdash_{t} e : s &
\Gamma[x:s']; \trace; A \vdash_{t} c : t & s \leq_{\trace} s' }
{\Gamma;\trace; A\vdash_{t} \lcLETVAR{x}{e}{c} :  t}
\]

\[
\inference[TT-IF]
{
\Gamma;\trace \vdash_{t} e : t & 
\Gamma;\trace; A \vdash_{t} c_1 : t_1 &
\Gamma;\trace; A \vdash_{t} c_2 : t_2 & 
t \leq_{\trace} t_1\sqcap t_2
}
{\Gamma;\trace;  A \vdash_{t} \lcIF{e}{c_1}{c_2} : t_1\sqcap t_2}
\]

\[
\inference[TT-WHILE]
{
\Gamma;\trace \vdash_{t} e : s & 
\Gamma;\trace; A \vdash_{t} c : t &  s\leq_{\trace} t  }
{\Gamma;\trace; A \vdash_{t} \lcWHILE{c}{e} : t }
\]

\[
\inference[TT-SEQ]
{\Gamma;\trace; A \vdash_{t} c_1 : t_1 \quad \Gamma;\trace; A \vdash_{t} c_2 : t _2}
{\Gamma;\trace; A \vdash_{t} \lcSEQ{c_1}{c_2} : t_1\sqcap t_2}
\]

\[
\inference[TT-CALL]
{
\FT(B.f) = \overline{t} \xrightarrow{} t'  & 
\Gamma;\trace \vdash_{t} \overline{e} : \overline{s} &
\overline{s}\leq_{\trace} \overline{\app{t}{\Theta(A)}} & 
\app{t'}{\Theta(A)}\leq_{\trace} \Gamma(x)
}
{\Gamma;\trace; A \vdash_{t} \lcCALL{x}{B.f}{\overline{e}} : \Gamma(x)}
\]

\[ 
\inference[TT-CP]
{\Gamma;\trace::\tplus{ p}; A \vdash_{t} c_1 : t_1 &
\Gamma;\trace::\tminus{p}; A \vdash_{t} c_2 : t_2}
{\Gamma;\trace; A \vdash_{t} \lcCP{p}{c_1}{c_2} : \omerge{p} {t_1}{t_2}}
\]

\[
\inference[TT-FUN]
{
[\overline{ x}: \overline{t},\lpRES:t'];\epsilon; B \vdash_{t} c : s
}
{\vdash_{t} \lpPROG{B.f}{\overline{x}}{c}{\lpRES} :  \overline{t}\xrightarrow{} t'}
\]

\caption{Trace rules for expressions, commands and functions, given a function-type table $\FT$ and a permission assignment $\Theta$.}
\label{fig:infer-rules}
\end{figure}

The next two lemmas show the trace rules are sound and complete with respect to the typing rules, i.e., an expression (command, function, resp.)  is typable under the trace rules, if and only if it is typable under the typing rules.

\begin{lemma}\label{lem:ptrsound}
\begin{enumerate}[label=(\alph*),topsep=1pt,itemsep=-1ex,partopsep=1ex,parsep=0ex]
\item\label{lem:ptrsound-1} If $\Gamma;\trace\vdash_{t} e : t$, then $\Gamma\cdot\trace\vdash e : (t\cdot\trace)$.\\
\item\label{lem:ptrsound-2} If $\Gamma;\trace; A \vdash_{t} c : t$, then $(\Gamma\cdot\trace); A \vdash c : (t\cdot\trace)$.\\
\item\label{lem:ptrsound-3} If $\vdash_{t} \lpPROG{B.f}{\overline{x}}{c}{\lpRES}  : \overline{t}\xrightarrow{} t'$, then $ \vdash \lpPROG{B.f}{\overline{x}}{c}{\lpRES}  : \overline{t}\xrightarrow{} t'$.
\end{enumerate}
\end{lemma}

\iffalse
To prove the completeness, we need to prove some auxiliary lemmas.
\begin{lemma}\label{lem:typingtrace}
Let $\trace$ be the permission trace collected from the context of $e$ or $c$. 
\begin{enumerate}[label=(\alph*),topsep=1pt,itemsep=-1ex,partopsep=1ex,parsep=1ex]
\item\label{lem:typingtrace-1} If $\Gamma \vdash e : t$, then $(\Gamma\cdot\trace) \vdash e : (t\cdot\trace) $ 
\item\label{lem:typingtrace-2} If $\Gamma;A \vdash c : t$, then $(\Gamma\cdot\trace); A \vdash c : (t\cdot\trace) $.
\end{enumerate}
\end{lemma}

\begin{lemma}\label{lem:typingtrace1}
\begin{enumerate}[label=(\alph*),topsep=1pt,itemsep=-1ex,partopsep=1ex,parsep=1ex]
\item\label{lem:typingtrace1-1} If $\Gamma\cdot\trace \vdash e : t$, then $(\Gamma\cdot\trace) \vdash e : (t\cdot\trace) $ 
\item\label{lem:typingtrace1-2} If $(\Gamma\cdot\trace);A \vdash c : t$, then $(\Gamma\cdot\trace); A \vdash c : (t\cdot\trace) $.
\end{enumerate}
\end{lemma}
\fi

\begin{lemma}\label{lem:ptrcomplete}
\begin{enumerate}[label=(\alph*),topsep=1pt,itemsep=-1ex,partopsep=1ex,parsep=0ex]
\item\label{lem:ptrcomplete-1} If $\Gamma\cdot\trace \vdash e : t\cdot \trace$, then there exists $s$ such that $\Gamma;\trace \vdash_{t} e : s$ and $s \leq_{\trace} t$.\\
\item\label{lem:ptrcomplete-2} If $(\Gamma\cdot\trace); A \vdash c : t\cdot \trace$, then there exists $s$ such that $\Gamma;\trace;A \vdash_{t} c : s$ and $t \leq_{\trace} s$.\\
\item\label{lem:ptrcomplete-3} If $\vdash \lpPROG{B.f}{\overline{x}}{c}{\lpRES}  : \overline{t}\xrightarrow{} s$, then $\vdash_{t} \lpPROG{B.f}{\overline{x}}{c}{\lpRES}: \overline{t}\xrightarrow{} s$.
\end{enumerate}
\end{lemma}

\subsection{Constraint Generation}\label{sec:constraint-gen-rules}
This section provides the constraint generation rules to reduce the type inference into a constraint solving problem.

To infer function types in System $\mathcal{S}$, we firstly assign a function type $\alpha\!\rightarrow\!\beta$ for each function $A.f$ whose type is unknown and a type variable $\gamma$ for each  variable $x$ with unknown type respectively, where $\alpha,\beta,\gamma$ are fresh type variables.
Then according to the trace rules, we try to build a derivation for $\mathcal{S}$, in which we collect the side conditions needed by the rules, i.e., the partial subtyping relation $\leq_{\trace}$ appearing in the rules.
If the side conditions hold under a context, then $\FD(A.f)$ is typed by $\FT(A.f)$ under the same context for each function $A.f$ in $\mathcal{S}$.
 
To describe the side conditions (\emph{i.e.},  $\leq_{\trace}$), we define the permission guarded constraints as follows: 
\begin{equation*}
\begin{array}{lll}
c &::=& (\trace, t_l \leq t_r)\\
t_l &::=& \alpha~|~t_{g} ~ |~ t_l~\sqcup~t_l ~|~\app{t_l}{P}\\
t_r &::=&\alpha~|~t_{g} ~|~t_r~\sqcap~t_r ~|~\omerge{p}{t_r}{t_r}~|~\app{t_r}{P}\\
\end{array}
\end{equation*}
where $\trace$ is a permission trace, $\alpha$ is a fresh type variable and $t_{g}$ is a ground type.

A \termEmph{type substitution} is a finite mapping from type variables to security types:
$
\theta ~::=~ \epsilon~|~\alpha \mapsto t,\theta
$

\begin{definition}
Given a constraint set $C$ and a substitution $\theta$, we say $\theta$ is a \termEmph{solution} of $C$, denoted as $\theta \vDash C $, iff. for each $(\trace, t_l \leq t_r) \in C$, $t_l\theta \leq_{\trace} t_r\theta $ holds.
\end{definition}

The constraint generation rules are presented in Fig. \ref{fig:constraint-rules}, where
$\FT_{C}$ is the extended function type table such that $\FT_{C}$ maps all function names to function types and their corresponding constraint sets.
The judgments of the constraint rules are similar to those of trace rules, except that each constraint rule generates a constraint set $C$, which consists of the side conditions needed by the typing derivation of $\mathcal{S}$.
In addition,  as the function call chains starting from a command are finite, the constraint generation will terminate.

\begin{figure}[ht]

\[
\inference[TG-VAR]
{}
{\Gamma;\trace \vdash_{g} \leVAR{x} : \Gamma(x)\leadsto \emptyset}
\qquad
\inference[TG-OP]
{\Gamma;\trace \vdash_{g} e_1: t_1 \leadsto C_1 &
 \Gamma;\trace \vdash_{g} e_2: t_2 \leadsto C_2}
{\Gamma;\trace \vdash_{g} \leOP{e_1}{e_2} : t_1\sqcup t_2\leadsto C_1\cup C_2}
\]

\[
\inference[TG-ASS]
{\Gamma;\trace \vdash_{g} e : t \leadsto C }
{\Gamma;\trace; A \vdash_{g} \lcASS{x}{e} : \Gamma(x)\leadsto C\cup \{ (\trace, t \leq_{} \Gamma(x))\} }
\]

\[
\inference[TG-LETVAR]
{\Gamma;\trace \vdash_{g} e : s \leadsto C_1&
\Gamma[x:\alpha]; \trace; A \vdash_{g} c : t \leadsto C_2 &
C = C_1\cup C_2 \cup \{(\trace, s \leq_{} \alpha) \}
}
{\Gamma;\trace; A\vdash_{g} \lcLETVAR{x}{e}{c} :  t\leadsto C}
\]

\[
\inference[TG-IF]
{
\Gamma;\trace; A \vdash_{g} c_1 : t_1 \leadsto C_1 &
\Gamma;\trace; A \vdash_{g} c_2 : t_2 \leadsto C_2 \\
\Gamma;\trace \vdash_{g} e : t \leadsto C_e &
C = C_e\cup C_1\cup C_2 \cup \{(\trace, t \leq_{} t_1\sqcap t_2) \}}
{\Gamma;\trace;  A \vdash_{g} \lcIF{e}{c_1}{c_2} : t_1\sqcap t_2 \leadsto C}
\]

\[
\inference[TG-WHILE]
{
\Gamma;\trace \vdash_{g} e : s  \leadsto C & 
\Gamma;\trace; A \vdash_{g} c : t \leadsto C'
}
{\Gamma;\trace;  A \vdash_{g} \lcWHILE{e}{c} : t  \leadsto C\cup C' \cup\{(\trace,s\leq_{} t)\}}
\]

\[
\inference[TG-SEQ]
{
\Gamma;\trace; A \vdash_{g} c_1 : t_1  \leadsto C_1 &
\Gamma;\trace; A \vdash_{g} c_2 : t_2 \leadsto C_2
}
{\Gamma;\trace; A \vdash_{g} \lcSEQ{c_1}{c_2} : t_1\sqcap t_2 \leadsto C_1\cup C_2}
\]

\[
\inference[TG-CALL]
{
\FT_{C}(B.f) = (\overline{t} \xrightarrow{} t', C_{f})  &
\Gamma;\trace \vdash_{g} \overline{e} : \overline{s} \leadsto \bigcup \overline{C_{e}} \\
C_{a} = \{(\trace, \overline{s}\leq_{} \overline{\app{t}{\Theta(A)}}), (\trace, \app{t'}{\Theta(A)}\leq_{} \Gamma(x))\} \\
C = C_{f} \cup \bigcup\overline{C_{e}}\cup C_{a}
}
{
\Gamma;\trace; A \vdash_{g} \lcCALL{x}{B.f}{\overline{e}} : \Gamma(x)\leadsto C
}
\]

\[
\inference[TG-CP]
{
\Gamma;\trace::\tplus p; A \vdash_{g} c_1 : t_1  \leadsto C_1 &
 \Gamma;\trace::\tminus p; A \vdash_{g} c_2 : t_2  \leadsto C_2
 }
{\Gamma;\trace; A \vdash_{g} \lcCP{p}{c_1}{c_2} : \omerge{p}{t_1}{t_2} \leadsto C_1\cup C_2} 
\]

\[
\inference[TG-FUN]
{
[\overline{x}: \overline{\alpha},\lpRES:\beta]; \epsilon; B \vdash_{g} c : s  \leadsto C
}
{
\vdash_{g} \lpPROG{B.f}{x}{c}{\lpRES} :  \overline{\alpha}\xrightarrow{} \beta  \leadsto  C
}
\]

\caption{Constraint generation rules for expressions, commands and functions, given function type table $FT_C$.}
\label{fig:constraint-rules}
\end{figure}

Next, we prove the constraint rules are sound and complete with respect to the trace rules, i.e., the constraint set generated by the derivation of an expression (command, function, resp.) under the constraint rules is solvable, if and only if an expression (command, function, resp.) is typable under the trace rules.

\begin{lemma}\label{lem:cgrsound}
The following statements hold:
\begin{enumerate}[label={(\alph*)}]
\item\label{lem:cgrsound-1} If $\Gamma;\trace \vdash_{g} e : t\leadsto C$ and $\theta \vDash C$, then $\Gamma\theta;\trace \vdash_{t} e : t\theta $.

\item\label{lem:cgrsound-2} If $\Gamma;\trace; A \vdash_{g} c : t\leadsto C$ and $\theta \vDash C$, 
then  $\Gamma\theta;\trace;A \vdash_{t} c : t\theta.$

\item\label{lem:cgrsound-3} If $ \vdash_{g} \lpPROG{B.f}{x}{c}{\lpRES} :  \overline{\alpha}\rightarrow \beta\leadsto C$ and $\theta \vDash C$, then 
\[
\vdash_{t} \lpPROG{B.f}{x}{c}{\lpRES} : \overline{\theta(\alpha)}\xrightarrow{}\theta(\beta).
\]

\end{enumerate}
\end{lemma}

\begin{lemma}\label{lem:cgrcomplete}
The following statements hold:
\begin{enumerate}[label={(\alph*)}]

\item\label{lem:cgrcomplete-1} If $\Gamma;\trace \vdash_{t} e : t$, then there exist $\Gamma',t',C,\theta$ s.t. $\Gamma';\trace \vdash_{g} e : t' \leadsto C$,  $\theta \vDash C$, $\Gamma'\theta =\Gamma$ and $t'\theta = t$.

\item\label{lem:cgrcomplete-2} If $\Gamma;\trace; A \vdash_{t} c : t$, then there exist $\Gamma',t',C,\theta$ s.t. $\Gamma';\trace; A \vdash_{g} c : t' \leadsto C$,  $\theta \vDash C$, $\Gamma'\theta =\Gamma$ and $t'\theta = t$.

\item\label{lem:cgrcomplete-3} If $\vdash_{t} \lpPROG{B.f}{\overline{x}}{c}{\lpRES} : \overline{t_p}\xrightarrow{} t_r$, then there exist $\alpha,\beta,C,\theta$ s.t.
$$
\vdash_{g} \lpPROG{B.f}{\overline{x}}{c}{\lpRES}: \overline{\alpha}\xrightarrow{} \beta \leadsto C,
$$ 
$\theta \vDash C$, and $(\overline{\alpha}\xrightarrow{} \beta)\theta = \overline{t_p}\xrightarrow{} t_r$, where $\alpha,\beta$ are fresh type variables.
\end{enumerate}
\end{lemma}
 \subsection{Constraint Solving}\label{sec:constraint_solve}
This section presents a constraint solving algorithm for the constraints generated by the rules in Fig.~\ref{fig:constraint-rules}.
For these constraints, both types appearing on the two sides of subtyping are 
guarded by the same permission trace, which is the one collected from the current context.
While during the solving of these constraints, there may be some constraints on whose two sides the types are guarded by different traces. 
Take the constraint \cRawCond{\trace}{\app{t_l}{P}}{\app{\alpha}{Q}} as an example, $t_l$ is indeed guarded by $P$ while $\alpha$ is guarded by $Q$, where $P$ and $Q$ are different permission sets. 
Therefore,  we use a generalized version of the permission guarded constraint, allowing types on the two sides to be guarded by \emph{different} permission traces:
$$
\cCond{\trace_l}{t_l}{\trace_r}{t_r}
$$
Similarly, a \emph{solution }to a generalized constraint set $C$ is a substitution $\theta$, denoted as $\theta \vDash C $,  such that for each
$\cCond{\trace_l}{t_l}{\trace_r}{t_r} \in C$,
$(t_l\theta\cdot \trace_l) \leq_{} (t_r\theta \cdot \trace_r)$
holds.

It is easy to transfer a permission guarded constraint set $C$ into a generalized constraint set $C'$: by rewriting each \cRawCond{\trace}{t_l}{t_r} as \cCond{\trace}{t_l}{\trace}{t_r}.
Moreover, it is trivial that $\theta \vDash C \Longleftrightarrow \theta \vDash C'$.
Therefore, we will focus on the solving of generalized constraints in the following.

Given a permission set $P$ and a permission trace $\trace$, we say $P$ \termEmph{entails} $\trace$, denoted as $P \vDash \trace$, iff. $\forall \oplus p \in \trace. \; p\in P$ and $\forall \ominus p\in\trace.\; p\notin P$.
A permission trace $\trace$ is \termEmph{satisfiable}, denoted as $\issatisfied{\trace}$, iff. there exists a permission set $P$ such that $p\vDash \trace$. 

The constraint solving consists of the following steps: 
\begin{itemize}
\item \textbf{decompose }the types in constraints into type variables and ground types;
\item \textbf{saturate} the constraint set by the transitivity of the subtyping relation;
\item \textbf{merge} the lower and upper bounds of same variables;
\item \textbf{unify} the constraints to emit a solution.
\end{itemize}

\subsubsection{Decomposition}

The first step to solve the guarded constraints is to decompose the types into the simpler ones, i.e., type variables and ground types, according to their structures.
The decomposing rules are given in Fig. \ref{fig:solving}.
Rules (\ruleTagText{CD-CUP}), (\ruleTagText{CD-CAP}) and (\ruleTagText{CD-SVAR}) are trivial. 
Rule (\ruleTagText{CD-MEGER}) states that two $p$-merged types satisfy the relation if and only if both their $p$-promotions and $p$-demotions satisfy the relation, where $t$ can be viewed as a $p$-merged type $\omerge{p}{t}{t}$. 
The projection of types yields a ``monomorphic type'' such that any successive trace application makes no changes, therefore we have (\ruleTagText{CD-LAPP}) and (\ruleTagText{CD-RAPP}).
Rules (\ruleTagText{CD-SUB$_0$}) and (\ruleTagText{CD-SUB$_1$}) handle the constraints on ground types.

After decomposition, constraints have one of the forms:
\begin{equation*}
\cCond{\trace}{\alpha}{\trace'}{t_g}, \cCond{\trace}{t_g}{\trace'}{\beta}, \cCond{\trace}{\alpha}{\trace'}{\beta}
\end{equation*}

\begin{figure*}

\[
\inference[CD-CUP]
{}
{C\cup\{\cCond{\trace}{t_1\sqcup t_2}{\trace'}{t}\}
\leadsto_{d}
{C\cup\{\cCond{\trace}{t_1}{\trace'}{t}, \cCond{\trace}{t_2}{\trace'}{t}\}}
}
\]

\[
\inference[CD-CAP]
{}
{C\cup\{\cCond{\trace}{t}{\trace}{t_1\sqcap t_2}\} \leadsto_{d}
{C\cup\{\cCond{\trace}{t}{\trace'}{t_1}, \cCond{\trace}{t}{\trace'}{t_2}\}}
}
\]

\[
\inference[CD-LAPP]
{}
{C\cup\{\cCond{\trace}{\app{t}{P}}{\trace'}{t'}\} \leadsto_{d}
{C\cup\{\cCond{\trace_{P}}{t}{\trace'}{t'}\}}
}
\]

\[
\inference[CD-RAPP]
{}
{C\cup\{\cCond{\trace}{t}{\trace'}{\app{\alpha}{P}}\} \leadsto_{d}
{C\cup\{\cCond{\trace}{t}{\trace'_{P}}{\alpha}\}}
}
\]

\[
\inference[CD-SVAR]
{}
{C\cup\{\cCond{\trace}{\alpha}{\trace}{\alpha}\}\leadsto_{d} C}
\]

\[
\inference[CD-SUB$_0$]
{t_g\cdot\trace \leq_{} s_g\cdot\trace'}
{C\cup\{\cCond{\trace}{t_g}{\trace'}{s_g}\} \leadsto_{d} C}
\quad
\inference[CD-SUB$_1$]
{t_g\cdot\trace \nleq_{} s_g \cdot\trace'}
{C\cup\{\cCond{\trace}{t_g}{\trace'}{s_g}\} \leadsto_{d} \bot }
\]

\[
\inference[CD-MERGE]
{C' = \{\cCond{\trace::\oplus p}{t}{\trace'::\oplus p}{t_1}, \cCond{\trace::\ominus p}{t}{\trace'::\ominus p}{t_2}\}}
{C\cup\{\cCond{\trace}{t}{\trace'}{\omerge{p}{t _1}{t_2}}\}\leadsto_{d} C\cup C'}
\]

\[
\inference[CS-LU]
{
\trace'_l = (\trace_l\land \trace_r)  - \trace_r &
\trace'_r =(\trace_l\land \trace_r) - \trace_l \\
\issatisfied{\trace_l \land \trace_r} &
\{(\trace_1\land \trace'_r, t_1\leq \trace_2\land \trace'_l,t_2) \}\leadsto_{d} C'
}
{\{\cCond{\trace_1}{t_1}{\trace_r}{\alpha},\cCond{\trace_l}{\alpha}{\trace_2}{t_2}\} \subseteq C
\leadsto_{s} C \cup C'}
\]

\[
\inference[CM-GLB]{
\phi(I') =
\{\trace \in \dnf{\bigwedge_{i\in I'}\trace_{ir}\wedge\bigwedge_{i\in I\setminus I'}\neg \trace_{ir}}~|~ \issatisfied{\trace}\} &
t_{I',\trace}^{\sqcup}  =  \sqcup_{i\in I'} (t_{i}\cdot (\trace_{il}\wedge (\trace - \trace_{ir})))
}
{
C\cup\{\cCond{\trace_{il}}{t_i}{\trace_{ir}}{\alpha}\}_{i\in I} 
\leadsto_{m1}
C\cup\{\cCond{\epsilon}{t_{I',\trace}^{\sqcup}}{\trace}{\alpha}\}_{I'\subseteq I, \trace \in \phi(I') }
}
\]

\[
\inference[CM-LUB]{
\phi(I') =
\{\trace \in \dnf{\bigwedge_{i\in I'}\trace_{il}\wedge\bigwedge_{i\in I\setminus I'}\neg \trace_{il}}~|~ \issatisfied{\trace}\} &
t_{I',\trace}^{\sqcap} =  \sqcap_{i\in I'} (t_{i}\cdot (\trace_{ir}\wedge (\trace - \trace_{il})))
}
{
C\cup\{\cCond{\trace_{il}}{\alpha}{\trace_{ir}}{t_i}\}_{i\in I} 
\leadsto_{m1}
C\cup\{\cCond{\trace}{\alpha}{\epsilon}{t_{I',\trace}^{\sqcap}}\}_{I'\subseteq I, \trace\in \phi(I')}
}
\]

\[
\inference[CM-BDS]
{
C_{i,j}^{1} = \{(t_i \cdot (\trace -\trace_i) \leq (\trace, \alpha) \leq t_j \cdot (\trace -\trace_j) ~|~ \trace \in \dnf{\trace_i \land \trace_j} \text{ and } \issatisfied{\trace} \} \\
C_{i,j}^{2} = \{(t_i \cdot (\trace -\trace_i) \leq (\trace, \alpha) \leq H ~|~ \trace \in \dnf{\trace_i \land \neg \trace_j} \text{ and } \issatisfied{\trace}\} \\
C_{i,j}^{3} = \{(L \leq (\trace, \alpha) \leq t_j \cdot (\trace -\trace_j) ~|~ \trace \in \dnf{\trace_j \land \neg \trace_j} \text{ and } \issatisfied{\trace} \} 
}
{
C\cup\{\cCond{\epsilon}{t_i}{\trace_i}{\alpha}_{i\in I}, \cCond{\trace_j}{\alpha}{\epsilon}{t_j}_{j\in J} \} \leadsto_{m2} C\cup \bigcup_{ i\in I,j \in J} (C_{i,j}^{1} \cup C_{i,j}^{2} \cup C_{i,j}^{3})
}
\]

\[
\inference[CM-SBD]
{}
{C\cup\{( t_1 \leq (\trace,\alpha) \leq s_1), (t_2\leq (\trace, \alpha)  \leq s_2) \} \leadsto_{m3} 
C\cup\{(t_1\sqcup t_2 \leq (\trace, \alpha) \leq s_1\sqcap s_2)\}}
\]

\caption{Constraint solving rules, including \termEmph{decomposition}~(\textbf{CD-}), \termEmph{saturation}~(\textbf{CS-}) and \termEmph{merging}~(\textbf{CM-}), where $\trace_{P}$ denotes the trace that only $P$ can entail, $-$ denotes set difference, $ t_{\emptyset,\trace} ^{\sqcup} = L$, and $ t_{\emptyset,\trace}^{\sqcap}  = H$. }
\label{fig:solving}
\end{figure*}

\subsubsection{Saturation}
We treat each permission trace $\trace$ as a boolean logic formula on permissions, where $\oplus$ and $\ominus$ denote positive and negative respectively. 
In the remaining we shall use the logic connectives on permission traces freely. 
We also adopt the disjunctive normal form, i.e., a disjunction of conjunctive permissions, and denote it as $\dnf{\cdot}$.
For example, $\dnf{(\oplus p) \land \neg (\oplus q \land \ominus r)}  = (\oplus p \land \ominus q) \lor (\oplus p \land \oplus r) $.

To ensure any lower bound (e.g., \cCond{\trace_1}{t_1}{\trace_r}{\alpha}) of a variable $\alpha$ is ``smaller'' than any of its upper bound (e.g., \cCond{\trace_l}{\alpha}{\trace_2}{t_2}, we need to saturate the constraint set by adding these conditions.
However, since our constraints are guarded by permission traces, we need to consider lower-upper bound relations only when the traces of the variable $\alpha$ can be entailed by the same permission set, i.e., their intersection is satisfiable (i.e., $\issatisfied{\trace_l \land \trace_r} $).
Then we extend the traces of both the lower and upper bound constraints such that the traces of $\alpha$ are the same (\emph{i.e.}, $\trace_l \land \trace_r$), by adding the missing traces (\emph{i.e.}, $\trace_l \land \trace_r - \trace_r$ for lower bound constraint while $\trace_l \land \trace_r - \trace_l$ for the upper one, where $-$ denotes set difference).
The rule is given as (\ruleTagText{CS-LU}) in Fig. \ref{fig:solving}.

Assume that there is an order $<$ on type variables and the smaller variable has a higher priority. 
If two variables $\alpha, \beta$ with $O(\alpha) < O (\beta)$ (the orderings) are in the same constraint $\beta \leq \alpha$, we consider the larger variable $\beta$ is a bound for the lower one $\alpha$, but not vice-versa.
There is a special case of \cCond{\trace}{\alpha}{\trace'}{\alpha}.
In that case, we regroup all the trace of $\alpha$ as $\{\trace_i~|~i \in I\}$ such that $\bigvee_{i\in I}\trace_i  = \epsilon$ and  $\forall i,j\in I. i\neq j \Rightarrow \neg\issatisfied{\trace_i\land\trace_j}$, and rewrite the constraints of $\alpha$ w.r.t. $\{\trace_i~|~i \in I\}$.  Then we treat each $(\trace_i, \alpha)$ as different fresh variables $\alpha_i$.
Therefore, with the ordering, there are no loops like: $\cPair{\trace}{\alpha} \leq \ldots\leq \cPair{\trace'}{\alpha}$. 

\subsubsection{Merging}

Next, we would like to merge the constraints on an identical variable.
As constraints are guarded by permission traces, 
we need to consider the satisfiability of (any subset of) the permission traces of the same variable under any permission set.
The merging rules are presented in Fig. \ref{fig:solving} as well.
Rule (\ruleTagText{CM-GLB}) handles the lower bounds of an identical type variable.
Let us consider the lower bounds $\{\cCond{\trace_i}{t_i}{\trace_{ir}}{\alpha}~|~i\in I\}$ of a type variable $\alpha$, guarded by the possibly different permission traces $\trace_{ir}$.
Assume that only the traces in $I' \subseteq I$ can be entailed by a permission set $P$ simultaneously, that is, the common trace $\dnf{\bigwedge_{i\in I'}\trace_{ir}\wedge\bigwedge_{i\in I\setminus I'}\neg \trace_{ir}}$ is satisfied.
We adjust each constraint in $I'$ by extending the traces with the missing one, which is the common trace minus the trace of $\alpha$.
After that, the variable $\alpha$ in all the constraints of $I'$ are guarded by an identical trace.
Then it is easy for us to deduce a greatest lower bound for $\alpha$, that is, $\sqcup_{i\in I'}t_i\cdot\trace'_i$,  where $\trace'_i$ is the extended trace of $t_i$.
Similar to Rule (\ruleTagText{CM-GLB}), Rule (\ruleTagText{CM-LUB}) handles the upper bounds of an identical type variable, yielding a least upper bound.
Rule (\ruleTagText{CM-BDS}) combines the greatest lower bound and the least upper bound of an identical type variable, which also needs to consider the satisfiability of the possibly different traces.
Finally, Rule (\ruleTagText{CM-SBD}) merges the possible redundant bounds under the same trace.
Note that, the traces of the bounds are $\epsilon$ and thus can be omitted.

After merging, for each type variable $\alpha$, its constraints $\{(t_i\leq(\trace_i, \alpha)\leq s_i)\}_{i\in I} $ (if exists) satisfy that $(1)$ the union of all $\trace_i$ is the full set (\emph{i.e.}, $\bigvee_{i\in I}\trace_i  = \epsilon$) and ($2$) the intersection of different traces is unsatisfiable (\emph{i.e.}, $\forall i,j\in I. i\neq j \Rightarrow \neg\issatisfied{\trace_i\land\trace_j}$).

Given two constraint sets $C_1,C_2$, we say $C_1$ entails $C_2$, denoted as $C_1 \vDash C_2$, iff. for any substitution $\theta$, if $\theta \vDash C_1 $, then $\theta \vDash C_2$.
We proved that the constraint solving rules are sound and complete, i.e., 
the original constraint set entails the converted set (obtained by decomposition, saturation and merging), and vice-versa.

\iffalse
\begin{lemma}\label{lem:tracesub1}
Given two types $s,t$ and a permission trace $\trace$, then
$s\leq t\Longleftrightarrow s\cdot\trace \leq t\cdot\trace$ and $\forall \trace' \in dnf(\neg\trace).\; s\cdot \trace' \leq t\cdot\trace'$. 
\end{lemma}
\fi

\begin{lemma}\label{lem:consolcorrect}
If $C\leadsto_{r} C'$, then $C\vDash C'$ and $C' \vDash C$, where $r\in \{d,s,m\}$. 
\end{lemma}

\subsubsection{Unification}
Now the constraints to be solved are of the form $ \{(t_i\leq(\trace_i, \alpha)\leq s_i)\}_{i\in I}$,  
stating that $\alpha$ can take a type that ranges from $t_i $ to $s_i$ under the permission trace $\trace_i$. 
Let $\alpha_i$ denote such a type, where $\alpha_i$ is a fresh type variable.
So the constraint $(t_i\leq(\trace_i, \alpha)\leq s_i)$ can be rewritten as an equivalent equation $(\trace_i, \alpha)= (t_i \sqcup \alpha_i )\sqcap s_i$. 

We define an equation system $E$ as a set of equations $(\trace, \alpha)= t$, and a solution to $E$ is a substitution $\theta$, denoted as $\theta\vDash E$, s.t. for each $ (\trace, \alpha)= t \in E$, $(\theta(\alpha)\cdot \trace) = t\theta$. 
It is trivial that  a constraint set $C$ can be equivalently transformed into an equation system $E$ s.t. $\theta\vDash C \Longleftrightarrow \theta\vDash E$.

As mentioned above,  the trace set $ \{\trace_i~|~i\in I\}$ of an identical variable is full (\emph{i.e.}, $\bigvee_{i\in I}\trace_i  = \epsilon$) and disjoint (\emph{i.e.}, $\forall i,j\in I. i\neq j \Rightarrow \neg\issatisfied{\trace_i\land\trace_j}$).  
For each permission set $P$, there exists a unique trace $\trace_i$ such that $P\vDash \trace_i$. 
Hence, a type $t$ for $\alpha$ can be constructed:
$$
t(P) = ((t_i \sqcup \alpha_i )\sqcap s_i)(P) \quad P\vDash \trace_i
$$

The algorithm below presents the unification algorithm that solves the equation system $E$.
It is provable that the unification algorithm is sound and complete.

\vspace{10pt}
\begin{algorithmic}
\STATE \textbf{let} toType~$\{\cPair{\trace_i}{t_i}\}_{i\in I} = t$ s.t. $t(P) = t_i$ if $P\vDash \trace_i$ \textbf{in}
\STATE \textbf{let} $subst ~\theta~(\cPair{\trace_j}{\beta} = s_j)~ = (\cPair{\trace_j}{\beta} = s_j\theta)$ \textbf{in}
\STATE \textbf{let~rec} $equ2sub~E $ =
\STATE \quad\textbf{match} $E$ \textbf{with}
\STATE \quad$|~[]\rightarrow []$
\STATE \quad$|~\{\cPair{\trace_i}{\alpha} = t_i\}_{i\in I}:: E'$  \textbf{where} $isMax(O(\alpha))$ $\rightarrow$
\STATE \quad\quad\textbf{let} $t_{\alpha} = $ toType~$\{(\trace_i, t_i)\}_{i\in I}$ \textbf{in}
\STATE \quad\quad\textbf{let} $E'' = List.map~ (subst~[\alpha \mapsto t_{\alpha}])~ E'$ \textbf{in}
\STATE \quad\quad \textbf{let} $\theta' = equ2sub~E'' $ \textbf{in}
\STATE \quad\quad $\theta'[\alpha \mapsto t_{\alpha}]$
\STATE \textbf{in} $equ2sub~E$
\end{algorithmic}
\vspace{10pt}

\begin{lemma}\label{lem:unifysound}
If $unify(E) = \theta$, then $\theta \vDash E$.
\end{lemma}

\begin{lemma}\label{lem:unifycomplete}
If $\theta \vDash E$, then there exist $\theta'$ and $\theta''$ such that $unify(E) = \theta'$ and $\theta = \theta'\theta''$.
\end{lemma}

To conclude, an expression (command, function, resp.) is typable, if and only if it is derivable under the constraint rules with a solvable constraint set by our algorithm.
Therefore, our type inference system is sound and complete.
Moreover, as the function call chains are finite, the constraint generation terminates with a finite constraint set, which can be solved by our algorithm in finite steps. 
Thus, our type inference system terminates.

\begin{theorem}
The type inference system is sound, complete and decidable.
\end{theorem}
 \section{Examples}\label{sec:cases}

This section presents two examples. 
The first one is to demonstrate how our inference system is performed and 
the second one is the motivating example in Sec.~\ref{sec:intro}. The prototype implementation of our inference tool and additional results can be found at~\url{http://bit.ly/2f3yVL0}.

\subsection{Illustrative Example}
Consider the following function:
\begin{lstlisting}
f() {
   init r=0 in                 
     test(p) r = info(p);      
     else r = 0;               
     test(q) r = r + info(q);  
     else r = r+ 0;            
   return r            
}
\end{lstlisting}
Suppose that the type variable assigned to $\lpRES$ is $\alpha$ and thus the type for $f$ is $()\rightarrow \alpha$. 
Let us apply the constraint rules in Fig.~\ref{fig:constraint-rules} on each command, yielding the constraint set: 
$
\{ 
(\oplus p, l_p \leq \alpha),
(\ominus p, L \leq \alpha),
(\oplus q, \alpha \sqcup l_q \leq \alpha), 
(\ominus q, \alpha \sqcup L \leq \alpha)
\}
$.
By transforming each constraint into the generalized form, we obtain the following set
\begin{gather*}
\{
(\oplus p,  l_p \leq \oplus p, \alpha),
(\ominus p,  L \leq \ominus p, \alpha),\\
(\oplus q, \alpha \sqcup l_q \leq \oplus q,\alpha), 
(\ominus q, \alpha \sqcup L \leq \ominus q, \alpha)
\}
\end{gather*}
Next, we perform the constraint solving algorithm on the set above as follows.
\begin{description}
\item[Decomposition]
	The last two constraints are not in their simple forms therefore Rule (\ruleTagText{CD-CUP}) is applied, yielding: 
\begin{gather*}
\{
(\oplus p,  l_p \leq \oplus p, \alpha),
(\ominus p, L \leq \ominus p, \alpha),\\
(\oplus q, l_q \leq \oplus q,\alpha), 
(\ominus q, L \leq \ominus q, \alpha)
 \}
\end{gather*}
\item [Saturation]
    Since there is only one type variable $\alpha$ and there are only the lower bounds for $\alpha$, the constraint set remains unchanged after saturation.
    \item[Merging]
    There are several lower bounds for $\alpha$ under different traces. Hence, we need to consider the satisfiability of the combinations of permission traces. There are four possible cases: (1) $\oplus p\oplus q$; (2) $\oplus p\ominus q$; (3) $\ominus p\oplus q$; (4) $\ominus p\ominus q$. For each case, only the constraints that are satisfied with it are considered. For example, since $\ominus p$ and $\ominus q$ are not compatible to $\oplus p \oplus q$, only $(\oplus p,  l_p \leq \oplus p, \alpha)$ and $(\oplus q, l_q \leq \oplus q,\alpha)$ are considered. The merging result is:
\begin{gather*}
\{
(l_{p}\sqcup l_{q} \leq (\oplus p \oplus q,  \alpha) ),
(l_{p} \leq (\oplus p \ominus q,\alpha)),\\
( l_{q} \leq (\ominus p \oplus q,\alpha)), 
(L \leq (\ominus p\ominus q, \alpha)) 
\}
\end{gather*}
\item[Unification]
    By picking the least upper bound, we obtain a feasible type $t$ for $\alpha$:
\begin{equation*}
 t = \{
(\oplus p \oplus q, l_{p} \sqcup l_{q}),
(\oplus p \ominus q,  l_{p} ),
(\ominus p \oplus q,  l_{q} ), 
(\ominus p\ominus q, L ) \}
\end{equation*} 
Focusing on permissions $p$ and $q$, $t$ can be rewritten as
\begin{equation*}
t=\big\{\emptyset \rightarrow \emptyset,
\{p\}\rightarrow \{p\},
\{q\}\rightarrow \{q\},
\{p,q\}\rightarrow \{p,q\}
\big\}
\end{equation*}
Therefore, the type of $f$ is $()\rightarrow t$.
\end{description}

\subsection{Retrospection on the Motivating Example}
Recall the function \texttt{getInfo} in Sec.~\ref{sec:intro}~\footnote{Note that the program syntax is sightly different from the one in Sec.~\ref{sec:language} but can be adapted into the latter easily. The empty string, like the integer literal 0 in the language, is considered $L$. Our prototype implementation does support primitive data types like integer, float, and string.}.
Assume the  type of {\lpRES} is $\alpha$, and thus the type for \texttt{getInfo} is $()\rightarrow\alpha$.
According to the constraint rules, we obtain the following constraint set  
$
\{(\oplus p\oplus q, l_1 \leq \alpha),
(\oplus p\ominus q, L  \leq \alpha),
(\ominus p\oplus q, l_1 \sqcup l_2 \leq \alpha), 
(\ominus p \ominus q, L \leq \alpha)
\}
$, where $l_1$ and $l_2$ are the security levels for $loc$ and $id$ respectively, and $l_1$ and $l_2$  are incomparable.
Following constraint solving steps, we get the inferred type
$
t=\big\{
(\oplus p \oplus q, l_1 ),
(\oplus p \ominus q, L ),
(\ominus p \oplus q, H ), 
(\ominus p\ominus q, L )
\big\}
$.
 \section{Related work} 
There is a large body of work on language-based information flow security. 
We shall discuss here only closely related work.

We have discussed the work by Banerjee and Naumann~\cite{Banerjee:2005ht} (in Sec.~\ref{sec:intro}) from which we base the design of our type system on. Our work differs significantly from their in three major aspects. The fundamental change is the merging type constructor that extends the type system, which are crucial to make our type system more precise than BN. Another difference with {\BN} type system is in the permission model. In Android apps, permissions are not associated with Java classes, so there is no need to take into account stack inspection in the type system.
Permissions are enforced only when doing remote procedure calls
and are only aware of the permissions of its immediate caller (app $B$). Our typing rule (\ruleTagText{T-CALL}) captures this feature. We also provide a decidable type inference algorithm, which was not given in prior relevant work. This is necessary since otherwise it is impractical to annotate each type explicitly.

Value-dependent type information flow system provides a general treatment of security types that may depend on other program variables~\cite{Toby2016csf,Dependent_SIFUM_Type_Systems-AFP,Dependent_SIFUM_Refinement-AFP,Murray:2015jm,Li:2016jb,Zhang:2015bc,Li:2015vw,Lourenco:2015:DIF,Lourenco13,Swamy13,Nanevski11,SwamyCC10,Zheng:2007cd,Tse07}. 
For example, Murray et al~\cite{Toby2016csf,Dependent_SIFUM_Type_Systems-AFP,Dependent_SIFUM_Refinement-AFP,Murray:2015jm} provided a value-dependent timing-sensitive noninterference for concurrent shared memory programs under compositional refinement.
We think, however, the closest to our work is perhaps the dependent type system of Louren\c{c}o and Caires~\cite{Lourenco:2015:DIF}.
In their type system, security labels may be indexed by data structures, which can be used to encode the dependency of the security labels on other values in the system. 
It may be possible to encode our notion of security types as a dependent type in their setting, by treating permission sets explicitly as an additional parameter to
a function or a service, and to specify security levels of the output of the function as a type dependent on that parameter.
Currently it is not yet clear to us how one could give a general construction of the index types in their type system that would correspond to our security types, and
how the merge operator would translate to their dependent type constructors, among other things. Such a correspondence would be of interests since it may help
to clarify some problems we encountered, such as the parameter laundering attack, to see whether they can be reduced to some well-understood
problems in the literature of noninterference type systems. We leave the exact correspondence to the future work.

Recent research on information flow has also been conducted to deal 
with mobile security issues~\cite{Ernst:2014,Chin:2011wa,Nadkarni:2016tf,Lortz:2014ku,Gunadi:2015,Chaudhuri:2009ii,Fuchs2010}, especially on Android platform. SCandroid~\cite{Chaudhuri:2009ii,Fuchs2010} is a tool targeting for automated security certification of Android apps, which focuses on typing communication between applications. Unlike our work, they do not consider implicit flows, and do not take into account access control in their type system. Ernst et al~\cite{Ernst:2014} proposed a verification model, SPARTA, for use in app stores to guarantee that apps are free of malicious information flows. Their approach, however, requires the collaboration between software vendor and app store auditor and the additional modification of Android permission model to fit for their Information Flow Type-checker (IFT); soundness proof is also absent. Our work is done in the context of providing information flow security ``certificate'' for Android applications, following the Proof-Carrying-Code (PCC)
architecture by Necula and Lee~\cite{necula:pcc:1996} and does not require extra changes.
 noninterference
\section{Conclusion and Future Work}\label{sec:conclusion}

We provided a lightweight type system featuring Android permission model for enforcing secure information flow in an
imperative language and proved its soundness with respect to
noninterference. The established system has a novel property that its
security type can express arbitrary relation between permissions and
security level, and therefore may be
applied in a wider range of security policies such as those involving
aggregation problems. We also developed a type inference algorithm by
reducing it to a constraint solving problem and applying an efficient
algorithm to solve the latter.

We next discuss briefly several directions for future work. 

\paragraph{Global variables}
The language presented in this article does not allow global variables. 
In our ongoing work, we have shown that global variables can be integrated seamlessly
in our framework. 
The introduction of global variables presents a potential side channel through which information
could leak, via what we call ``global variable laundering''.  
That is, information may flow through global variables when they are \emph{written} and \emph{read} by apps 
with different permission contexts. 
Our countermeasure is to forbid global variables to have polymorphic types (i.e., each of the types is of the form~$\hat{l}$ given a security level $l$), and we introduce another ``global variable security type'' in typing rules for commands to track the function body effect. 
We are presently working on the type inference on the extended type system. We also plan to extend our type system to include more language features such as heaps, object-oriented extension, and exceptions.

\paragraph{Runtime permission granting}
If the permissions granting for an app are non-deterministic, such as the case when the permission granting is subject to the user's approval, then enforcing a non-monotonic policies becomes problematic. In Android, certain permissions are classified as dangerous permissions, and are further partitioned into permission groups. 
These permissions are not automatically given to apps, but are instead subject to user's approval at runtime when the permissions are exercised. However, an app can only request for permission it has explicitly declared in the manifest file, so to this extent, we can statically determine whether a permission request is definitely not going to be granted (because it is absent from the manifest), and whether it can {\em potentially} be granted. So it is still possible to enforce monotonic policies, under the assumption that all permissions in the manifest file are always granted.  One possible way to accommodate limited non-monotonic policies would be to separate {\em dynamic permissions} (i.e., those that require runtime approval) and {\em static permissions}  (any permissions declared in the manifest file that are not dynamic 
permissions), and to allow policies to be non-monotonic on static permissions only. We leave the detailed treatment of this refinement to future work.

\paragraph{Bytecode typing} Our eventual goal is to translate source code typing into Dalvik bytecode typing, following a similar approach done by Gilles Barthe et al~\cite{Barthe:2005ju,Barthe:2006jh,Barthe:2007fxa} from Java source to JVM bytecode. The key idea that we describe in the paper, i.e., precise characterizations of security of IPC channels that depends on permission checks, can be applied to richer type systems such as such as that used in the Cassandra project~\cite{Lortz:2014ku} or Gunadi et. al. type system~\cite{Gunadi:2015}. We envision our implementation can piggyback on, say, Cassandra system to improve the coverage of typable applications.
 
% \bibliography{ref,impl}
%%%%%%%%%%%%%%%%%%%%%%%%%

%%%%%%%%%%%%%%%%%%%%%%%%%

\clearpage
\appendix

\section{Soundness of Type System}\label{sec:app-soundness}

\begin{definition}\label{app-def:base_type}
A \termEmph{base security type} (or \termEmph{base type}) $t$ is a mapping from $\PPS$ to $\SL$.
We denote with $\Tcal$ the set of base types. 
Given two base types $s,t$, we say $s=t$ \text{iff} $s(P)=t(P)$ $\forall$ $P\in\PPS$ and  
$s\leq_\Tcal t$ \text{iff} $\forall$ $P\in \PPS$, $s(P) \leq_{\SL} t(P)$.
\end{definition}

\begin{lemma}\label{app-lem:tleq_po}
	$\leq_\Tcal$ is a partial order relation on $\Tcal$.
\end{lemma}
\begin{proof}
\begin{ProofEnumDesc}[style=standard]
	\item[\textbf{Reflectivity}] $\forall P, t(P)\leq_{\SL} t(P)$ therefore $t\leq_\Tcal t$.
	\item[\textbf{Antisymmetry}] $t\leq_\Tcal s\wedge s\leq_\Tcal t\iff \forall P, t(P)\leq_{\SL} s(P) 
\wedge s(P) \leq_{\SL} t(P)$ therefore $\forall P, t(P)=s(P)$, which means that $s=t$.
	\item[\textbf{Transitivity}] if $r\leq_\Tcal s$ and $s\leq_\Tcal t$, $\forall P, r(P) \leq_{\SL} s(P)$ and $s(P)\leq_{\SL} t(P)$ therefore $r(P) \leq_{\SL} t(P)$, which means that $r \leq_\Tcal t$.
\end{ProofEnumDesc}
\end{proof}

\begin{definition}\label{app-def:type-cup-cap}
For $s, t\in\Tcal$, $s\sqcup t$ and $s\sqcap t$ are defined as
\begin{align*}
(s\sqcup t)(P) =  s(P)\sqcup t(P), \forall P \in \PPS\\
(s\sqcap t)(P) =  s(P)\sqcap t(P), \forall P \in \PPS
\end{align*}
\end{definition}

\begin{lemma}\label{app-lem:bound}
Given two base types $s$ and $t$, it follows that
\begin{enumerate}[label*=(\alph*)]
\item $s \leq_\Tcal s\sqcup t$ and $t\leq_\Tcal s\sqcup t$.
\item $s\sqcap t \leq_\Tcal s$ and $s\sqcap t \leq_\Tcal t$.
\end{enumerate}
\end{lemma}
 \begin{proof}
Immediately from Definition~\ref{app-def:base_type}.
 \end{proof}

\begin{lemma}\label{app-lem:lattice}
$(\Tcal, \leq_\Tcal)$ forms a lattice.
\end{lemma}
 \begin{proof}
 $\forall s,t\in\PPS$, according to Lemma \ref{app-lem:bound}, $s \sqcup t$ is their upper bound. Suppose $r$ is another upper bound of them, i.e., $s\leq_\Tcal r$ and $t\leq_\Tcal r$, which means $\forall P\in\PPS, (s\sqcup t)(P)=s(P)\sqcup t(P)\leq_{\SL} r(P)$, so $s\sqcup t\leq r$. Therefore $s\sqcup t$ is the least upper bound of $\left\{s,t\right\}$. Similarly, $s\sqcap t$ is $s$ and $t$'s greatest lower bound. This makes $(\Tcal, \leq_\Tcal)$ a lattice.
 \end{proof}

\begin{definition}\label{app-def:fun-type}
   A \termEmph{function type} has the form $\overline{t} \rightarrow t$, where $\overline{t}=(t_1,\ldots,t_m)$, $m\geq 0$ and $t$ are base types.
   The types $\overline t$ are the types for the arguments of the function and the base type $t$ is the return type of the function.
 \end{definition}
 
 \begin{definition}\label{app-def:updown}
Given a permission $p$, the \termEmph{promotion} and \termEmph{demotion} of a base type $t$
with respect to $p$ are defined as follows:
\begin{align*}
(t\uparrow_{p}) (P) = t(P\cup \{p\}), \forall P\in \PPS
\tag{promotion}
\\
(t\downarrow_{p}) (P) = t(P\setminus\{p\}), \forall P\in \PPS
\tag{demotion}
\end{align*}
The \termEmph{promotion} and \termEmph{demotion} of a function type $\overline{t}\rightarrow t$, where $\overline{t} = (t_1,\dots, t_m)$, are respectively 
\begin{align*}
(\overline{t}\rightarrow t)\uparrow_{p} = \overline{t\uparrow_{p}}\rightarrow t\uparrow_{p},~
where~\overline{t\uparrow_{p}}=(t_1\uparrow_{p},\ldots, t_m\uparrow_{p}),
\\
(\overline{t}\rightarrow t)\downarrow_{p} = \overline{t\downarrow_{p}}\rightarrow t\downarrow_{p},~
where~\overline{t\downarrow_{p}}=(t_1\downarrow_{p},\ldots, t_m\downarrow_{p}).
\end{align*}
\end{definition}

\begin{lemma}\label{app-lem:promote-demote}
Given $P\in\PPS$ and $p\in\PS$,
\begin{enumerate}[label*={(\alph*)}]
  \item\label{app-lem:promote-demote-1} If $p \in P$, then $(t\uparrow_{p})(P) = t(P)$.
	\item\label{app-lem:promote-demote-2} If $p \notin P$, then $(t\downarrow_{p})(P) = t(P)$.
\end{enumerate}
\end{lemma}
 \begin{proof}
If $p\in P$, $P\cup\{p\}=P$, therefore $(t\uparrow_{p}) (P) = t(P\cup \{p\})=t(P)$; if $p\notin P$, $P\setminus\{p\}=P$, therefore $(t\downarrow_{p}) (P) = t(P\setminus\{p\})=t(P)$.
 \end{proof}

\begin{lemma}\label{app-lem:pd-order}
If $s\leq t$, then $s\uparrow_{p} \leq t\uparrow_{p}$ and $s\downarrow_{p} \leq t\downarrow_{p}$.
\end{lemma}
 \begin{proof}
For $P\in\PS$, since $s\leq t$, $s(P\cup\{p\})\leq t(P\cup\{p\})$ and $s(P\setminus\{p\})\leq t(P\setminus\{p\})$, according to Definition~\ref{def:base_type} and Definition~\ref{def:updown}, the conclusion follows.
 \end{proof}
 
\begin{definition}\label{app-def:projection}
Given $t\in\Tcal$ and $P\in\PPS$, the \termEmph{projection} of $t$ on $P$ is a security type  $\app{t}{P}$ defined as follows:
\begin{equation*}
\app{t}{P}(Q) = t(P),\; \forall Q\in \PPS.
\end{equation*}
This definiton of projection is extended to the projection of a list of types
on $P$ as follows:
$$
\app{(t_1,\dots,t_n)}{P} = (\app{t_1}{P}, \dots, \app{t_n}{P}).
$$
\end{definition}
 
 \begin{definition}\label{app-def:merge}
Given a permission $p$ and types $t_1$ and $t_2$, we define the {\termEmph{merging}} of 
$t_1$ and $t_2$ along $p$ as $\omerge {p} {t_1} {t_2}$, which is: 
\begin{equation*}
(\omerge p {t_1} {t_2}) (P) = 
\begin{cases}
t_1(P) & p\in P \\
t_2(P) & p\not\in P\\
\end{cases}
\quad\forall P\in{\PPS}
\end{equation*}
\end{definition}
 
 \begin{definition}\label{app-def:tenv-pd}
For all $x \in dom(\Gamma)$, the environment \termEmph{promotion} and \termEmph{demotion} are respectively
$(\Gamma\uparrow_{p}) (x) = \Gamma(x)\uparrow_{p}$ and 
$(\Gamma\downarrow_{p}) (x) = \Gamma(x)\downarrow_{p}$.	
The projection of $\Gamma$ on $P \in \PPS$ is defined as
$(\app{\Gamma}{P})(x) = \app{\Gamma(x)}{P}$ for each $x \in dom(\Gamma).$
\end{definition}
 
 \begin{definition}\label{app-def:envequ}
Given two evaluation environments $\eta, \eta'$, 
a typing environment $\Gamma$, a security level $l_O \in \SL$ of the observer,  
we define an {\termEmph{indistinguishability relation}} as follows:
\begin{equation*}
\eta =_{\Gamma}^{l_{O}} \eta' \text{ iff. } 
\forall x\in dom(\Gamma) .\; (\Gamma(x) \leq \hat{l}_{O} \Rightarrow \eta(x) = \eta'(x))
\end{equation*}
where 
$\eta(x) = \eta'(x)$ holds iff both sides of the equation are defined and their values are identical,
or both sides are undefined.
\end{definition}
 
\begin{lemma}\label{app-lem:ni-eq}
$=_{\Gamma}^{l_{O}}$ is an equivalence relation on $EEnv$.
\end{lemma}
 \begin{proof}
\begin{ProofEnumDesc}[style=standard]
	\item [Reflexivity] Obviously $\eta=_{\Gamma}^{l_O}\eta$.
	\item [Symmetry] Since $\forall x\in dom(\Gamma). (\Gamma(x)\leq\hat{l}_O\Rightarrow\eta'=_{\Gamma}^{l_O}\eta)$.
	\item [Transitivity] If $\eta_1=_{\Gamma}^{l_O}\eta_2$ and $\eta_2=_{\Gamma}^{l_O}\eta_3$, for a given $x\in dom(\Gamma)$, when $\Gamma(x)\leq\hat{l_O}$, we have $\eta_1(x)=\eta_2(x)$ and $\eta_2(x)=\eta_3(x)$. 
	\begin{enumerate}[label={(\arabic*)}]
	\item If $\eta_1(x)\neq\bot$, then $\eta_2(x)\neq\bot$ by the first equation, which in return requires $\eta_3(x)\neq\bot$ by the second equation; by transitivity $\eta_1(x)=\eta_3(x)$($\neq\bot$).
	\item If $\eta_1(x)=\bot$, the first equation requires that $\eta_2(x)=\bot$, which makes $\eta_3(x)=\bot$, therefore both $\eta_1(x)=\eta_3(x)$($=\bot$).
	\end{enumerate}
	Therefore $\eta_1(x)=\eta_3(x)$.
\end{ProofEnumDesc}
 \end{proof}
 
 \begin{lemma}\label{app-lem:proj}
If $\eta =_{\Gamma}^{l_{O}} \eta'$
then for each $P\in\PPS$,
$\eta =_{\app{\Gamma}{P}}^{l_O} \eta'.$
\end{lemma}
 \begin{proof}
$\forall x\in dom(\Gamma)$, we need to prove that when $\app{\Gamma}{P}(x)\leq l_O$ then $\eta(x)=\eta'(x)$. But $\app{\Gamma}{P}(x)=\app{\Gamma(x)}{P}=t(P)$, from the definition of $\eta(x)=\eta'(x)$, the conclusion holds.
 \end{proof}

 \begin{definition}\label{app-def:cmd-welltype}
Let $\mathcal{S}$ be a system, and let $\FD$, $\FT$ and $\Theta$ be its function declaration table, function type table, and permission assignments. We say $\mathcal{S}$ is \termEmph{well-typed} iff for every function $A.f$, $\vdash \FD(A.f) : \FT(A.f)$ is derivable. 
\end{definition}

\begin{lemma}\label{app-lem:up}
If $p \in P$, then $\eta =_{\app{\Gamma}{P}}^{l_{O}} \eta'$ iff $\eta =_{\app{\Gamma\uparrow_{p}}{P}}^{l_{O}} \eta'$.
\end{lemma}
\begin{proof}
  We first note that $dom(\app{\Gamma}{P}) = dom(\app{\Gamma\uparrow_{p}}{P})$ since both
  promotion and projection do not change the domain of a typing environment.
  We then show below that $\app{\Gamma}{P} = \app{\Gamma\uparrow_p}{P}$, from which the
  lemma follows immediately.
  Given any $x \in dom(\app{\Gamma\uparrow_{p}}{P})$,
  for any $Q \in \PPS$, we have
  $$
  \begin{array}{ll}
    \app{\Gamma\uparrow_p}{P}(x)(Q) & \\
    = (\app{\Gamma\uparrow_p(x)}{P})(Q) & \mbox{ by Def.\ref{app-def:tenv-pd} }\\
    = (\Gamma\uparrow_p(x))(P) & \mbox{ by Def. \ref{app-def:projection} }\\
    = \Gamma(x)(P \cup \{p\}) & \mbox { by Def. \ref{app-def:updown} } \\
    = \Gamma(x)(P) & \mbox{ by assumption $p \in P$ }\\
    = (\app{\Gamma(x)}{P})(Q) & \mbox{ by Def. \ref{app-def:projection} }\\
    = \app{\Gamma}{P}(x)(Q) & \mbox{ by Def. \ref{app-def:tenv-pd} }
  \end{array}
  $$
  Since this holds for arbitrary $Q$,it follows that $\app{\Gamma\uparrow_p}{P} = \app{\Gamma}{P}.$	
\end{proof}

\begin{lemma}\label{app-lem:down}
If $p \notin P$, then $\eta =_{\app{\Gamma}{P}}^{l_{O}} \eta' \Longleftrightarrow \eta =_{\app{\Gamma\downarrow_{p}}{P}}^{l_{O}} \eta'$.
\end{lemma}
\begin{proof}
Similar to the proof of Lemma~\ref{app-lem:up}. 
\end{proof}
 
\begin{lemma}\label{app-lem:expsafe}
Suppose $\Gamma\vdash e : t$. For $P\in\PPS$, if~$t(P)\leq l_{O}$ and 
$\eta =_{\app{\Gamma}{P}}^{l_{O}} \eta'$,  $\eta\vdash e\leadsto v$ and $\eta'\vdash e\leadsto v'$, then $v = v'$.
\end{lemma}
\begin{proof}
Consider any $P$ which satisfies $t(P)\leq l_{O}$ and $\eta =_{\app{\Gamma}{P}}^{l_{O}} \eta'$. 
The proof proceeds by induction on the derivation 
of $\Gamma;A \vdash e : t$.
\begin{ProofEnumDesc}
\item[T-VAR] We have $\Gamma \vdash \leVAR{x} :\Gamma(x) = t$.
Since $t(P) \leq l_{O}$ and $\eta =_{\app{\Gamma}{P}}^{l_{O}} \eta'$, 
it is deducible that $v = \eta(x) = \eta'(x) = v'$.
\item[T-OP] We have
\begin{equation*}
\inference
{\Gamma \vdash e_1 : t & \Gamma \vdash e_2 : t}
{\Gamma \vdash \leOP{e_1}{e_2} : t}
\end{equation*} 
and 
\begin{equation*}
\inference{\eta \vdash e_i \leadsto v_i}{\eta\vdash \leOP{e_1}{e_2}\leadsto v_1~op~v_2} ~
\inference{\eta' \vdash e_i \leadsto v'_i}{\eta'\vdash \leOP{e_1}{e_2}\leadsto v'_1~op~v'_2}
\end{equation*}  
By induction on $e_i$, we can get $v_i = v'_i$. Therefore $v = v'$.
\item[T-SUB$_e$] we have
\begin{equation*}
\inference
{\Gamma\vdash e : s & s \leq t}
{\Gamma\vdash e : t}
\end{equation*}
since $s(P)\leq t(P)$ and $t(P)\leq l_{O} $, then $s(P)\leq l_{O}$ as well, 
thus the result follows by induction on $\Gamma\vdash e : s$.
\end{ProofEnumDesc}
\end{proof}

\begin{lemma}\label{app-lem:comsafe}
Suppose $\Gamma; A \vdash c : t$. Then for any $P\in\PPS$, 
if $t(P)\nleq l_{O}$ and $\eta;A; P \vdash c \leadsto \eta' $, then 
$\eta =_{\app{\Gamma}{P}} \eta'$.
\end{lemma}
\begin{proof}
By induction on the derivation of $\Gamma; A \vdash c : t$, with subinduction 
on the derivation of $\eta;A; P \vdash c \leadsto \eta'.$

\begin{ProofEnumDesc}
\item[T-ASS] 
In this case $t = \Gamma(x)$ and 
the typing derivation has the form: 
\begin{equation*}
\inference
{\Gamma \vdash e : \Gamma(x)}
{\Gamma; A  \vdash \lcASS{x}{e} : \Gamma(x)}
\end{equation*} 
and the evaluation under $\eta$ takes the form: 
$$
\inference
{\eta \vdash v}
{\eta; A; P \vdash \lcASS{x}{e} \leadsto \eta[x \mapsto v] }
$$
That is, $\eta' = \eta[x \mapsto v].$
So $\eta$ and $\eta'$ differ possibly 
only in the mapping of $x$. 
Since $\Gamma(x)(P) = t(P) \nleq l_{O}$, 
that is $\app{\Gamma}{P}(x) \nleq \hat{l}_{O}$,
the difference in the valuation of $x$ is not observable
at level $l_O.$
It then follows from Definition~\ref{app-def:envequ} 
that $\eta =_{\app{\Gamma}{P}}^{l_{O}} \eta' $.

\item[T-CALL] In this case
the command $c$ has the form
$\lcCALL{x}{B.f}{\overline{e}}$
and the typing derivation
takes the form:
$$
\setpremisesspace{7pt}
\setpremisesend{3pt}
\inference
{
\FT(B.f) = \overline{s} \rightarrow {s'} &
\Gamma \vdash \overline{e} : \app{\overline{s}}{\Theta(A)} & 
\app{s'}{\Theta(A)} \leq \Gamma(x)}
{\Gamma; A \vdash \lcCALL{x}{B.f}{\overline{e}} : \Gamma(x) }
$$
and we have that 
$t = \Gamma(x).$ 
The evaluation under $\eta$ is derived as follows:
$$
\inference{
\FD(B.f) = \lpFUN{B.f}{\overline{y}}{c_1} \\
\eta\vdash\overline{e} \leadsto \overline{v} \quad 
[\overline{y} \mapsto \overline{v}, r \mapsto 0] ; B; \Theta(A) \vdash c_1 
\leadsto \eta_1
}
{\eta;A;P \vdash \lcCALL{x}{B.f}{\overline{e}} \leadsto 
\eta [x \mapsto \eta_1(r) ]}
$$
Since $t(P) \not \leq l_O$ and $\Gamma(x) = t$, we have 
$\Gamma(x)(P) \not \leq l_O$
and therefore $\Gamma(x) \not \leq l_O$ 
and
$$\eta =^{l_O}_{\app{\Gamma}{P}} \eta[x \mapsto v'] = \eta'.$$

\item[T-IF]  
This follows straightforwardly from the induction hypothesis.

\item[T-WHILE]
We look at the case where the condition of the while loop
evaluates to true, otherwise it is trivial. 
In this case the typing derivation is 
$$
\inference
{\Gamma \vdash e : t & \Gamma ; A\vdash c : t  }
{\Gamma; A \vdash \lcWHILE{e}{c} : t }
$$
and the evaluation derivation is
$$
\inference{
\eta \vdash e \leadsto v\quad v\neq 0\quad 
\eta;A; P \vdash c \leadsto \eta_1\\
\quad \eta_1 ; A ; P \vdash \lcWHILE{e}{c} \leadsto \eta'
}
{\eta;A; P \vdash \lcWHILE{e}{c} \leadsto \eta'}
$$
Applying the induction hypothesis (on typing derivation)
and the inner induction hypothesis (on the evaluation derivation)
we get $\eta =^{l_O}_{\app{\Gamma}{P}} \eta_1$ and $\eta_1 =^{l_O}_{\app{\Gamma}{P}} \eta'$; by transitivity of $=^{l_O}_{\app{\Gamma}{P}}$ we get $\eta =^{l_O}_{\app{\Gamma}{P}} \eta'.$

\item[T-SEQ] This case follows from the induction hypothesis
and transitivity of the indistinguishability relation.

\item[T-LETVAR] This follows from the induction hypothesis
and the fact that we can choose fresh variables for local
variables, and that the local variables are not visible
outside the scope of letvar.

\item[T-CP] 
We have:
$$
\inference
{\Gamma\uparrow_{p} ; A \vdash c_1 : t_1 &
\Gamma\downarrow_{p} ; A\vdash c_2 : t_2 &
t = \omerge{p}{t_1}{t_2}}
{\Gamma; A\vdash \lcCP{p}{c_1}{c_2} : t}
$$
There are two possible derivations for the evaluation.
In one case, we have 
$$
\inference{p\in P \quad \eta; A; P \vdash c_1 \leadsto \eta'}
{\eta; A ; P \vdash \lcCP{p}{c_1}{c_2} \leadsto \eta'}
$$
Since $t(P) \nleq l_{O}$ and $p \in P$, by Definiton~\ref{app-def:merge},
we have $t_1(P) \nleq l_{O}$.
By induction hypothesis, we have $\eta =_{\app{\Gamma\uparrow_{p}}{P}}^{l_{O}} \eta'$,
by Lemma \ref{app-lem:up}, we have $\eta =_{\app{\Gamma}{P}}^{l_{O}} \eta'$. 

The case where $p \notin P$ can be handled similarly,
making use of Lemma \ref{app-lem:down}. 

\item[T-SUB$_c$] Straightforward by induction.
\end{ProofEnumDesc}
\end{proof}

\begin{definition}\label{app-def:ni}
A command $c$ executed in app $A$ is said to be {\termEmph{noninterferent}} \text{iff.}
for all $\eta_1, \eta'_1,\Gamma, P, l_{O}$, 
if 
\begin{itemize}
\item $\eta_1 =_{\app{\Gamma}{P}}^{l_{O}} \eta'_1$, 
\item $\eta_1;A ; P\vdash c \leadsto \eta_2 $ and 
\item $\eta'_1; A; P\vdash c\leadsto \eta'_2  $, 
\end{itemize}
then $\eta_2 =_{\app{\Gamma}{P}}^{l_{O}} \eta'_2$.
\end{definition}

\begin{lemma}\label{app-lem:comni}
Suppose $\Gamma; A\vdash c : t$, for any $P\in\PPS$, if  $\eta_1 =_{\app{\Gamma}{P}}^{l_{O}} \eta'_1$, $\eta_1; A; P \vdash c \leadsto \eta_2$,
and  $\eta'_1; A; P \vdash c \leadsto \eta'_2$, 
then  $\eta_2 =_{\app{\Gamma}{P}}^{l_{O}} \eta'_2$. 
\end{lemma}
\begin{proof}
The proof proceeds by induction on \rank{c}, with subinduction
on the derivations of $\Gamma;\Theta;A \vdash c : t$ and $\eta_1;\Theta;P;A\vdash c \leadsto \eta_2$.
In the following, we shall omit the superscript $l_O$ from
$=^{l_O}_{\app{\Gamma}{P}}$ to simplify presentation.

\begin{ProofEnumDesc}
\item[T-ASS] In this case, $c \equiv x := e$ and the typing derivation takes the form:
$$
\inference
{\Gamma \vdash e : \Gamma(x)}
{\Gamma ; A \vdash \lcASS{x}{e} : \Gamma(x)}
$$
where $t = \Gamma(x)$, and suppose the two executions
of $c$ are derived as follows: 
$$
\inference{\eta_1 \vdash e \leadsto v}
{\eta_1; A; P \vdash \lcASS{x}{e} \leadsto \eta_1 [ x \mapsto v_1 ]}
$$
$$
\inference{\eta_2 \vdash e \leadsto v'}
{\eta_1'; A; P \vdash \lcASS{x}{e} \leadsto \eta_1' [ x \mapsto v_2 ]}
$$
where $\eta_2 = \eta_1[x \mapsto v]$
and $\eta_2' = \eta_1'[x \mapsto v'].$
Note that if $\Gamma(x) \not \leq l_O$ then 
$\eta_2 =_{\app{\Gamma}{P}} \eta_2'$ holds trivially
by Definition~\ref{app-def:ni}. 
So let us assume $\Gamma(x) \leq l_O$. Then applying Lemma~\ref{app-lem:expsafe}
to 
$\eta_1 \vdash e \leadsto v$
and $\eta_1' \vdash e \leadsto v'$
we get
$v = v'$, so it then follows that  $\eta_2 =_{\app{\Gamma}{P}} \eta_2'.$
\item[T-IF] In this case $c \equiv \lcIF{e}{c_1}{c_2}$ and we have
$$
\inference
{\Gamma \vdash e : t & \Gamma ; A\vdash c_1 : t & \Gamma; A \vdash c_2 : t }
{\Gamma; A \vdash \lcIF{e}{c_1}{c_2} : t}
$$
If $t(P)\not\leq l_{O}$, then the lemma follows easily from 
Lemma~\ref{app-lem:comsafe}. So we assume $t(P) \leq l_O.$

The evaluation derivation under $\eta_1$
takes either one of the following forms:
$$
\inference{\eta_1 \vdash e \leadsto v\quad v\neq 0\quad \eta_1; A; P \vdash c_1 \leadsto \eta_2}
{\eta_1; A; P \vdash \lcIF{e}{c_1}{c_2} \leadsto \eta_2}
$$
$$
\inference{\eta_1 \vdash e \leadsto v\quad v = 0\quad \eta_1;A; P \vdash c_2 \leadsto \eta_2}
{\eta;A; P \vdash \lcIF{e}{c_1}{c_2} \leadsto \eta_2}
$$
We consider here only the case where $v\not = 0$; the case with $v=0$
can be dealt with similarly.
We first need to show that the evaluation of $c$ under $\eta_1'$
would take the same if-branch. That is, suppose $\eta_1' \vdash e \leadsto v'$.
Since $t(P) \leq l_O$, we can apply Lemma~\ref{app-lem:expsafe} to conclude
that $v=v' \not = 0$, hence the evaluation of $c$ under $\eta_1'$ takes the
form:
$$
\inference{\eta_1' \vdash e \leadsto v'\quad v'\neq 0\quad \eta_1'; A; P \vdash c_1 \leadsto \eta_2'}
{\eta_1'; A; P \vdash \lcIF{e}{c_1}{c_2} \leadsto \eta_2'}
$$
The lemma then follows straightforwardly from the induction hypothesis.

\item[T-WHILE] $c\equiv\lcWHILE{e}{c_b}$ and we have
$$
\inference
{\Gamma \vdash e : t & \Gamma ; A\vdash c : t  }
{\Gamma; A \vdash \lcWHILE{e}{c} : t }
$$
If $t(P) \nleq l_{O}$, the conclusion holds by
Lemma~\ref{app-lem:comsafe}. Otherwise, since
$\eta_1=_{\app{\Gamma}{P}}^{l_{O}}\eta_1'$, by
Lemma~\ref{app-lem:expsafe}, if
$\eta_1 \vdash e \leadsto v$ and $\eta_1' \vdash e \leadsto v'$
then $v=v'$.
If both are $0$ then the conclusion holds according to
(\ruleTagText{E-WHILE-F}). Otherwise, we have
$$
\inference{
\eta_1 \vdash e \leadsto v\quad v\neq 0\quad 
\eta_1;A; P \vdash c_b \leadsto \eta_3\\
\quad \eta_3 ; A ; P \vdash \lcWHILE{e}{c_b} \leadsto \eta_2
}
{\eta;A; P \vdash \lcWHILE{e}{c_b} \leadsto \eta_2}
$$

$$
\inference{
\eta_1' \vdash e \leadsto v\quad v'\neq 0\quad 
\eta_1';A; P \vdash c_b \leadsto \eta_3'\\
\quad \eta_3' ; A ; P \vdash \lcWHILE{e}{c_b} \leadsto \eta_2'
}
{\eta;A; P \vdash \lcWHILE{e}{c_b} \leadsto \eta_2'}
$$
Applying the induction hypothesis to $\Gamma ; A\vdash c : t$,
$\eta_1;A; P \vdash c_b \leadsto \eta_3$ and
$\eta_1';A; P \vdash c_b \leadsto \eta_3'$, we obtain
$\eta_3 =_{\app{\Gamma}{P}} \eta_3'.$ Then applying the inner
induction hypothesis to $\eta_3 ; A ; P \vdash \lcWHILE{e}{c_b} \leadsto \eta_2$
and $\eta_3' ; A ; P \vdash \lcWHILE{e}{c_b} \leadsto \eta_2'$,
we obtain $\eta_2 =_{\app{\Gamma} P} \eta_2'.$

\item[T-SEQ] In this case we have
  $c\equiv\lcSEQ{c_1}{c_2}$ and $\Gamma;A\vdash c:t$. If $t(P)\not\leq l_{O}$, it is a direct conclusion from Lemma \ref{app-lem:comsafe}; otherwise it holds by induction on $c_1$ and $c_2$.

\item[T-LETVAR]
  In this case we have $c\equiv\lcLETVAR{x}{e}{c_b}.$
  If $t(P) \nleq l_O$ then the lemma follows from Lemma \ref{app-lem:comsafe}.
  Otherwise, this case follows from the induction hypothesis and the fact
  that the mapping for the local variable $x$ is removed
  in $\eta_2$ and $\eta_2'.$
\item[T-CALL] 
In this case, $c$ has the form 
$\lcCALL{x}{B.f}{\overline{e}}$. Suppose the typing
derivation is the following (where we label the premises
for ease of reference later):
$$
\inference
{
\FT(B.f) = \overline{s} \rightarrow s' \\
(\mathbf{T_1}) ~ \Gamma \vdash \overline{e} : \app{\overline{s}}{\Theta(A)} &
(\mathbf{T_2}) ~ \app{s'}{\Theta(A)} \leq \Gamma(x)}
{\Gamma; A \vdash \lcCALL{x}{B.f}{\overline{e}} : 
 \Gamma(x)  }
$$
where $t = \Gamma(x)$,
and the executions under $\eta_1$ and $\eta_1'$ are 
derived, respectively, as follows:
$$
\inference{
(\mathbf{E_1}) ~ \eta_1\vdash\overline{e} \leadsto \overline{v_1} \\
(\mathbf{E_2}) ~ [\overline{y} \mapsto \overline{v_1}, r \mapsto 0] ; B; 
  \Theta(A) \vdash c_1 \leadsto \eta_3
}
{\eta_1; A; P \vdash \lcCALL{x}{B.f}{\overline{e}} \leadsto 
              \eta_1 [x \mapsto \eta_3(r) ]}
$$
and
$$
\inference{
(\mathbf{E_1'}) ~ \eta_1'\vdash\overline{e} \leadsto \overline{v_2}
\\
(\mathbf{E_2'}) ~ [\overline{y} \mapsto \overline{v_2}, r \mapsto 0] ; B; 
\Theta(A) \vdash c_1 \leadsto \eta_3'}
{\eta_1';A;P \vdash \lcCALL{x}{B.f}{\overline{e}} \leadsto 
              \eta_1' [x \mapsto \eta_3'(r) ]}
$$
where 
$\FD(B.f) = \lpFUN{B.f}{\overline{x}}{c_1}$, 
$\eta_2 = \eta_1[x \mapsto \eta_3(r)]$
and $\eta_2' = \eta_1'[x \mapsto \eta_3'(r)].$

Moreover, since we consider only well-typed systems,
the function $\FD(B.f)$ is also typable:
$$
\inference
    {
      (\mathbf{T_3}) ~ [\overline{y}:\overline{s}, r : s']; B\vdash c_1 : s
    }
    {
      \Theta \vdash \lpFUN{B.f}{\overline{y}}{c_1} :  
      \overline{s}\rightarrow s'
    }
$$
First we note that if $t(P) \nleq l_O$ then the result
follows from Lemma~\ref{app-lem:comsafe}.
So in the following, we assume $t(P) \leq l_O$.
Since $t = \Gamma(x)$, it follows that
$\Gamma(x)(P) \leq l_O.$

Let $\Gamma' = \app{[\overline{y} : \overline{t}, r : s]}{\Theta(A)}.$
We first prove several claims:
\begin{itemize}
\item Claim 1: 
$[\overline{y} \mapsto \overline{v_1}, r \mapsto 0] 
=_{\Gamma'} [\overline{y} \mapsto \overline{v_2}, r\mapsto 0].$

Proof: Let $\rho = [\overline{y} \mapsto \overline{v_1}, r \mapsto 0]$
and $\rho' = [\overline{y} \mapsto \overline{v_2}, r \mapsto 0]$.
We need only to check that the two mapping
agrees on mappings of $\overline{y}$ that are of type $\leq l_O.$
Suppose $\Gamma'(y_u) = u \leq l_O$
and suppose $\rho(y_u) = v_u$ and
$\rho'(y_u) = v_u'$ for some $y_u \in \overline{y}.$
From $(\mathbf{E_1})$ we have 
$\eta_1 \vdash e_u \leadsto v_u$
and from $(\mathbf{E_2})$ 
we have $\eta_1' \vdash e_u \leadsto v_u'$,
and from $(\mathbf{T_1})$ we have $\Gamma \vdash e_u : u.$ 
Since $u \leq l_O$, applying Lemma~\ref{app-lem:expsafe},
we get $v_u = v_u'$.

\item Claim 2: $\eta_3 =_{\Gamma'} \eta_3'.$ 

Proof: From Claim 1, we know that 
$$[\overline{y} \mapsto \overline{v_1}, r \mapsto 0] 
=_{\Gamma'} [\overline{y} \mapsto \overline{v_2}, r \mapsto 0].$$
Since $\rank{c_1}<\rank{c}$, we can apply the outer induction
hypothesis to $(\mathbf{E_2})$, $(\mathbf{E_2'})$
and $(\mathbf{T_3})$ to obtain
$\eta_3 =_{\Gamma'} \eta_3'.$

\item Claim 3: $\eta_3(r) = \eta_3'(r).$

Proof: We first note that from $(\mathbf{T_2})$
and the assumption that $\Gamma(x)(P) \leq l_O$, we get
$(\app{s'}{\Theta(A)})(P) \leq l_O$.
The latter, by Definition~\ref{app-def:projection}, implies  that $s'(\Theta(A)) \leq l_O.$
Since $r \in dom(\Gamma')$, it is obvious that
$\Gamma' \vdash r : s'$,
$\eta_3 \vdash r \leadsto \eta_3(r)$
and $\eta_3' \vdash r \leadsto \eta_3'(r).$
From Claim 2 above, we have $\eta_3 =_{\Gamma'} \eta_3'$.
Therefore by Lemma~\ref{app-lem:expsafe}, we have
$\eta_3(r) = \eta_3'(r).$
\end{itemize}
The statement we are trying to prove, i.e., $\eta_2 =_{\app{\Gamma}{P}} \eta_2'$,
follows immediately from Claim 3 above.
\item[T-CP] $c\equiv\lcCP{p}{c_1}{c_2}$ and we have
$$
\inference
{\Gamma\uparrow_{p} ; A \vdash c_1 : t_1 &
\Gamma\downarrow_{p} ; A\vdash c_2 : t_2 &
t=\omerge {p} {t_1} {t_2}}
{\Gamma; A\vdash \lcCP{p}{c_1}{c_2} : t}
$$
We need to consider two cases, one where $p \in P$
and the other where $p \not \in P$.

Assume that $p \in P.$
Then the evaluation of $c$ under $\eta_1$ and $\eta_1'$
are respectively:
$$
\inference{p\in P \quad \eta_1; A; P \vdash c_1 \leadsto \eta_2}
{\eta_1; A ; P \vdash \lcCP{p}{c_1}{c_2} \leadsto \eta_2}
$$
and
$$
\inference
    {p\in P \quad \eta_1'; A; P \vdash c_1 \leadsto \eta_2'}
    {\eta_1'; A ; P \vdash \lcCP{p}{c_1}{c_2} \leadsto \eta_2'}
$$
Since $\eta_1 =_{\app{\Gamma}{P}} \eta_1'$ and since $p \in P$,
by Lemma~\ref{app-lem:up}, we have $\eta_1 =_{\app{\Gamma\uparrow_p}{P}} \eta_1'$.
Therefore by the induction hypothesis applied to $\Gamma\uparrow_p ; A \vdash c_1 : t_1$,
$\eta_1; A; P \vdash c_1 \leadsto \eta_2$ and
$\eta_1'; A; P \vdash c_1 \leadsto \eta_2'$, we obtain
$\eta_2 =_{\app{\Gamma\uparrow_p}{P}} \eta_2'$, and by Lemma~\ref{app-lem:up},
we get $\eta_2 =_{\app{\Gamma}{P}} \eta_2'$.

For the case where $p \not \in P$, we apply a similar reasoning as above,
but using Lemma~\ref{app-lem:down} in place of Lemma~\ref{app-lem:up}.
\end{ProofEnumDesc}
\end{proof}

\begin{definition}\label{app-def:sys-ni}
Let $\mathcal{S}$ be a system.
A function
$$\lpFUN{A.f}{\overline{x}}{c}$$
in $\mathcal{S}$
with $FT(A.f) = \overline{t} \rightarrow t'$ 
is {\termEmph{noninterferent}}
if for all $\eta_1, \eta'_1, P, v, l_{O}$,
if the following hold:   
\begin{itemize}
\item $t(P) \leq l_{O}$,
\item $\eta_1 =_{\app{\Gamma}{P}}^{l_{O}} \eta'_1$, where
$\Gamma = [\overline{x} : \overline{t}, r : t']$,  
\item $\eta_1;A ; P\vdash c \leadsto \eta_2 $, and
\item $\eta'_1; A; P\vdash c\leadsto \eta'_2 $, 
\end{itemize}
then $\eta_2(r) = \eta_2'(r).$
The system $\mathcal{S}$ is noninterferent iff all functions
in $\mathcal{S}$ are noninterferent.
\end{definition}

\begin{theorem}\label{app-thm:ni}
Well-typed systems are noninterferent.
\end{theorem}
\begin{proof}
Follows from Lemma~\ref{app-lem:comni}.
\end{proof}

 \section{Type Inference}\label{sec:app-infer}

\begin{definition}\label{app-def:app-pt-t}
Given a base type $t$ and a permission trace $\trace$, the \termEmph{application} of $\trace$ on $t$, denoted as $t\cdot\trace$,  is defined as follows:
\begin{equation*}
{t\cdot \trace =}
\begin{cases}
t & {\normalfont\text{if}}~\trace = \epsilon\\
(t\uparrow_{p})\cdot \trace' & {\normalfont\text{if}}~\exists p, \trace'. (\trace = \tplus{p} :: \trace')\\
(t\downarrow_{p})\cdot \trace' & {\normalfont\text{if}}~\exists p, \trace'. (\trace = \tminus{p} :: \trace')
\end{cases}
\end{equation*}
\end{definition}

\begin{definition}\label{app-def:partial-subtype}
The \termEmph{partial subtyping relation} $\leq_{\trace} $, which is the subtyping relation applied on the permission trace, is defined as $s\leq_{\trace}  t \text{ iff. } s\cdot \trace \leq t\cdot \trace$.
\end{definition}

\begin{lemma}\label{app-lem:monotrace}
$\forall s, t \in \Tcal$, $s\leq t \implies s\leq_{\trace} t$.
\end{lemma}
\begin{proof}
By induction on $\tlength{\trace}$.\\
\begin{itemize}
\item $\trace=\epsilon$. 
Since $s\cdot\trace=s$ and $t\cdot\trace=t$, the conclusion holds trivially.
\item 
$\trace=\tplus{p} :: \trace'$ or $\trace=\tminus{p} :: \trace'$. Assume it is the former case, the latter is similar. 
By the hypothesis and Lemma \ref{app-lem:pd-order}, $s\uparrow_{p}\leq t\uparrow_{p}$.
Then by induction, $s\uparrow_{p}\cdot\trace'\leq t\uparrow_{p}\cdot\trace'$, 
that is, $s\cdot\trace\leq t\cdot\trace$.
\end{itemize}
\end{proof}

\begin{lemma}\label{app-lem:traceorder}
$\forall t\in\Tcal$, $p,q\in\PS$ s.t. $ p\neq q$ , $t \cdot (\tplusminus{p} \tplusminuss{q}) =  t \cdot (\tplusminuss{q}  \tplusminus{p}) $, where $\tplusminus{},\tplusminuss{}, \in\{\tplus{},\tminus{}\}$. 
\end{lemma}
\begin{proof}
We only prove $t\cdot(\tminus{p} \tplus{ q})=t\cdot (\tplus{ q} \tminus{ p})$, the other cases are similar.
Consider any $P\in\PPS$,
\begin{equation*}
\begin{array}{lcl}
	 t\cdot(\tminus{p} \tplus{ q})(P) 
	 &=&((t\downarrow_{p})\uparrow_{q})(P)\\
	 &=&(t\downarrow_{p}(P\cup\{q\}))\\
	 &=&t((P\cup\{q\})\setminus\{p\})\\
\end{array}
\end{equation*}
and
\begin{equation*}
\begin{array}{lcl}
t\cdot(\tplus{ q} \tminus{ p})(P)
	&=&((t\uparrow_{q})\downarrow_{p})(P)\\
	&=&t((P\setminus\{p\})\cup \{q\})\\
	&=&t((P\cup\{q\})\setminus(\{p\}\setminus\{q\}))\\
	&=&t((P\cup\{q\})\setminus\{p\})\\
\end{array}
\end{equation*}
Therefore, $t\cdot(\tminus{ p} \tplus{ q})=t\cdot(\tplus{ q} \tminus{ p})$.
\end{proof}

\begin{lemma}\label{app-lem:traceorder-whole}
$\forall t\in\Tcal$,$(t\cdot \tplusminus{p} )\cdot \trace = (t\cdot\trace)\cdot \tplusminus{p}$, $\tplusminus{}\in\{\tplus{}, \tminus{}\}$ and $p \notin \trace$.
\end{lemma}
\begin{proof}
By induction on $\tlength{\trace}$. The conclusion holds when $\tlength{\trace}=0$ and $\tlength{\trace}=1$ by Lemma~\ref{app-lem:traceorder}. Suppose $\tlength{\trace}>1$, there exists $\trace'$ and $q$ such that $\trace=\tplusminuss{q} :: \trace'$ where $\tplusminuss{}\in\{\tplus{}, \tminus{}\}$. 
\begin{equation*}
\begin{array}{lcll}
&&(t\cdot\tplusminus{p})\cdot \trace &\\
&=& ((t\cdot \tplusminus{ p})  \cdot \tplusminuss{ q}) \cdot \trace'&\text{(by Definition~\ref{app-def:app-pt-t})}\\
              &=& (t\cdot(\tplusminus{ p}\tplusminuss{ q} )) \cdot \trace' &\text{(by Definition~\ref{app-def:app-pt-t})}\\
              &=& (t\cdot(\tplusminuss{ q} \tplusminus{ p})) \cdot \trace' &\text{(by Lemma \ref{app-lem:traceorder})}\\
              &=& ((t\cdot\tplusminuss{ q}) \cdot \tplusminus{ p}) \cdot \trace' &\text{(by Definition~\ref{app-def:app-pt-t})}\\
              &=& ((t\cdot\tplusminuss{ q}) \cdot \trace') \cdot \tplusminus{ p}  & \text{(induction hypothesis)}\\
              &=& (t\cdot (\tplusminuss{ q} :: \trace')) \cdot \tplusminus{ p} &\text{(by Definition~\ref{app-def:app-pt-t})}\\
              &=& (t\cdot \trace) \cdot \tplusminus{p} & \text{}
\end{array}
\end{equation*}
\end{proof}

\begin{lemma}\label{app-lem:tracepsame}
$\forall t\in\Tcal$, $p\in\PS$,$(t \cdot \tplusminus{ p})\cdot\tplusminuss{ p} =  t \cdot (\tplusminus{ p}) $, where $\tplusminus{}, \tplusminuss{} \in \{\tplus{},\tminus{}\}$.
\end{lemma}
\begin{proof}
By case analysis. 
\begin{itemize}
\item $((t\cdot \tplus{ p})\cdot \tplus{ p})(P)=t(P\cup\{p\}\cup\{p\})  = t(P\cup\{p\}) =t\cdot(\tplus{ p})$ for each $P\in\PPS$.
\item $((t\cdot \tplus{ p})\cdot \tminus{ p})(P)=t((P\setminus\{p\})\cup\{p\})  = t(P\cup\{p\}) =t\cdot(\tplus{ p})$ for each $P\in\PPS$.
\item $((t\cdot \tminus{ p})\cdot \tplus{ p})(P) = t((P\cup\{p\})\setminus\{p\})=t(P\setminus\{p\})=t\cdot(\tminus{ p})$ for each $P\in\PPS$.
\item $((t\cdot \tminus{ p})\cdot \tminus{ p})(P) = t((P\setminus\{p\})\setminus\{p\})=t(P\setminus\{p\})=t\cdot(\tminus{ p})$ for each $P\in\PPS$.
\end{itemize}
\end{proof}

\begin{lemma}\label{app-lem:tracesame}
$\forall t\in\Tcal$, $(t \cdot \trace)\cdot \trace =  t \cdot \trace $. 
\end{lemma}
\begin{proof}
By induction on $\tlength{\trace}$. The conclusion holds trivially for $\tlength{\trace}=0$ and for $\tlength{\trace} = 1$ by Lemma \ref{app-lem:tracepsame} . When $\tlength{\trace}>1$, without loss of generality, assume $\trace=\tplus {p} :: \trace'$.
\begin{equation*}
\begin{array}{lcll}
&&t\cdot\trace\cdot\trace&\\
&=&(t\cdot (\tplus{ p} \cdot \trace')) \cdot (\tplus{ p}::\trace') &\\
&=&(t\cdot (\trace'\cdot\tplus{ p}))\cdot (\tplus{ p}::\trace') & \text{(by Lemma~\ref{app-lem:traceorder-whole})}\\
&=&(((t\cdot\trace')\cdot \tplus{ p})\cdot\tplus{ p})\cdot\trace' &\text{(by Definition~\ref{app-def:app-pt-t})}\\
&=&((t\cdot\trace')\cdot \tplus{ p})\cdot\trace' &\text{(by Lemma~\ref{app-lem:tracepsame})}\\
&=&((t\cdot \tplus{ p})\cdot\trace')\cdot\trace' &\text{(by Lemma~\ref{app-lem:traceorder-whole})}\\
&=&(t\cdot \tplus{ p}) \cdot \trace' & \text{(induction hypothesis)}\\
&=&t\cdot (\tplus{ p}::\trace') & \text{(by Definition~\ref{app-def:app-pt-t})}\\
&=&t\cdot\trace
\end{array}
\end{equation*}
\end{proof}

\begin{lemma}\label{app-lem:tracemerge}
$\forall s,t\in\Tcal.\; \forall p\in\PS.  (\omerge{p}{s}{t})\cdot \trace = \omerge{p}{(s\cdot\trace)}{(t\cdot\trace)} $, where $p\notin \trace$.
\end{lemma}
\begin{proof}
By induction on $\tlength{\trace}$.
\begin{ProofEnumDesc}
\item[$\tlength{\trace} = 0$:] Trivially.
\item[$\tlength{\trace} > 0$:] In this case we have $\trace = \trace'::\tplus{q}$ or $\trace'::\tminus{q}$, where $\tlength{\trace'} \geq 0$, $p\notin \trace' $, $q\neq p$. We only prove $\trace'::\tplus{q}$, the other case is similar. Consider any $P$, we have
$$
\begin{array}{rlll}
&&((\omerge{p}{s}{t})\cdot (\trace'::\tplus{q}))(P)&\\ 
&= & ((\omerge{p}{s}{t})\cdot\trace')(P\cup\{q\}) & \\
 &= & (\omerge{p}{(s\cdot\trace')}{(t\cdot\trace')})(P\cup\{q\})~~~~~(\textrm{By induction}) & \\
& = & \left \{ \begin{array}{ll} (s\cdot\trace')(P\cup \{q\})  & p\in P \\ (t\cdot\trace')(P\cup \{q\})  & p\notin P \end{array} \right.&\\
& = & \left \{ \begin{array}{ll} (s\cdot(\trace'::\tplus{q}))(P)  & p\in P \\ (t\cdot(\trace'::\tplus{q}))(P)  & p\notin P \end{array} \right.&\\
&= & (\omerge{p}{(s\cdot(\trace'::\tplus{q}))}{(t\cdot(\trace'::\tplus{q}))})(P) &\\ 
\end{array}
$$ 
\end{ProofEnumDesc}
\end{proof}

\begin{lemma}\label{app-lem:tracesub}
$\forall s,t\in\Tcal.\; \forall p\in\PS.$ $s\leq t\Leftrightarrow  s\cdot\tplus{ p}\leq t\cdot\tplus{ p}$ and $s\cdot\tminus{ p}\leq t\cdot\tminus{ p}$. 
\end{lemma}
\begin{proof}
($\Rightarrow$) by applying Lemma \ref{app-lem:monotrace} with $\trace=\tplus{ p}$ and $\trace=\tminus{p}$ respectively.\\
($\Leftarrow$) $\forall P\in\PPS$,
\begin{enumerate}[label=(\arabic*)]
	\item If $p\in P$, by Lemma~\ref{app-lem:promote-demote}\ref{app-lem:promote-demote-1}, $s(P)=(s\uparrow_{p})(P)=(s\cdot\tplus{ p})(P)$ and $t(P)=(t\uparrow_{p})(P)=(t\cdot\tplus{ p})(P)$, since $s\cdot\tplus{ p}\leq t\cdot\tplus{ p}$, then $s(P)\leq t(P)$.
	\item If $p\not\in P$, by Lemma~\ref{app-lem:promote-demote}\ref{app-lem:promote-demote-2}, $s(P)=(s\downarrow_{p})(P)=(s\cdot\tminus{ p})(P)$ and $t(P)=(t\downarrow_{p})(P)=(t\cdot\tminus{ p})(P)$, since $s\cdot\tminus{ p}\leq t\cdot\tminus{p}$, then $s(P)\leq t(P)$.
\end{enumerate}
This indicates that $s\leq t$.
\end{proof}

\begin{lemma}\label{app-lem:ptrsound}
Let $\trace$ be the permission trace collected from the context of $e$ or $c$. 
\begin{enumerate}[label=(\alph*),topsep=1pt,itemsep=-1ex,partopsep=1ex,parsep=0ex]
\item\label{app-lem:ptrsound-1} If $\Gamma;\trace\vdash_{t} e : t$, then $\Gamma\cdot\trace\vdash e : (t\cdot\trace)$.\\
\item\label{app-lem:ptrsound-2} If $\Gamma;\trace; A \vdash_{t} c : t$, then $(\Gamma\cdot\trace); A \vdash c : (t\cdot\trace)$.\\
\item\label{app-lem:ptrsound-3} If $\vdash_{t} \lpPROG{B.f}{\overline{x}}{c}{\lpRES}  : \overline{t}\xrightarrow{} t'$, then $ \vdash \lpPROG{B.f}{\overline{x}}{c}{\lpRES}  : \overline{t}\xrightarrow{} t'$.
\end{enumerate}
\end{lemma}
\begin{proof}
The proof of {\ref{app-lem:ptrsound-1}}, by induction on $\Gamma;\trace \vdash_{t} e: t$.
\begin{ProofEnumDesc}
\item[TT-VAR] Trivially.
\item[TT-OP]  In this case we have $e \cong \leOP{e_1}{e_2}$ and the following derivation
\begin{equation*}
\inference
{\Gamma;\trace \vdash_{t} e_1: t_1 & \Gamma;\trace\vdash_{t} e_2: t_2}
{\Gamma;\trace \vdash_{t} \leOP{e_1}{e_2} : t_1\sqcup t_2}
\end{equation*}
By induction on $e_i$, we can get $\Gamma\cdot\trace \vdash e_i : (t_i\cdot\trace)$. 
By Lemma \ref{app-lem:monotrace} and  $t_i \leq t_1\sqcup t_2$, we have $t_i\cdot \trace \leq (t_1\sqcup t_2)\cdot\trace $.  
By subsumption, we have $\Gamma\cdot\trace \vdash e_i : ((t_1\sqcup t_2)\cdot\trace) $.
Finally, by applying Rule {(T-OP)}, we get $\Gamma\cdot\trace \vdash  \leOP{e_1}{e_2} :((t_1\sqcup t_2)\cdot\trace)  $.
\end{ProofEnumDesc}
The proof of {\ref{app-lem:ptrsound-2}}:
\begin{ProofEnumDesc}

\item[TT-ASS] In this case we have $c \cong \lcASS{x}{e}$ and the following derivation
\begin{equation*}
\inference
{(\mathbf{T_1}) ~ \Gamma;\trace \vdash_{t} e : t & (\mathbf{T_2}) ~ t \leq_{\trace} \Gamma(x)}
{\Gamma;\trace; A \vdash_{t} \lcASS{x}{e} : \Gamma(x)}
\end{equation*}
By \ref{app-lem:ptrsound-1} on $(\mathbf{T_1})$,  $\Gamma\cdot\trace \vdash e : (t\cdot \trace)$.
From $(\mathbf{T_2})$, we get $ t\cdot \trace \leq (\Gamma\cdot\trace)(x)$.
So by subsumption, $\Gamma\cdot\trace \vdash e : (\Gamma\cdot\trace)(x)$.
Finally, by Rule {(T-ASS)}, the result follows.

\item[TT-LETVAR] 
In this case we have $c \cong \lcLETVAR{x}{e}{c'}$ and the following derivation
\begin{equation*}
\inference
{(\mathbf{T_1}) ~\Gamma;\trace \vdash_{t} e : s &
(\mathbf{T_2}) ~s \leq_{\trace} s'  \\
(\mathbf{T_3}) ~\Gamma[x:s']; \trace; A \vdash_{t} c' : t }
{\Gamma;\trace; A\vdash_{t} \lcLETVAR{x}{e}{c'} :  t}
\end{equation*}
By \ref{app-lem:ptrsound-1} on $(\mathbf{T_1})$,  $\Gamma\cdot\trace \vdash e : (s\cdot \trace)$.
From $(\mathbf{T_2})$, we get $ s\cdot \trace \leq s'\cdot\trace$.
So by subsumption, $\Gamma\cdot\trace \vdash e : s'\cdot\trace$.
By induction on $(\mathbf{T_3})$,  we have $(\Gamma[x:s']\cdot \trace); A \vdash c' : t \cdot \trace$, that is $(\Gamma\cdot \trace)[x:s'\cdot \trace]; A \vdash c' : t\cdot \trace $.
Finally, by Rule {(T-LETVAR)}, the result follows.

\item[TT-IF] By induction and Rules {(T-SUB$_{c}$)}, {(T-IF)}.
\item[TT-WHILE] By induction and Rules {(T-SUB$_{c}$)}, {(T-WHILE)}.
\item[TT-SEQ] By induction and Rules {(T-SUB$_{c}$)}, {(T-SEQ)}.

\item[TT-CALL] In this case we have $c \cong \lcCALL{x}{B.f}{\overline{e}}$ and the following derivation
\begin{equation*}
\inference
{
(\mathbf{T_1})~\FT(B.f) = \overline{t} \xrightarrow{} t'  & 
(\mathbf{T_2})~\Gamma;\trace \vdash_{t} \overline{e} : \overline{s}\\
(\mathbf{T_3})~\overline{s}\leq_{\trace} \overline{\app{t}{\Theta(A)}} & 
(\mathbf{T_4})~\app{t'}{\Theta(A)}\leq_{\trace} \Gamma(x)
}
{\Gamma;\trace; A \vdash_{t} \lcCALL{x}{B.f}{\overline{e}} : \Gamma(x)}
\end{equation*}
By \ref{app-lem:ptrsound-3} on $(\mathbf{T_1})$, $\vdash B.f : \overline{t} \xrightarrow{} t'$.
By \ref{app-lem:ptrsound-1} on $(\mathbf{T_2})$ , $\Gamma\cdot\trace \vdash \overline{e} : \overline{s\cdot\trace}$.
From $(\mathbf{T_3})$, we get
$\overline{s\cdot \trace}\leq_{} \overline{(app(t,\Theta(A))\cdot \trace)} = \overline{app(t,\Theta(A))} $.
So by subsumption, $\Gamma\cdot\trace \vdash \overline{e} : \overline{app(t,\Theta(A))}$.
Similarly, from $(\mathbf{T_4})$, we get $app(t', \Theta(A))\leq_{} (\Gamma\cdot\trace)(x) $. Then by Rule {(T-ASS)}, $(\Gamma\cdot\trace); A \vdash_{} \lcCALL{x}{B.f}{\overline{e}} : (\Gamma\cdot\trace)(x) $.

\item[TT-CP] In this case we have $c\cong \lcCP{p}{c_1}{c_2}$ and the following derivation
\begin{equation*}
\inference
{\Gamma;\trace::\tplus p;  A \vdash_{t} c_1 : t_1 &
\Gamma;\trace::\tminus p; A \vdash_{t} c_2 : t_2}
{\Gamma;\trace; A \vdash_{t} \lcCP{p}{c_1}{c_2} : \omerge{p}{t_1}{t_2}}
\end{equation*}
By induction on $c_i$, we have 
\begin{gather*}
(\Gamma\cdot(\trace :: \tplus {p}));  A \vdash c_1 : (t_1\cdot(\trace::\tplus{p}) )\\
(\Gamma\cdot (\trace:: \tminus {p})); A \vdash c_2 : (t_2\cdot(\trace::\tminus{ p}))
\end{gather*}
which is equivalent to
\begin{gather*}
(\Gamma\cdot\trace)\uparrow_{p};  A \vdash c_1 : (t_1\cdot\trace)\uparrow_{p}\\
(\Gamma\cdot \trace)\downarrow_{p}; A \vdash c_2 : (t_2\cdot \trace)\downarrow_{p}
\end{gather*}
Let $t'$ be $\omerge{p}{(t_1\cdot\trace)\uparrow_{p}}{(t_2\cdot\trace)\downarrow_{p}}$.
By Rule {(T-CP)}, we have  
$$
(\Gamma\cdot\trace);  A \vdash \lcCP{p}{c_1}{c_2} : t'.
$$
Since $\trace$ is collected from the context of $c$ and there are no nested checks of $p$, we have $p\notin \trace$.
Let's consider any $P$.
If $p \in P$, then
$$
\begin{array}{rlll}
t'(P) &=&  (t_1\cdot\trace)\uparrow_{p}(P) & (p\in P)\\
&=& (t_1\cdot\trace)(P) & (\textrm{Lemma \ref{app-lem:promote-demote}})\\
&=&(\omerge{p}{(t_1\cdot \trace)}{(t_2\cdot \trace)})(P) &(p\in P)\\
&=&((\omerge{p}{t_1}{t_2})\cdot \trace)(P) &(\textrm{Lemma \ref{app-lem:tracemerge}})\\
\end{array} 
$$ 
Similarly, if $p \notin P$, $t'(P) =((\omerge{p}{t_1}{t_2})\cdot \trace)(P)$.
Therefore,  $t' = (\omerge{p}{t_1}{t_2})\cdot \trace$ and thus the result follows.

\end{ProofEnumDesc}
The proof of \ref{app-lem:ptrsound-3} :\\
Clearly, we have
\begin{equation*}
\inference
{
[\overline{x}: \overline{t},\lpRES:t'];\epsilon; B \vdash_{t} c : s & 
}
{\vdash_{t} \lpPROG{B.f}{\overline{x}}{c}{\lpRES} :  \overline{t}\xrightarrow{} t'}
\end{equation*}
Applying \ref{app-lem:ptrsound-2} on $c$, $[\overline{x}: \overline{t}, \lpRES:t'] ; B \vdash c : s$.
Finally, by Rule {(T-FUN)}, the result follows.
\end{proof}

\begin{lemma}\label{app-lem:typingtrace}
Let $\trace$ be the permission trace collected from the context of $e$ or $c$. 
\begin{enumerate}[label=(\alph*),topsep=1pt,itemsep=-1ex,partopsep=1ex,parsep=1ex]
\item\label{app-lem:typingtrace-1} If $\Gamma \vdash e : t$, then $(\Gamma\cdot\trace) \vdash e : (t\cdot\trace) $ 
\item\label{app-lem:typingtrace-2} If $\Gamma;A \vdash c : t$, then $(\Gamma\cdot\trace); A \vdash c : (t\cdot\trace) $.
\end{enumerate}
\end{lemma}
\begin{proof}
The proof of \ref{app-lem:ptrcomplete-1} : by induction on $\Gamma \vdash e : t$.
\begin{ProofEnumDesc}
\item[T-VAR] Trivially.

\item[T-OP]  In this case we have $e \cong \leOP{e_1}{e_2}$ and the following derivation
\begin{equation*}
\inference
{\Gamma \vdash e_1 : t & \Gamma \vdash e_2 : t}
{\Gamma \vdash \leOP{e_1}{e_2} : t}
\end{equation*}  
By induction on $e_i$, $(\Gamma\cdot\trace) \vdash e_i : (t\cdot\trace)$.
By Rule {(T-OP)}, we have $(\Gamma\cdot\trace)\vdash_{t} \leOP{e_1}{e_2} : (t\cdot\trace)$.

\item[T-SUB$_e$] In this case we have the following derivation
\begin{equation*}
\inference
{(\mathbf{T1})~\Gamma \vdash e : s & (\mathbf{T2})~s \leq t}
{\Gamma \vdash e : t}
\end{equation*}
By induction on $(\mathbf{T1})$, $(\Gamma\cdot\trace) \vdash_{t} e : (s\cdot\trace)$. 
From $(\mathbf{T2})$ and Lemma \ref{app-lem:monotrace}, we get $s\cdot\trace\leq t\cdot\trace $. So by subsumption, the result follows.
\end{ProofEnumDesc}

The proof of \ref{app-lem:typingtrace-2}: by induction on $\Gamma; A \vdash c: t$.
\begin{ProofEnumDesc}
\item[T-ASS] In this case we have $c \cong \lcASS{x}{e}$ and the following derivation
\begin{equation*}
\inference
{\Gamma \vdash e : \Gamma(x)}
{\Gamma; A \vdash \lcASS{x}{e} : \Gamma(x)}
\end{equation*}
By \ref{app-lem:typingtrace-1} on $e$,  $(\Gamma\cdot\trace) \vdash e : (\Gamma\cdot\trace)(x)$. 
Then by Rule {(T-ASS)},  the result follows.

\item[T-LETVAR] 
In this case we have $c \cong \lcLETVAR{x}{e}{c'}$ and the following derivation
\begin{equation*}
\inference
{\Gamma \vdash e : s &
\Gamma[x:s]; A \vdash c' : t }
{\Gamma; A\vdash \lcLETVAR{x}{e}{c'} :  t}
\end{equation*}
By \ref{app-lem:typingtrace-1} on $e$, $(\Gamma\cdot\trace) \vdash e :  (s\cdot\trace)$.
By induction on $c$, $((\Gamma[x:s])\cdot\trace);A \vdash c' :  (t\cdot\trace)$,
that is $(\Gamma\cdot\trace)[x:s\cdot\trace];A \vdash c' :  (t\cdot\trace)$.
Then by Rule {(T-LETVAR)}, the result follows.

\item[T-IF] By induction and Rule {(T-IF)}.
\item[T-WHILE] By induction and Rule {(T-WHILE)}.
\item[T-SEQ] By induction and Rule {(T-SEQ)}.

\item[T-CALL] in this case we have $c \cong \lcCALL{x}{B.f}{\overline{e}}$ and the following derivation
\begin{equation*}
\inference
{
\FT(B.f) = \overline{t} \xrightarrow{} t'  \\
(\mathbf{T1})~\Gamma \vdash \overline{e} : \overline{\app{t}{\Theta(A)}} &
(\mathbf{T2})~\app{t'}{\Theta(A)}\leq \Gamma(x)
}
{\Gamma; A \vdash \lcCALL{x}{B.f}{\overline{e}} : \Gamma(x)}
\end{equation*}
By \ref{app-lem:typingtrace-1} on $(\mathbf{T1})$, $(\Gamma\cdot\trace)\vdash \overline{e} : \overline{\app{t}{\Theta(A)}\cdot\trace}$, that is, $(\Gamma\cdot\trace)\vdash \overline{e} : \overline{\app{t}{\Theta(A)}}$.
From $(\mathbf{T2})$ and Lemma \ref{app-lem:monotrace}, $\app{t'}{\Theta(A)}\cdot\trace \leq (\Gamma\cdot\trace)(x)$, that is, $\app{t'}{\Theta(A)} \leq (\Gamma\cdot\trace)(x)$.
Then by Rule {(T-CALL)}, $(\Gamma\cdot\trace);A\vdash \lcCALL{x}{B.f}{\overline{e}}  : (\Gamma\cdot\trace)(x)$.

\item[T-CP] In this case we have $c\cong \lcCP{p}{c_1}{c_2}$ and the following derivation
\begin{equation*}
\inference
{\Gamma\uparrow_{p};  A \vdash c_1 : t_1 &
\Gamma\downarrow_{p}; A \vdash c_2 : t_2}
{\Gamma; A \vdash \lcCP{p}{c_1}{c_2} : \omerge{p}{t_1}{t_2}}
\end{equation*}
By induction on $c_1$ and $c_2$, we have
\begin{equation*}
((\Gamma\uparrow_{p})\cdot \trace); A \vdash c_1 : t_1\cdot\trace
\quad 
((\Gamma\downarrow_{p})\cdot \trace); A \vdash c_2 : t_2\cdot\trace
\end{equation*}
Since $\trace$ is collected from the context of $c$ and there are no nested checks of $p$, we have $p \notin \trace$.
By Lemma \ref{app-lem:traceorder-whole}, 
$(\Gamma\uparrow_{p})\cdot \trace= (\Gamma\cdot \trace)\uparrow_{p}$
and $(\Gamma\downarrow_{p})\cdot \trace= (\Gamma\cdot \trace)\downarrow_{p}$.
Then by Rule {(T-CP)}, we have 
$$
(\Gamma\cdot \trace); A \vdash \lcCP{p}{c_1}{c_2} : \omerge{p}{(t_1\cdot \trace)}{(t_2\cdot \trace)}.
$$
Finally, by Lemma \ref{app-lem:tracemerge}, $(\omerge{p}{t_1}{t_2})\cdot \trace = \omerge{p}{(t_1\cdot \trace)}{(t_2\cdot \trace)} $, and thus the result follows.

\item[T-SUB$_c$] 
In this case we have the following derivation
\begin{equation*}
\inference
{(\mathbf{T1})~\Gamma;A \vdash c : s & (\mathbf{T2})~t \leq s}
{\Gamma;A  \vdash c : t}	
\end{equation*}
By induction on $(\mathbf{T1})$,  $(\Gamma\cdot\trace);A \vdash c : (s\cdot\trace)$.
From $(\mathbf{T2})$ and by Lemma \ref{app-lem:monotrace}, $t\cdot\trace\leq  s\cdot\trace$.
Then by subsumption, the result follows.
\end{ProofEnumDesc}
\end{proof}

\begin{lemma}\label{app-lem:typingtrace1}
Let $\trace$ be the permission trace collected from the context of $e$ or $c$. 
\begin{enumerate}[label=(\alph*),topsep=1pt,itemsep=-1ex,partopsep=1ex,parsep=1ex]
\item\label{app-lem:typingtrace1-1} If $\Gamma\cdot\trace \vdash e : t$, then $(\Gamma\cdot\trace) \vdash e : (t\cdot\trace) $ 
\item\label{app-lem:typingtrace1-2} If $(\Gamma\cdot\trace);A \vdash c : t$, then $(\Gamma\cdot\trace); A \vdash c : (t\cdot\trace) $.
\end{enumerate}
\end{lemma}
\begin{proof}
By Lemma \ref{app-lem:typingtrace} and Lemma \ref{app-lem:tracesame}.
\end{proof}

\begin{lemma}\label{app-lem:ptrcomplete}
Let $\trace$ be the permission trace collected from the context of $e$ or $c$.
\begin{enumerate}[label=(\alph*),topsep=1pt,itemsep=-1ex,partopsep=1ex,parsep=0ex]
\item\label{app-lem:ptrcomplete-1} If $\Gamma\cdot\trace \vdash e : t\cdot \trace$, then there exists $s$ such that $\Gamma;\trace \vdash_{t} e : s$ and $s \leq_{\trace} t$.\\
\item\label{app-lem:ptrcomplete-2} If $(\Gamma\cdot\trace); A \vdash c : t\cdot \trace$, then there exists $s$ such that $\Gamma;\trace;A \vdash_{t} c : s$ and $t \leq_{\trace} s$.\\
\item\label{app-lem:ptrcomplete-3} If $\vdash \lpPROG{B.f}{\overline{x}}{c}{\lpRES}  : \overline{t}\xrightarrow{} s$, then $\vdash_{t} \lpPROG{B.f}{\overline{x}}{c}{\lpRES}: \overline{t}\xrightarrow{} s$.
\end{enumerate}
\end{lemma}
\begin{proof}
The proof of \ref{app-lem:ptrcomplete-1} : by induction on $\Gamma\cdot\trace \vdash e : t\cdot \trace $.
\begin{ProofEnumDesc}
\item[T-VAR] Trivially with $s = \Gamma(x)$.

\item[T-OP]  In this case we have $e \cong \leOP{e_1}{e_2}$ and the following derivation
\begin{equation*}
\inference
{\Gamma\cdot\trace\vdash e_1 : t\cdot\trace & \Gamma\cdot\trace \vdash e_2 : t\cdot\trace}
{\Gamma\cdot\trace\vdash \leOP{e_1}{e_2} : t\cdot\trace}
\end{equation*}  
By induction on $e_i$, there exists $s_i$ such that $\Gamma; \trace \vdash_{t} e_i : s_i$ and $s_i \leq_{\trace} t$.
By Rule {(TT-OP)}, we have $\Gamma; \trace\vdash_{t} \leOP{e_1}{e_2} : s_1 \sqcup s_2$.
Moreover, it is clear that $s_1 \sqcup s_2 \leq_{\trace} t$. Therefore, the result follows.

\item[T-SUB$_e$] In this case we have the following derivation
\begin{equation*}
\inference
{(\mathbf{T1})~\Gamma\cdot\trace \vdash e : s & (\mathbf{T2})~s \leq t\cdot\trace}
{\Gamma\cdot\trace \vdash e : t\cdot\trace}
\end{equation*}
Applying Lemma \ref{app-lem:typingtrace1} on $(\mathbf{T1})$, $\Gamma\cdot\trace \vdash e : s\cdot \trace$. 
Then by induction, there exists $s'$ such that $\Gamma;\trace \vdash_{t} e : s'$ and $s' \leq_{\trace} s$.
From $(\mathbf{T2})$ and Lemma \ref{app-lem:tracesame}, we can get $s'\leq_{\trace} t\cdot\trace \leq_{\trace} t $. Thus the result follows.
\end{ProofEnumDesc}
The proof of \ref{app-lem:ptrcomplete-2}: by induction on $\Gamma\cdot\trace; A \vdash c : t\cdot \trace $. 
\begin{ProofEnumDesc}
\item[T-ASS] In this case we have $c \cong \lcASS{x}{e}$ and the following derivation
\begin{equation*}
\inference
{(\Gamma\cdot\trace) \vdash e : (\Gamma\cdot\trace)(x)}
{(\Gamma\cdot\trace);A \vdash \lcASS{x}{e} : (\Gamma\cdot\trace)(x)}
\end{equation*}
By \ref{app-lem:ptrcomplete-1} on $e$,  there exists $s$ such that
$\Gamma; \trace \vdash_{t} e : s$ and $s \leq_{\trace} \Gamma(x)$.
Then by Rule {(TT-ASS)},  $\Gamma;\trace; A \vdash_{t} \lcASS{x}{e} :\Gamma(x)$.

\item[T-LETVAR] 
In this case we have $c \cong \lcLETVAR{x}{e}{c'}$ and the following derivation
\begin{equation*}
\inference
{(\mathbf{T_1}) ~(\Gamma\cdot\trace) \vdash e : s &
(\mathbf{T_2}) ~(\Gamma\cdot\trace)[x:s]; A \vdash c' : t\cdot\trace }
{(\Gamma\cdot\trace); A\vdash \lcLETVAR{x}{e}{c'} :  t\cdot\trace}
\end{equation*}
Applying Lemma \ref{app-lem:typingtrace1} on $(\mathbf{T_1})$, $(\Gamma\cdot\trace) \vdash e : s\cdot\trace$.
Then by \ref{app-lem:ptrcomplete-1},  there exists $s'$ such that $\Gamma;\trace \vdash_{t} e : s'$ and $s' \leq_{\trace} s$.
Applying Lemma \ref{app-lem:typingtrace} on $(\mathbf{T_2})$, $((\Gamma\cdot\trace)[x:s])\cdot\trace; A \vdash c' : (t\cdot\trace)\cdot\trace$.
And by Lemma \ref{app-lem:tracesame}, we have $(\Gamma[x:s]\cdot\trace); A \vdash c' : (t\cdot\trace)$.
Then by induction, there exists $t'$ such that $\Gamma[x:s]; \trace; A \vdash_{t} c' : t' $ and $t \leq_{\trace} t'$.
Finally, by Rule {(TT-LETVAR)}, we have $\Gamma; \trace; A \vdash_{t} \lcLETVAR{x}{e}{c'} : t' $. 

\item[T-IF] By induction and Rule {(TT-IF)}.
\item[T-WHILE] By induction and Rule {(TT-WHILE)}.
\item[T-SEQ] By induction and Rule {(TT-SEQ)}.

\item[T-CALL] In this case we have $c \cong \lcCALL{x}{B.f}{\overline{e}}$ and the following derivation
\begin{equation*}
\inference
{
\FT(B.f) = \overline{t_p} \xrightarrow{} t_r  &
(\mathbf{T1})~\Gamma\cdot \trace \vdash \overline{e} : \overline{\app{t_p}{\Theta(A)}} \\
(\mathbf{T2})~\app{t_r}{\Theta(A)}\leq (\Gamma\cdot\trace)(x)
}
{\Gamma\cdot \trace; A \vdash \lcCALL{x}{B.f}{\overline{e}} : \Gamma(x)\cdot \trace}
\end{equation*}
From $(\mathbf{T1})$, we have $\Gamma\cdot \trace \vdash_{t} \overline{e} : \overline{(\app{t_p}{\Theta(A)}\cdot\trace)}$.
Then by \ref{app-lem:ptrcomplete-1}, there exists $\overline{t_e}$ such that $\Gamma;\trace \vdash_{t} \overline{e} : \overline{t_e}$ and $\overline{t_e}\leq_{\trace}  \overline{\app{t_p}{\Theta(A)}}$.
From $(\mathbf{T2})$, we have $\app{t_r}{\Theta(A)}\leq_{\trace} \Gamma(x)$.
Then by Rule {(TT-CALL)}, $\Gamma;\trace; A \vdash \lcCALL{x}{B.f}{\overline{e}}: \Gamma(x)$.

\item[T-CP] In this case we have $c\cong \lcCP{p}{c_1}{c_2}$ and the following derivatioin
\begin{equation*}
\inference
{(\Gamma\cdot\trace)\uparrow_{p};  A \vdash c_1 : t_1 &
(\Gamma\cdot\trace)\downarrow_{p}; A \vdash c_2 : t_2}
{(\Gamma\cdot\trace);A\vdash \lcCP{p}{c_1}{c_2} : t\cdot \trace}
\end{equation*}
where $t\cdot \trace = \omerge{p}{t_1}{t_2}$.\\

Clearly, 
$(\Gamma\cdot \trace)\uparrow_{p}= \Gamma\cdot (\trace::\tplus{p})$
and $(\Gamma\cdot \trace)\downarrow_{p}= \Gamma\cdot (\trace::\tminus{p})$.
Applying Lemma \ref{app-lem:typingtrace1} on $c_1$ and $c_2$ with $\trace::\tplus  {p}$ and $ \trace::\tminus{p}$ respectively, we have 
\begin{equation*}
\begin{array}{l}
(\Gamma\cdot(\trace::\tplus  {p})); A \vdash c_1 : t_1\cdot(\trace::\tplus  {p})\\
(\Gamma\cdot (\trace::\tminus {p})); A \vdash c_2 : t_2\cdot(\trace::\tminus  {p})
\end{array}
\end{equation*}
By induction, there exist $s_1$ and $s_2$ such that
\begin{equation*}
\begin{array}{l}
\Gamma; (\trace :: \tplus {p});  A \vdash_{t} c_1 : s_1\quad
t_1 \leq_{\trace::\tplus p} s_1 \\
\Gamma; (\trace:: \tminus {p});  A \vdash_{t} c_2 : s_2\quad
t_2 \leq_{\trace::\tminus p} s_2
\end{array}
\end{equation*}
Then by Rule {(TT-CP)}, we have
$$
\Gamma;\trace;A\vdash_{t} \lcCP{p}{c_1}{c_2} : \omerge{p}{s_1}{s_2}.
$$
The remaining is to prove $t\cdot\trace \leq (\omerge{p}{s_1}{s_2})\cdot\trace$.
Since $\trace$ is collected from the context of $c$ and there are no nested checks of $p$, we have $p \notin \trace$.
Consider any $P$. If $p \in P$, then
$$
\begin{array}{rlll}
&&((\omerge{p}{s_1}{s_2})\cdot\trace)(P) &\\
&=&(\omerge{p}{(s_1\cdot\trace)}{(s_2\cdot\trace)})(P) &(\textrm{Lemma \ref{app-lem:tracemerge}})\\
&=&(s_1\cdot\trace)(P) & (p\in P)\\
&=&(s_1\cdot(\trace::\tplus{p}))(P) & (\textrm{Lemma \ref{app-lem:promote-demote}})\\
&\geq&(t_1\cdot(\trace::\tplus{p}))(P) & (t_1 \leq_{\trace::\tplus{p}} s_1)\\
& = & (t_1\cdot\trace)(P) &(\textrm{Lemma \ref{app-lem:promote-demote}})\\
& = & (\omerge{p}{(t_1\cdot\trace)}{(t_2\cdot\trace)})(P)&(p\in P)\\
& = & ((\omerge{p}{t_1}{t_2})\cdot\trace)(P)& (\textrm{Lemma \ref{app-lem:tracemerge}})\\
& = & ((t\cdot\trace)\cdot\trace)(P)&\\
& = & (t\cdot\trace)(P)& (\textrm{Lemma \ref{app-lem:tracesame}})\\
\end{array}
$$
Similarly, if $p\notin P$, we also have $(t \cdot \trace)(P)\leq ((\omerge{p}{s_1}{s_2})\cdot\trace)(P)$.
Therefore, the result follows.

\item[T-SUB$_c$] In this case we have the following derivation
\begin{equation*}
\inference
{(\mathbf{T1})~\Gamma\cdot\trace; A\vdash c : s & (\mathbf{T2})~ t\cdot\trace \leq s}
{\Gamma\cdot\trace; A\vdash c : t\cdot \trace}	
\end{equation*}
Applying Lemma \ref{app-lem:typingtrace1} on $(\mathbf{T1})$, $\Gamma\cdot\trace; A\vdash c : s\cdot\trace$.
Then by induction, there exists $s'$ such that $\Gamma;\trace; A\vdash_{t} c : s'$ and $s \leq_{\trace} s' $. 
From ($\mathbf{T2}$) and by Lemma \ref{app-lem:monotrace}, we can get $t\leq_{\trace} s\leq_{\trace} s'$.  Thus the result follows.

\end{ProofEnumDesc}
The proof of \ref{app-lem:ptrcomplete-3}: \\
Clearly, we have
\begin{equation*}
\inference
{
[\overline{x}: \overline{t}, \lpRES: t']; B \vdash c : s 
}
{\vdash \lpPROG{B.f}{\overline{x}}{c}{\lpRES} :  \overline{t}\xrightarrow{} t'}
\end{equation*}
By \ref{app-lem:ptrcomplete-2} on $c$, there exists $t_b$ such that $[\overline{x}: \overline{t}, \lpRES: t']; B \vdash_{t} c : t_b$ and $s \leq_{} t_b $.
By Rule {(TT-FUN)}, 
$$\vdash_{t} \lpPROG{B.f}{\overline{x}}{c}{\lpRES} :  \overline{t}\xrightarrow{} t'$$
\end{proof}

\begin{definition}
Given a constraint set $C$ and a substitution $\theta$, we say $\theta$ is a \termEmph{solution} of $C$, denoted as $\theta \vDash C $, iff. for each $(\trace, L \leq R) \in C$, $L\theta \leq_{\trace} R\theta $ holds.
\end{definition}

\begin{lemma}\label{app-lem:cgrsound}
\begin{enumerate}[label={(\alph*)}]
\item\label{app-lem:cgrsound-1} If $\Gamma;\trace \vdash_{g} e : t\leadsto C$ and $\theta \vDash C$, then $\Gamma\theta;\trace \vdash_{t} e : t\theta $.
\item\label{app-lem:cgrsound-2} If $\Gamma;\trace; A \vdash_{g} c : t\leadsto C$ and $\theta \vDash C$, then $\Gamma\theta;\trace;A \vdash_{t} c : t\theta $.
\item\label{app-lem:cgrsound-3} If $ \vdash_{g} \lpPROG{B.f}{x}{c}{\lpRES} :  \overline{\alpha}\xrightarrow{} \beta  \leadsto  C$ and $\theta \vDash C$, then $\vdash_{t} \lpPROG{B.f}{x}{c}{\lpRES} : \overline{\theta(\alpha)}\xrightarrow{}\theta(\beta)$.
\end{enumerate}
\end{lemma}
\begin{proof}
The proof of \ref{app-lem:cgrsound-1}: by induction on the derivation of $\Gamma;\trace \vdash_{g} e : t\leadsto C$. 
\begin{ProofEnumDesc}
\item[TG-VAR:] In this case we have the following derivation
\begin{equation*}
\inference{}
{\Gamma;\trace\vdash_{g} x : \Gamma(x)\leadsto \emptyset}
\end{equation*}
Clearly, for any $\theta$, $\theta \vDash \emptyset$. 
By Rule {(TT-VAR)}, we have
\begin{equation*}
\inference{}
{\Gamma\theta;\trace; A \vdash_{t} \leVAR{x} : \Gamma\theta(x)} 
\end{equation*}

\item[TG-OP] In this case we have $e \cong \leOP{e_1}{e_2}$ and the following derivation
\begin{equation*}
\inference
{\Gamma;\trace \vdash_{g} e_1: t_1 \leadsto C_1 &
 \Gamma;\trace \vdash_{g} e_2: t_2 \leadsto C_2}
{\Gamma;\trace \vdash_{g} \leOP{e_1}{e_2} : t_1\sqcup t_2\leadsto C_1\cup C_2}
\end{equation*}
Since $\theta \vDash C_1\cup C_2$, $\theta\vDash C_1$ and $\theta\vDash C_2$.
Then by induction on $e_i$, $\Gamma\theta;\trace \vdash_{t} e_i: t_i\theta$.
So by Rule {(TT-OP)}, we have $\Gamma\theta;\trace \vdash_{t} \leOP{e_1}{e_2}: t_1\theta\sqcup t_2\theta $.
Clearly, $t_1\theta\sqcup t_2\theta = (t_1\sqcup t_2)\theta$.
\end{ProofEnumDesc}
The proof of \ref{app-lem:cgrsound-2}:
\begin{ProofEnumDesc}
\item[TG-ASS] In this case we have $c \cong \lcASS{x}{e}$ and the following derivation 
\begin{equation*}
\inference
{\Gamma;\trace \vdash_{g} e : t \leadsto C_{e} }
{\Gamma;\trace; A \vdash_{g} \lcASS{x}{e} : \Gamma(x)\leadsto C_{e}\cup \{ (\trace, t \leq_{} \Gamma(x))\} }
\end{equation*}
Since $\theta\vDash C_{e}\cup \{ (\trace, t \leq_{} \Gamma(x))\}$, 
$\theta\vDash C_{e}$ and $t\theta \leq_{\trace} \Gamma(x)\theta $.
By \ref{app-lem:cgrsound-1} on $e$, $\Gamma\theta;\trace \vdash_{t} e : t \theta$.
By Rule {(TT-ASS)}, $\Gamma\theta;\trace; A \vdash_{t} x := e : \Gamma\theta(x)$.

\item[TG-LETVAR] In this case we have $c \cong  \lcLETVAR{x}{e}{c'} $ and the following derivation 
\begin{equation*}
\inference
{\Gamma;\trace \vdash_{g} e : s \leadsto C_1&
\Gamma[x:\alpha]; \trace; A \vdash_{g} c' : t \leadsto C_2 \\
C = C_1\cup C_2 \cup \{(\trace, s \leq_{} \alpha) \}
}
{\Gamma;\trace; A\vdash_{g} \lcLETVAR{x}{e}{c} :  t\leadsto C}
\end{equation*}
Since $\theta \vDash C$, $\theta \vDash C_1$, $\theta\vDash C_2$, and $s\theta \leq_{\trace} \theta(\alpha)$.
By \ref{app-lem:cgrsound-1} on $e$, $\Gamma\theta;\trace \vdash_{t} e : s \theta $.
By induction on $c'$,  $\Gamma\theta[x:\theta(\alpha)]; \trace; A \vdash_{t} c' : t\theta$.
Finally, by Rule {(TT-LETVAR)}, $\Gamma\theta; \trace; A \vdash_{t} \lcLETVAR{x}{e}{c}: t\theta $.

\item[TG-CALL] In this case we have $c \cong  \lcCALL{x}{B.f}{\overline{e}} $ and the following derivation 
\begin{equation*}
\inference
{
\FT_{C}(B.f) = (\overline{t} \xrightarrow{} t', C_{f})  \\
\Gamma;\trace \vdash_{g} \overline{e} : \overline{s} \leadsto \bigcup \overline{C_{e}} \\
C_{a} = \{(\trace, \overline{s}\leq_{} \overline{\app{t}{\Theta(A)}}), (\trace, \app{t'}{\Theta(A)}\leq_{} \Gamma(x))\} \\
C = C_{f} \cup \bigcup\overline{C_{e}}\cup C_{a}
}
{
\Gamma;\trace; A \vdash_{g} \lcCALL{x}{B.f}{\overline{e}} : \Gamma(x)\leadsto C
}
\end{equation*}
Since $\theta \vDash C$, then $\theta \vDash C_{f}$, $\theta\vDash \overline{C_{e}}$,
$\overline{s\theta}\leq_{\trace} \overline{\app{t\theta}{\Theta(A)}}$, 
and $\app{t'\theta}{\Theta(A)}\leq_{\trace} \Gamma\theta(x))$.
By \ref{app-lem:cgrsound-3} on $B.f$,  we have
$$
\vdash_{t} \lpPROG{B.f}{x}{c}{\lpRES} :  \overline{t\theta}\xrightarrow{} t'\theta
$$
that is, $\FT(B.f) =  \overline{t\theta}\xrightarrow{} t'\theta$.
By \ref{app-lem:cgrsound-1} on $\overline{e}$, 
$$\Gamma\theta;\trace; A \vdash_{t} \overline{e} : \overline{s\theta}.$$
Finally, by Rule  {(TT-CALL)}, 
$$
\Gamma\theta;\trace; A \vdash_{t} \lcCALL{x}{B.f}{\overline{e}} : \Gamma\theta(x)
$$

\item[TG-CP] In this case we have $c \cong \lcCP{p}{c_1}{c_2}$ and the following derivation 
\begin{equation*}
\inference
{
\Gamma;\trace::\tplus p; A \vdash_{g} c_1 : t_1  \leadsto C_1 \\
 \Gamma;\trace::\tminus p; A \vdash_{g} c_2 : t_2  \leadsto C_2
 }
{\Gamma;\trace; A \vdash_{g} \lcCP{p}{c_1}{c_2} : \omerge{p}{t_1}{t_2} \leadsto C_1\cup C_2}
\end{equation*}
Since $\theta \vDash C_1 \cup C_2$, then $\theta\vDash C_1$ and $\theta \vDash C_2$.
By induction on $c_1$ and $c_2$, we get 
$$
\begin{array}{l}
\Gamma\theta;(\trace::\tplus p);  A \vdash_{t} c_1 : t_1\theta\\
\Gamma\theta;(\trace::\tminus p); A \vdash_{t} c_2 : t_2\theta  
\end{array}
$$
By Rule {(TT-CP)}, we have
$$
\Gamma\theta;\trace; A \vdash_{t} \lcCP{p}{c_1}{c_2} : \omerge{p}{t_1\theta}{t_2\theta} 
$$
Moreover, it is clear that $(\omerge{p}{t_1}{t_2})\theta = \omerge{p}{(t_1\theta)}{(t_2\theta)}$.

\item[others] By induction.
\end{ProofEnumDesc}
The proof of \ref{app-lem:cgrsound-3}:\\
\begin{equation*}
\inference
{
[\overline{x}: \overline{\alpha},\lpRES:\beta]; \epsilon; B \vdash_{g} c : s  \leadsto C \\
}
{
\vdash_{g} \lpPROG{B.f}{x}{c}{\lpRES} :  \overline{\alpha}\xrightarrow{} \beta  \leadsto  C
}
\end{equation*}
Since $\theta \vDash C$,  
by \ref{app-lem:cgrsound-2} on $c$, we have
$$
[\overline{ x}: \overline{\theta(\alpha)},\lpRES:\theta(\beta)];\epsilon; B \vdash_{t} c : s\theta 
$$
Finally, by Rule {(TT-FUN)}, we get 
$$
\vdash_{t} \lpPROG{B.f}{\overline{x}}{c}{\lpRES} : \overline{\theta(\alpha)}\xrightarrow{} \theta(\beta) 
$$
\end{proof}

\begin{lemma}\label{app-lem:cgrcomplete}
\begin{enumerate}[label={(\alph*)}]
\item\label{app-lem:cgrcomplete-1} If $\Gamma;\trace \vdash_{t} e : t$, then there exist $\Gamma',t',C,\theta$ such that $\Gamma';\trace \vdash_{g} e : t' \leadsto C$,  $\theta \vDash C$, $\Gamma'\theta =\Gamma$ and $t'\theta = t$.
\item\label{app-lem:cgrcomplete-2} If $\Gamma;\trace; A \vdash_{t} c : t$, then there exist $\Gamma',t',C,\theta$ such that $\Gamma';\trace; A \vdash_{g} c : t' \leadsto C$,  $\theta \vDash C$, $\Gamma'\theta =\Gamma$ and $t'\theta = t$.
\item\label{app-lem:cgrcomplete-3} If $\vdash_{t} \lpPROG{B.f}{\overline{x}}{c}{\lpRES} : \overline{t_p}\xrightarrow{} t_r$, then there exist $\alpha,\beta,C,\theta$ such that 
$$
\vdash_{g} \lpPROG{B.f}{\overline{x}}{c}{\lpRES}: \overline{\alpha}\xrightarrow{} \beta \leadsto C,
$$ 
$\theta \vDash C$, and $(\overline{\alpha}\xrightarrow{} \beta)\theta = \overline{t_p}\xrightarrow{} t_r$, where $\alpha,\beta$ are fresh type variables.
\end{enumerate}
\end{lemma}
\begin{proof}
The proof of \ref{app-lem:cgrcomplete-1}: 
Let $\Gamma_{0} = \{x \mapsto \alpha_{x} ~|~ x \in dom(\Gamma) \} $ 
and $\theta_{0} = \{\alpha_{x}\mapsto \Gamma(x)~|~x\in dom(\Gamma)\} $, where $\alpha_x$s are fresh type variables.
Clearly,  we have $\Gamma_{0}\theta_{0} = \Gamma$.
The remaining is to prove that
\begin{equation*}
\exists C, t'.~\Gamma_{0};\trace \vdash_{g} e : t' \leadsto C,~ \theta_{0}\vDash C  \textrm{ and } t'\theta_{0} = t \tag{1'}
\end{equation*}
\begin{ProofEnumDesc}
\item[TT-VAR] In this case we have $e\cong \leVAR{x}$ and the following derivation
\begin{equation*}
\inference{}
{\Gamma;\trace \vdash_{t} \leVAR{x} : \Gamma(x)} 
\end{equation*}
By Rule {(TG-VAR)}, we have
\begin{equation*}
\inference{}
{\Gamma_{0};\trace \vdash_{g} \leVAR{x} : \Gamma_{0}(x)\leadsto \emptyset}
\end{equation*}
Clearly, $\theta_{0} \vDash \emptyset$ and $\Gamma_{0}\theta_{0}(x) = \Gamma(x)$.

\item[TT-OP] In this case we have $e\cong \leOP{e_1}{e_2}$ and the following derivation 
\begin{equation*}
\inference
{\Gamma;\trace \vdash_{t} e_1: t_1 & \Gamma;\trace \vdash_{t} e_2: t_2}
{\Gamma;\trace \vdash_{t} \leOP{e_1}{e_2} : t_1\sqcup t_2}
\end{equation*}
By induction on $e_i$, there exist $C_i, t'_i$ such that
\begin{equation*}
\Gamma_{0};\trace \vdash_{g} e_i: t'_i\leadsto C_i\quad\theta_{0} \vDash C_i\quad t'_i\theta_{0} = t_i.
\end{equation*}
By Rule {(TG-OP)}, we have 
\begin{equation*}
\Gamma_{0};\trace \vdash_{g} \leOP{e_1}{e_2} : t'_1\sqcup t'_2 \leadsto C_1\cup C_2.
\end{equation*}
Moreover, it is clear that $\theta_{0} \vDash C_1\cup C_2$ and $(t'_1\sqcup t'_2)\theta_{0} = t_1\sqcup t_2 $.
\end{ProofEnumDesc}
The proof of \ref{app-lem:cgrcomplete-2}.
Let $\Gamma_{0} = \{x \mapsto \alpha_{x} ~|~ x \in dom(\Gamma) \} $ 
and $\theta_{0} = \{\alpha_{x}\mapsto \Gamma(x)~|~x\in dom(\Gamma)\} $, where $\alpha_x$s are fresh type variables.
Clearly,  we have $\Gamma_{0}\theta_{0} = \Gamma$.
Assume that different parameters of different functions have different names.
Let $V_1 ~\sharp~ V_2 $ denote $V_1\cap V_2 = \emptyset$. 
In the remaining we prove the following statement:
\begin{gather*}
\exists C, t',\theta.~\Gamma_{0};\trace; A \vdash_{g} c : t' \leadsto C,~ dom(\theta_{0})~\sharp~ dom(\theta),\\
\theta_{0}\cup\theta \vDash C,  \textrm{ and } t'(\theta_{0}\cup\theta) = t \tag{2'}
\end{gather*}
\begin{ProofEnumDesc}
\item[TT-ASS] In this case we have $c\cong  \lcASS{x}{e}$ and the following derivation
\begin{equation*}
\inference
{\Gamma;\trace \vdash_{t} e : t \quad   t \leq_{\trace} \Gamma(x) }
{\Gamma;\trace; A \vdash_{t} \lcASS{x}{e} : \Gamma(x)}
\end{equation*}
By $(1')$ on $e$, there exist $C, t'$ such that $\Gamma_{0};\trace \vdash_{g} e : t' \leadsto C$, 
$\theta_{0} \vDash C$ and $t'\theta_{0} = t$.
By Rule {(TG-ASS)}, we get 
\begin{equation*}
\Gamma_{0};\trace; A \vdash_{g} \lcASS{x}{e} : \Gamma_{0}(x) \leadsto C \cup \{(\trace, t'\leq \Gamma_{0}(x)) \}.
\end{equation*} 
Take $\theta = \emptyset $.
Since $t'\theta_{0} = t \leq_{\trace} \Gamma(x) = \Gamma_{0}\theta_{0}(x)$, 
then $\theta_{0} \vDash \{ (\trace, t' \leq_{} \Gamma_{0}(x)) \} $, 
and thus $\theta_{0} \vDash C \cup  \{ (\trace, t' \leq_{} \Gamma_{0}(x)) \} $.

\item[TT-LETVAR] In this case we have $c\cong \lcLETVAR{x}{e}{c'}$ and the following derivation
\begin{equation*}
\inference
{\Gamma;\trace \vdash_{t} e : s &
\Gamma[x:s']; \trace; A \vdash_{t} c' : t & s \leq_{\trace} s' }
{\Gamma;\trace; A\vdash_{t} \lcLETVAR{x}{e}{c'} :  t}
\end{equation*}
By $(1')$ on $e$, there exist $C_{e}, s_{e}$ such that $\Gamma_{0};\trace \vdash_{g} e : s_{e} \leadsto C_{e}$, 
$\theta_{0} \vDash C_{e}$ and $s_{e}\theta_{0} = s$.\\
Let $\Gamma'_{0} = \Gamma_{0}[x:\alpha_x]$ and $\theta'_{0} = \theta_0\cup\{\alpha_x \mapsto s'\}$, where $\alpha_x$ is fresh.
By induction on $c'$, there exist $C', t',\theta$ such that $\Gamma'_{0};\trace; A \vdash_{g} c' : t' \leadsto C'$,  $dom(\theta'_{0})~\sharp~ dom(\theta)$,  $\theta'_{0}\cup\theta \vDash C'$,  and $t'(\theta'_{0}\cup\theta) = t$. \\
By Rule (TG-LETVAR), we have
$$
\Gamma_0;\trace; A\vdash_{g} \lcLETVAR{x}{e}{c} :  t'\leadsto C
$$
where $C = C_{e} \cup C' \cup \{ (\trace, s_{e} \leq_{} \alpha_x)\}$.
Let $\theta' = \theta\cup \{\alpha_x\mapsto s'\}$. It is clear that $dom(\theta_{0})~\sharp~ dom(\theta')$
and $\theta_{0}\cup\theta' = \theta'_{0}\cup\theta$.
From the construction of $\theta'$ (\emph{i.e.}, the proof of $(2')$ and $(3')$), the type variables in $dom(\theta')$ are collected from the types of the functions and (local) variables, which are fresh.
So we also have $\theta_{0}\cup\theta'\vDash C_{e}$ and $\Gamma_{0}(\theta_{0}\cup\theta') =\Gamma_{0}\theta_{0}=\Gamma $.
Moreover,  $s_{e}(\theta_{0}\cup\theta')=s_{e}\theta_{0}=s\leq_{\trace} s' = \alpha_x(\theta_{0}\cup\theta')$.
Therefore, $\theta_{0}\cup\theta'\vDash C$.

\item[TT-CALL] In this case we have $c\cong \lcCALL{x}{B.f}{\overline{e}}$ and the following derivation 
\begin{equation*}
\inference
{
\FT(B.f) = \overline{t_p} \xrightarrow{} t_r  & 
\Gamma;\trace \vdash_{t} \overline{e} : \overline{s}\\
\overline{s}\leq_{\trace} \overline{\app{t_p}{\Theta(A)}} & 
\app{t_r}{\Theta(A)}\leq_{\trace} \Gamma(x)
}
{\Gamma;\trace; A \vdash_{t} \lcCALL{x}{B.f}{\overline{e}} : \Gamma(x)} 
\end{equation*}
By \ref{app-lem:cgrcomplete-3} on B.f, there exist $\alpha,\beta, C_{f}, \theta_{f}$ such that
$\vdash_{g}\lpPROG{B.f}{\overline{x}}{c}{\lpRES}:  \overline{\alpha}\xrightarrow{} \beta \leadsto C_{f}$,
$\theta_f \vDash C_f $ and $(\overline{\alpha}\xrightarrow{} \beta)\theta_{f} = \overline{t_p}\xrightarrow{} t_r$. So we have $\FT_C(B.f) =(\overline{\alpha}\xrightarrow{} \beta, C_{f}) $.
Since $\alpha_{x}$s are fresh, we can safely get $dom(\theta_{0}) ~\sharp~ dom(\theta_f)$.
Therefore, $\theta_{0}\cup\theta_f \vDash C_{f}$ and $(\overline{\alpha}\xrightarrow{} \beta)(\theta_{0}\cup\theta_{f}) = \overline{t_p}\xrightarrow{} t_r$.\\
By $(1')$ on $\overline{e}$, there exist $\overline{s'}, \overline{C_{e}}$ such that 
$\Gamma_{0};\trace; A \vdash_{g} \overline{e} : \overline{s'}\leadsto \overline{C_e}$, $\theta_{0} \vDash \overline{C_e} $, and $\overline{s'(\theta_{0})} = \overline{s}$.\\
From the construction of $\theta_f$ (\emph{i.e.}, the proof of $(2')$ and $(3')$), the type variables in $dom(\theta_f)$ are collected from the types of functions and (local) variables, which are fresh. 
So we also have $\theta_{0}\cup \theta_f \vDash \overline{C}$,  $\overline{s'(\theta_{0}\cup\theta_f)} = \overline{s}$ and $\Gamma_{0}(\theta_{0}\cup \theta_f) = \Gamma$. \\
By Rule {(TG-CALL)}, we have 
\begin{equation*}
\Gamma_{0};\trace; A \vdash_{g} \lcCALL{x}{B.f}{\overline{e}} : \Gamma_{0}(x) \leadsto C
\end{equation*}
where $C = C_{f} \cup \bigcup\overline{C} \cup C'$
and $C'= \{(\trace, \overline{s'}\leq \overline{\app{\alpha}{\Theta(A)}}), (\trace, \app{\beta}{\Theta(A)}) \leq \Gamma_{0}(x)) \}$.
Since $dom(\theta_f)$ are fresh, we have
\begin{equation*}
\overline{s'(\theta_{0}\cup \theta_f)} = \overline{s} \leq_{\trace} \overline{\app{t_p}{\Theta(A)}} = \overline{\app{\alpha(\theta_{0}\cup \theta_f)}{\Theta(A)}}
\end{equation*}
and
\begin{equation*}
\app{\beta(\theta_{0}\cup \theta_f)}{\Theta(A)} =  \app{t_r}{\Theta(A)} 
\leq_{\trace} \Gamma(x) = \Gamma_{0}(\theta_{0}\cup \theta_f)(x)
\end{equation*}
So $(\theta_{0}\cup \theta_f) \vDash C'$, and thus $(\theta_{0}\cup \theta_f)\vDash C$.
Moreover, we also have
$(\Gamma_{0}(x))(\theta_0\cup\theta_f)
= \Gamma(x)$.
Thus the result follows.

\item[TT-CP] In this case we have $c\cong \lcCP{p}{c_1}{c_2}$ and the following derivation
\begin{equation*}
\inference
{\Gamma;\trace::\tplus p; A \vdash_{t} c_1 : t_1  &
\Gamma;\trace::\tminus p; A \vdash_{t} c_2 : t_2}
{\Gamma;\trace; A \vdash_{t} \lcCP{p}{c_1}{c_2} : \omerge{p}{t_1}{t_2}}
\end{equation*}
By induction on $c_i$, we have
\begin{equation*}
\begin{array}{l}
\exists C_1, s_1,\theta_1.~\Gamma_{0};\trace::\tplus p; A \vdash_{g} c_1 : s_1 \leadsto C_1,~ \theta_{0}\cup\theta_1 \vDash C_1,\\
\quad\quad dom(\theta_{0})\sharp dom(\theta_1)  \textrm{ and } s_1(\theta_{0}\cup\theta_{1}) = t_1\\
\exists C_2, s_2,\theta_2.~\Gamma_{0};\trace::\tminus p; A \vdash_{g} c_2 : s_2 \leadsto C_2,~ \theta_{0}\cup\theta_2 \vDash C_2, \\
\quad\quad dom(\theta_{0})\sharp dom(\theta_2)  \textrm{ and } s_2(\theta_{0}\cup\theta_{2}) = t_2\\
\end{array}
\end{equation*}
By Rule {(TG-CP)}, we get 
\begin{equation*}
\Gamma_{0};\trace; A \vdash_{g} \lcCP{p}{c_1}{c_2} : \omerge{p}{s_1}{s_2}\leadsto C_1\cup C_2
\end{equation*}
From the construction of $\theta_1$ and $\theta_2$ (\emph{i.e.}, the proof of $(2')$ and $(3')$), the type variables in $dom(\theta_1)$ and $dom(\theta_2)$ are collected from the types of functions and (local) variables, which are fresh. Moreover, if $\theta_1 \cap \theta_2 \neq \emptyset$, that is, they share some functions, then $\theta_1(\alpha_x) = \theta_2(\alpha_x) $ for all $\alpha_x \in dom(\theta_1) \cap dom(\theta_2)$.
So we can safely get $(\theta_{0}\cup\theta_1\cup\theta_2) \vDash C_i$ and  $s_i(\theta_{0}\cup\theta_{1}\cup\theta_{2}) = t_i$, and $\Gamma_{0}(\theta_{0}\cup\theta_{1}\cup\theta_{2}) =\Gamma$.
Therefore, $(\theta_{0}\cup\theta_1\cup\theta_2) \vDash C_1\cup C_2$ and  
$(\omerge{p}{s_1}{s_2})(\theta_{0}\cup\theta_{1}\cup\theta_{2}) = \omerge{p}{t_1}{t_2}$.

\item[others] By induction.

\end{ProofEnumDesc}

The proof of \ref{app-lem:cgrcomplete-3}: \\
\begin{equation*}
\inference
{[\overline{x}: \overline{t_p},\lpRES:t_r];\epsilon; B \vdash_{t} c : s }
{\vdash_{t} \lpPROG{B.f}{\overline{x}}{c}{\lpRES} :  \overline{t_p}\xrightarrow{} t_r}
\end{equation*}
Let $\Gamma_{0} = \{\overline{x} \mapsto \overline{\alpha}, r\mapsto \beta \} $ 
and $\theta_{0} = \{\overline{\alpha}\mapsto \overline{t_p}, \beta\mapsto t_r\} $,
where $\overline{\alpha},\beta$ are fresh.
By $(2')$ on $c$, we have
\begin{equation*}
\begin{array}{l}
\exists t, C_f, \theta .~ \Gamma_{0};\epsilon; B \vdash_{g} c : t\leadsto C_f,~ (\theta_{0}\cup\theta')\vDash C_f,\\
\quad\quad dom(\theta_{0}) ~\sharp~ dom(\theta) \textrm{ and } t(\theta_{0}\cup\theta') = s.
\end{array}
\end{equation*}
By Rule {(TG-FUN)}, we have
\begin{gather*}
\vdash_{g} \lpPROG{B.f}{\overline{x}}{c}{\lpRES} :  \overline{\alpha}\xrightarrow{} \beta\leadsto C_f
\end{gather*}
Finally, it is clear that $(\overline{\alpha}\xrightarrow{} \beta) (\theta_{0}\cup\theta' ) = \overline{t_p}\xrightarrow{} t_r$.
\end{proof}
 \begin{lemma}\label{app-lem:tracesub1}
Given two types $s,t$ and a permission trace $\trace$, then
$s\leq t\Longleftrightarrow s\cdot\trace \leq t\cdot\trace$ and $\forall \trace' \in dnf(\neg\trace).\; s\cdot \trace' \leq t\cdot\trace'$. 
\end{lemma}
\begin{proof}
Generalized from Lemma \ref{app-lem:tracesub}.
\end{proof}

\begin{lemma}\label{app-lem:consolcorrect}
If $C\leadsto_{r} C'$, then $C\vDash C'$ and $C' \vDash C$, where $r\in \{d,s,m\}$. 
\end{lemma}
\begin{proof} By case analysis.
\begin{ProofEnumDesc}
\item[CD-CUP,~CD-CAP] Trivial.
\item[CD-MERGE]
By Lemma \ref{app-lem:tracesub}. 
\item[CD-LAPP,~CD-RAPP] By definition of projection.
\item[CD-SVAR,~CD-SUB$_0$,~CD-SUB$_1$] Trivial.
\item[CS-LU]
It is clear that $C\cup C' \vDash C$. 
For the other direction, we only need to prove
$\{(\trace_1, t_1 \leq \trace_r, \alpha ),(\trace_l, \alpha \leq \trace_2, t_2)\} \vDash 
\{(\trace_1\land\trace'_r, t_1\leq \trace_2\land\trace'_l, t_2) \}$, where $\trace'_l = (\trace_l\land \trace_r)  - \trace_r$, $\trace'_r =(\trace_l\land \trace_r) - \trace_l$ and $\issatisfied{\trace_l\land\trace_r}$.
Assume that $\theta \vDash  \{(\trace_1, t_1 \leq \trace_r, \alpha ),(\trace_l, \alpha \leq \trace_2, t_2)\} $, that is, 
\begin{equation*}
 (t_1\theta)\cdot\trace_{1} \leq \theta(\alpha)\cdot\trace_{l}
\quad\quad
 \theta(\alpha)\cdot\trace_{r} \leq (t_2\theta)\cdot\trace_{2} 
\end{equation*}
Since $\issatisfied{\trace_l\land\trace_r}$, by Lemmas \ref{app-lem:monotrace} and \ref{app-lem:traceorder}, we have
\begin{equation*}
 (t_1\theta)\cdot\trace_{1} \land\trace'_r \leq \theta(\alpha)\cdot\trace_{l}\land\trace_r
\quad
 \theta(\alpha)\cdot\trace_{l}\land\trace_r \leq (t_2\theta)\cdot\trace_{2} \land\trace'_l
\end{equation*}
which deduces
\begin{equation*}
(t_1\theta)\cdot\trace_{1} \land\trace'_r \leq (t_2\theta)\cdot\trace_{2} \land\trace'_l
\end{equation*}
that is, $\theta \vDash \{(\trace_1\land\trace'_r, t_1\leq \trace_2\land\trace'_l, t_2) \}$.

\item[CM-GLP]
\begin{align*}
&\theta \vDash \{(\trace_i, t_i \leq \trace_{ir}, \alpha)\}_{i\in I} \\
\Longleftrightarrow & \forall i \in I.\;  (~ (t_i\theta)\cdot \trace_i \leq \theta(\alpha) \cdot \trace_{ir}  ~) \\
\Longleftrightarrow & \forall I' \subseteq I.\; \forall \trace \in \phi(I'). \; \forall i \in I'.\; \\
&((t_i\theta)\cdot(\trace_i \land (\trace -\trace_{ir})) \leq \theta(\alpha)\cdot \trace )~\text{(Lemma \ref{app-lem:tracesub1})} \\
\Longleftrightarrow & \forall I' \subseteq I.\; \forall \trace \in \phi(I'). \;  \\
&\quad (~ \sqcup_{i \in I'} (t_i\theta)\cdot(\trace_i \land (\trace -\trace_{ir})) \leq \theta(\alpha)\cdot \trace ~) \\
\Longleftrightarrow &\theta \vDash \{(\epsilon,  t_{I',\trace}^{\sqcup}  \leq \trace, \alpha)\}_{I'\subseteq I, \trace \in \phi(I') }\\
\end{align*}
\item[CM-LUB] Similar to (CM-GLP).
\item[CM-BDS] Similar to (CM-GLP).
\item[CM-SBD] Trivial.
\end{ProofEnumDesc}
\end{proof}

\begin{lemma}\label{app-lem:unifysound}
If $unify(E) = \theta$, then $\theta \vDash E$.
\end{lemma}
\begin{proof}
By induction on $|E|$.
\begin{ProofEnumDesc}
\item[$|E| = 0$] Trivial.
\item[$|E| > 0$] In this case we have $E = \{(\trace_i,\alpha) = t_i\}_{i\in I}::E'$, where $O(\alpha)$ is the greatest. Let $t_\alpha$ be the type constructed  from $(\trace_i,t_i)$ and $E''$ be the equation set obtained by replacing in $E'$ every occurrence of $\alpha$ by $t_{\alpha}$.
Assume $unify(E'') = \theta'$, then $\theta = \theta'\cup\{\alpha\mapsto t_\alpha\}$.
By induction, we have $\theta' \vDash E''$.
By construction, it is clear that $\theta \vDash \{\alpha\mapsto t_\alpha\}$.
Let's consider the constraints $\{(\trace_i,\beta) = s_i\}_{i\in I} \in E'$ of any other variable $\beta$. Then we have $\{(\trace_i,\beta) = s_i\{\alpha\mapsto t_\alpha\}\}_{i\in I} \in E''$.
$$
\begin{array}{rlll}
\theta(\beta)\cdot\trace_i &=&\theta'(\beta)\cdot\trace_i &(\text{Apply } \theta) \\
&=&(s_i\{\alpha\mapsto t_\alpha\})\theta' & (\theta'\vDash E'')\\
& = &s_i(\theta' \cup \{\alpha\mapsto t_\alpha\})& \\
\end{array}
$$
\end{ProofEnumDesc}
\end{proof}

\begin{lemma}\label{app-lem:unifycomplete}
If $\theta \vDash E$, then there exist $\theta'$ and $\theta''$ such that $unify(E) = \theta'$ and $\theta = \theta'\theta''$.
\end{lemma}
\begin{proof}
Conclusion holds trivially when $|E| = 0$]. \\
When $|E| > 0$, we have $E = \{(\trace_i,\alpha) = t_i\}_{i\in I}::E'$, where $O(\alpha)$ is the greatest. 
Let $t_\alpha$ be the type constructed  from $(\trace_i,t_i)$ and $E_0$ be the equation set obtained by replacing in $E'$ every occurrence of $\alpha$ by $t_{\alpha}$.
Since $\sigma \vDash E$, we also have $\sigma\vDash E'$ and thus $\sigma\vDash E_0$.
By induction on $E_0$, there exist $\theta'_0$ such that $unify(E_0) = \theta'_0$ and $\theta = \theta'_0\theta$.
According to unify(), we get $unify(E) = \theta'_0\cup\{\alpha\mapsto t_\alpha\} =\theta'$.
For any $\beta\notin dom(\theta')$, clearly $\beta(\theta'\theta) = \beta\theta$.
For $\alpha$, $\alpha(\theta'\theta) = (t_\alpha)\theta = \alpha\theta$, since $\theta\vDash E$.
While for any other variable $\beta \in dom(\theta')$, we have
$$
\beta(\theta'\theta) = \beta((\theta'_0\cup\{\alpha\mapsto t_\alpha\})\theta)
= \beta(\theta'_0\theta)=\beta\theta
$$
\end{proof}

\end{document}